\documentclass[onecolumn, a4paper, unpublished]{quantumarticle}
\pdfoutput=1

\usepackage[utf8]{inputenc}
\usepackage{amssymb, amsthm}
\usepackage[numbers]{natbib}
\usepackage{hyperref}
\usepackage{authblk}
\usepackage{graphicx}
\usepackage[labelformat=simple]{subcaption}
\usepackage{booktabs}
\usepackage{makecell}
\usepackage{multirow}
\usepackage{pifont}
\usepackage[export]{adjustbox}
\usepackage{tikz}
\usetikzlibrary{quantikz2}
\usetikzlibrary{external}
\usepackage{xcolor}

\newlength{\circuitrowsep}
\newlength{\circuitcolsep}
\theoremstyle{definition}
\newtheorem{definition}{Definition}[section]
\theoremstyle{plain}
\newtheorem{lemma}{Lemma}
\newtheorem{theorem}{Theorem}

\newtheorem{corollary}{Corollary}

\newcommand{\cctrl}[1]{\ctrl[vertical wire=c]{#1}}

\newcommand{\bs}[1]{\boldsymbol #1}
\newsavebox{\weakbox}
\savebox{\weakbox}{%
  \begin{tikzpicture}[baseline=-0.65ex]
    \node[inner sep=0pt] (d) {$\diamondsuit$};
    \draw[overlay, line width=0.4pt, shorten >=0.7pt, shorten <=0.7pt] (d.north) -- (d.south);
    \draw[overlay, line width=0.4pt, shorten >=1pt, shorten <=1pt] (d.west) -- (d.east);
  \end{tikzpicture}%
}
\newcommand{\weakdiamondsuit}{\usebox{\weakbox}}
\newsavebox{\partctrlbox}
\savebox{\partctrlbox}{%
  \begin{tikzpicture}[baseline=-0.65ex]%
    \draw (0,0) circle (0.25em);%
    \fill (0,0) circle (0.125em);%
  \end{tikzpicture}%
}
\newcommand{\partctrl}{\usebox{\partctrlbox}}
\newsavebox{\partdiamondsuitbox}
\savebox{\partdiamondsuitbox}{%
  \begin{tikzpicture}[baseline=-0.65ex]%
    \node[inner sep=0pt, outer sep=0pt] at (0,0) {$\diamondsuit$};%
    \node[inner sep=0pt, outer sep=0pt, scale=0.4] at (0,0) {$\blacklozenge$};%
  \end{tikzpicture}%
}
\newcommand{\partdiamondsuit}{\usebox{\partdiamondsuitbox}}

\tikzset{external/mode=graphics if exists}

\title{Asymptotically Optimal Quantum Circuits for Comparators and Incrementers}

\author[1]{Vivien Vandaele}
\affil[1]{Quantinuum, Terrington House, 13–15 Hills Road, Cambridge CB2 1NL, United Kingdom}

\date{}

\begin{document}

\maketitle

\begin{abstract}
    We present quantum circuits for comparison and increment operations that achieve an asymptotically optimal gate count of $\Theta(n)$ and depth of $\Theta(\log n)$ over the Clifford+Toffoli gate set, while using a provably minimal number of qubits.
    We extend these results to classical--quantum comparators, yielding an improved classical--quantum adder with an optimal qubit count.
    Given the ubiquity of these operations as algorithmic building blocks, our constructions translate directly into reduced circuit complexity for many quantum algorithms.
    As a notable example, they can be used to improve a space-efficient circuit for Shor's factoring algorithm, reducing circuit depth from $\mathcal{O}(n^3)$ to $\mathcal{O}(n^2 \log^2 n)$ without increasing either the qubit count or the asymptotic gate complexity.
    Underpinning these results is a general theorem demonstrating how to trade ancilla qubits for control qubits with low overhead in both depth and gate count, providing a broadly applicable tool for quantum circuit design.
\end{abstract}

\section{Introduction}
Quantum arithmetic operators are fundamental building blocks of many quantum algorithms.
Comparators and incrementers, in particular, appear in a wide variety of key applications, including quantum walk algorithms~\cite{Douglas_2009}, Grover-based minimum finding on unsorted lists~\cite{Durr_1999}, state-of-the-art approximate rotation synthesis algorithms~\cite{Hindlycke_2024}, and modular multiplication constructions for Shor's factoring algorithm~\cite{Shor_1994, Haner_2017}.
Consequently, improving the efficiency of these elementary arithmetic operators is crucial for reducing the cost of a broad range of quantum algorithms and bringing practical quantum computing closer to realization.

Three metrics primarily govern the cost of a quantum circuit: the gate count, the circuit depth, and the number of qubits.
Quantum circuit design often involves trade-offs among these metrics; for instance, one can typically reduce circuit depth at the expense of additional ancilla qubits.
Despite this, we demonstrate that it is possible to achieve simultaneous optimality in all three metrics when implementing a comparator or an incrementer.
Specifically, we present quantum circuits for both operators that match the known lower bounds on gate count and circuit depth while using a provably minimal number of qubits.
Moreover, our constructions use only classical reversible logic gates (Toffoli, CNOT, and NOT) and thus avoid the need for small-angle rotation gates, which typically incur significant rotation synthesis overhead in fault-tolerant quantum computing.

A comparison with prior work is provided in Table~\ref{tab:costs} for the various arithmetic operators considered in this paper.
Our main contributions are summarized below.

\begin{itemize}
\item \textbf{Promise gates.}
We introduce the notion of a \emph{promise gate}, a unitary whose action on a target register is guaranteed only when a designated promise register satisfies a given condition.
We demonstrate that promise gates provide a useful complementary framework for reasoning about conditionally clean ancilla qubits, introduced in~\cite{Nie_2024} and used in several breakthrough results~\cite{Nie_2024, Claudon_2024, Khattar_2025, Remaud_2025}.
In particular, we identify common cases in which implementing a promise gate reduces exactly to implementing a unitary with clean ancilla qubits at one's disposal, making a large body of existing efficient quantum circuit synthesis methods directly and systematically applicable.

\item \textbf{Add controls, save ancillae.}
Building on the formalism of promise gates, we prove a general theorem showing how to trade ancilla qubits for control qubits when implementing a unitary gate, with low overhead in both gate count and circuit depth.
This theorem applies beyond simply adding controls that can be traded for ancilla qubits.
In particular, when an operator decomposes into a multi-controlled unitary and a simpler one (as is the case for comparators and incrementers), the theorem can then be applied directly to the multi-controlled unitary, capturing its full benefit without having to add controls to the entire operator.

\item \textbf{Optimal incrementers.}
We present a quantum circuit for the $n$-bit increment operator with $\Theta(n)$ gates and $\Theta(\log n)$ depth, using a single dirty ancilla.
This matches the lower bounds in all three metrics simultaneously.

\item \textbf{Optimal comparators.}
We present quantum circuits for the $n$-bit quantum--quantum comparator and classical--quantum comparator (comparing a quantum register to a classical constant), with $\Theta(n)$ gates and $\Theta(\log n)$ depth, using no ancilla qubits and a single dirty ancilla, respectively.
Both constructions match the respective lower bounds in all three metrics.

\item \textbf{Improved classical--quantum adder.}
As a direct consequence of our optimal constructions for the increment and comparison operators, we obtain an improved circuit for adding a classical constant to a quantum register, with $\Theta(n \log n)$ gate count and $\Theta(\log^2 n)$ depth, using one dirty ancilla.
When plugged into the construction of Häner et al.~\cite{Haner_2017}, this yields an improved quantum circuit for qubit-efficient implementations of Shor's algorithm.
\end{itemize}

\paragraph{Outline.}
The paper is organized as follows.
Section~\ref{sec:preliminaries} introduces the preliminary notions used throughout the paper.
In Section~\ref{sec:promise_gates}, we introduce promise gates and show how they can be used to design efficient quantum circuits.
We then apply these ideas in Section~\ref{sec:comparator} to construct optimal quantum--quantum and classical--quantum comparator circuits, and in Section~\ref{sec:incrementer} to construct optimal incrementer circuits.
Finally, in Section~\ref{sec:classical_quantum_adder}, we use these results to derive an improved circuit for classical--quantum addition and apply it to the modular multiplication subroutine of Shor's algorithm.

\begin{table}[t]
    \centering
    \renewcommand{\arraystretch}{1.12}
    \resizebox{\textwidth}{!}{%
    \begin{tabular}{llcccc}
        \toprule
        Type & Reference & Ancillae & Depth & Gate Count & Optimal? \\
        \midrule
        \multirow{5}{*}{Incrementer}
        & Gidney~\cite{Gidney_incrementer} & 1 dirty & $\mathcal{O}(n)$ & $\mathcal{O}(n)$ & \ding{55} \\
        & Nie et al.~\cite{Nie_2024} & 1 clean & $\mathcal{O}(\log^2 n)$ & $\mathcal{O}(n)$ & \ding{55} \\
        & Khattar et al.~\cite{Khattar_2025} & $\log_2^* n$ clean & $\mathcal{O}(n)$ & $\mathcal{O}(n)$ & \ding{55} \\
        & Remaud et al.~\cite{Remaud_2025} & 1 dirty & $\mathcal{O}(\log^2 n)$ & $\mathcal{O}(n\log n)$ & \ding{55} \\
        & This paper & 1 dirty & $\mathcal{O}(\log n)$ & $\mathcal{O}(n)$ & \ding{51} \\
        \midrule
        \multirow{4}{*}{\shortstack[l]{Quantum--Quantum \\ Comparator}}
        & Cuccaro et al.~\cite{Cuccaro_2004} & 1 clean & $\mathcal{O}(n)$ & $\mathcal{O}(n)$ & \ding{55} \\
        & Takahashi et al.~\cite{Takahashi_2010} & 0 & $\mathcal{O}(n)$ & $\mathcal{O}(n)$ & \ding{55} \\
        & Remaud et al.~\cite{Remaud_2025} & 0 & $\mathcal{O}(\log^2 n)$ & $\mathcal{O}(n\log n)$ & \ding{55} \\
        & This paper & 0 & $\mathcal{O}(\log n)$ & $\mathcal{O}(n)$ & \ding{51} \\
        \midrule
        \multirow{3}{*}{\shortstack[l]{Classical--Quantum \\ Comparator}}
        & Gidney~\cite{Gidney_2018} & 2 clean & $\mathcal{O}(n\log n)$ & $\mathcal{O}(n\log n)$ & \ding{55} \\
        & Khattar et al.~\cite{Khattar_2025} & $\log_2^* n$ clean & $\mathcal{O}(n)$ & $\mathcal{O}(n)$ & \ding{55} \\
        & This paper & 1 dirty & $\mathcal{O}(\log n)$ & $\mathcal{O}(n)$ & \ding{51} \\
        \midrule
        \multirow{4}{*}{\shortstack[l]{Classical--Quantum \\ Adder}}
        & Häner et al.~\cite{Haner_2017} & 1 dirty & $\mathcal{O}(n)$ & $\mathcal{O}(n\log n)$ & \ding{55} \\
        & Remaud et al.~\cite{Remaud_2025} & 1 dirty & $\mathcal{O}(\log^3 n)$ & $\mathcal{O}(n\log^2 n)$ & \ding{55} \\
        & Gidney~\cite{Gidney_2025} & 3 clean & $\mathcal{O}(n)$ & $\mathcal{O}(n)$ & \ding{55} \\
        & This paper & 1 dirty & $\mathcal{O}(\log^2 n)$ & $\mathcal{O}(n\log n)$ & \ding{55} \\
        \bottomrule
    \end{tabular}}
    \caption{Comparison of prior work with our constructions in terms of ancilla requirements, circuit depth, and gate count.
    Only results using a sublinear number of ancilla qubits and classical reversible gates are reported.
    Results marked with \ding{51} achieve a provably minimal number of ancilla qubits and asymptotically optimal depth and gate count.}\label{tab:costs}
\end{table}

\section{Preliminaries}\label{sec:preliminaries}
\subsection{Multi-controlled \texorpdfstring{$X$}{X} gates and lower bounds}
The gates used in our constructions belong to a subset of the family of $k$-controlled $X$ gates, denoted $C^kX$ and defined as follows.
\begin{definition}
    For a non-negative integer $k$, the $C^kX$ gate acts as
    \begin{equation}
        C^kX \ket{\bs{x}, t} = \lvert \bs{x}, t \oplus \bigwedge_{i=0}^{k-1} x_i \rangle,
    \end{equation}
    where $\bs{x} \in \{0,1\}^k$ and $t \in \{0,1\}$.
\end{definition}

Specifically, the elementary building blocks of our quantum circuits are the $C^kX$ gates with $k \leq 2$:
\begin{itemize}
    \item $C^0X$ corresponds to the $X$ gate, also known as the NOT gate;
    \item $C^1X$ corresponds to the $CX$ gate, also known as the CNOT gate;
    \item $C^2X$ corresponds to the $CCX$ gate, also known as the Toffoli or CCNOT gate.
\end{itemize}
This gate set is a natural choice, as the $CCX$ gate is well known to be the $C^kX$ gate with the fewest controls that is universal for classical reversible computation~\cite{Toffoli_1980}.
While the constructions and optimality results in this paper target the $\{CCX, CX, X\}$ gate set, other gate sets not based on classical reversible logic are also of practical interest, as exemplified by QFT-based arithmetic~\cite{Draper_2000}.
A useful property of the $\{CCX, CX, X\}$ gate set is that any $C^kX$ gate can be efficiently implemented over it using a single dirty ancilla qubit, as stated in the following lemma, proved in~\cite{Nie_2024}.
\begin{lemma}[Nie et al.~\cite{Nie_2024}]\label{lem:optimal_ckx}
    The $C^kX$ gate can be implemented over the $\{CCX, CX, X\}$ gate set with a gate count of $\Theta(k)$ and a circuit depth of $\Theta(\log k)$, using one dirty ancilla qubit.
\end{lemma}
Throughout this paper, a \emph{dirty ancilla qubit} refers to a qubit in an arbitrary unknown state that must be restored to that same state after use, whereas a \emph{clean ancilla qubit} is one initialized in the $\ket{0}$ state and returned to $\ket{0}$ after use.

Lemma~\ref{lem:optimal_ckx} is asymptotically optimal in both gate count and circuit depth and uses a minimal number of ancilla qubits.
The gate count and circuit depth lower bounds of $\Omega(k)$ and $\Omega(\log k)$, respectively, hold for any implementation of the $C^kX$ gate over any set of bounded-size gates, regardless of the number of ancilla qubits~\cite{Fang_2006}.
The ancilla lower bound follows from the fact that, for $k \geq 3$, the $C^kX$ gate induces an odd permutation on $\{0,1\}^{k+1}$, whereas all generators in $\{CCX, CX, X\}$ induce even permutations~\cite{Shende_2003}.

The same lower bounds extend to $k$-bit incrementers and classical--quantum comparators.
Indeed, a $C^kX$ gate can be constructed from a constant number of incrementers as follows:
\begin{equation}\label{eq:ckx_from_inc}
    \begin{quantikz}[row sep={\the\circuitrowsep,between origins}, column sep=\the\circuitcolsep, align equals at=1.5]
        & \qwbundle{k} & & \ctrl{1} & & \\
        & & & \targ{} & &
    \end{quantikz}
    \;=\;
    \begin{quantikz}[row sep={\the\circuitrowsep,between origins}, column sep=\the\circuitcolsep, align equals at=1.5]
        & \qwbundle{k} & & \gate[2]{+1} & \gate{-1} & & \\
        & & & & & &
    \end{quantikz}
\end{equation}
where the decrement gate can be obtained by conjugating the increment gate with $X$ gates.
Similarly, a $C^kX$ gate can be constructed from a single classical--quantum comparator with the constant $2^k - 1$:
\begin{equation}\label{eq:ckx_from_cmp}
    \begin{quantikz}[row sep={\the\circuitrowsep,between origins}, column sep=\the\circuitcolsep, align equals at=1.5]
        & \qwbundle{k} & & \ctrl{1} & & \\
        & & & \targ{} & &
    \end{quantikz}
    \;=\;
    \begin{quantikz}[row sep={\the\circuitrowsep,between origins}, column sep=\the\circuitcolsep, align equals at=1.5]
        & \qwbundle{k} & & \gate{\geq 2^k-1}\wire[d][1]{a} & & \\
        & & & \targ{} & &
    \end{quantikz}
\end{equation}
Therefore, if either of these operators could be implemented with no ancilla qubits, or with lower asymptotic gate count or depth complexities, then so could the $C^kX$ gate, contradicting the established lower bounds.
The quantum--quantum comparator, however, can be implemented without ancilla qubits.
Nevertheless, the gate count and circuit depth lower bounds of $\Omega(k)$ and $\Omega(\log k)$ still apply, since a quantum--quantum comparator can be used to construct a $k$-bit classical--quantum comparator using ancilla qubits, which in turn can be used to construct a $C^kX$ gate, as shown by the circuit identity in Equation~\eqref{eq:ckx_from_cmp}.

It is worth noting that the $\Omega(k)$ lower bound also applies specifically to the number of $CCX$ gates, not merely to the total gate count.
Indeed, Beverland et al.~\cite{Beverland_2020} showed that the $CCX$ count required to implement the $C^kX$ gate over the $\{CCX, CX, X\}$ gate set is $\Omega(k)$, regardless of the number of ancilla qubits.
Together with the reductions in Equations~\eqref{eq:ckx_from_inc} and~\eqref{eq:ckx_from_cmp}, this implies an $\Omega(k)$ lower bound on the $CCX$ count for any exact implementation of a $k$-bit incrementer or comparator.
Moreover, this lower bound on the $CCX$ count also holds for unitary circuits approximating the $C^kX$ gate, or $k$-bit incrementer or comparator operators~\cite{Gosset_2025}.
However, for $\epsilon$-approximate implementation by mixed circuits, the $C^kX$ gate requires only $\Theta(\log(1/\epsilon))$ $CCX$ gates, whereas the lower bound for incrementers and comparators remains $\Omega(k)$ $CCX$ gates~\cite{Gosset_2025}.

In this work, we present quantum circuits for incrementers and comparators that match these lower bounds.
Previously, these operators were considered more costly to implement than the $C^kX$ gate, notably because a constant number of them suffice to construct the $C^kX$ gate, as in Equations~\eqref{eq:ckx_from_inc} and~\eqref{eq:ckx_from_cmp}, whereas the converse does not hold.
Despite this asymmetry, we show that these operators can in fact be implemented with the same asymptotically optimal gate count and circuit depth, using a minimal number of ancilla qubits.

\subsection{Structural primitives}
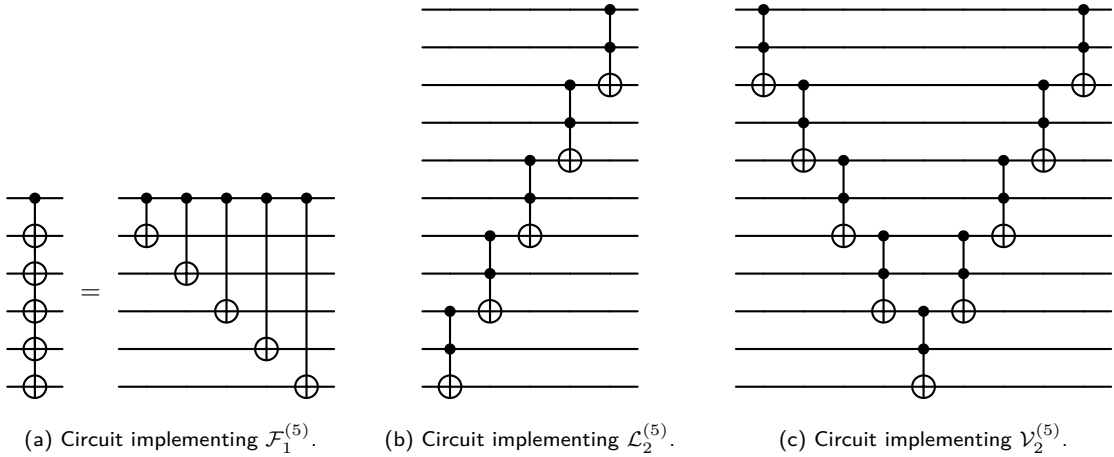
\begin{figure}[t]
    \begin{subfigure}[b]{0.31\textwidth}
        \centering
        \begin{quantikz}[row sep={.5cm,between origins}, column sep=0.2cm, align equals at=3.5]
            & \ctrl{5} & \\
            & \targ{} & \\
            & \targ{} & \\
            & \targ{} & \\
            & \targ{} & \\
            & \targ{} &
        \end{quantikz}
        =
        \begin{quantikz}[row sep={.5cm,between origins}, column sep=0.2cm, align equals at=3.5]
            & \ctrl{1} & \ctrl{2} & \ctrl{3} & \ctrl{4} & \ctrl{5} & \\
            & \targ{} & & & & & \\
            & & \targ{} & & & & \\
            & & & \targ{} & & & \\
            & & & & \targ{} & & \\
            & & & & & \targ{} &
        \end{quantikz}
        \caption{Circuit implementing $\mathcal{F}_1^{(5)}$.}\label{fig:fanout}
    \end{subfigure}
    \begin{subfigure}[b]{0.31\textwidth}
        \centering
        \begin{quantikz}[row sep={.5cm,between origins}, column sep=0.2cm]
            & & & & & \ctrl{1} & \\
            & & & & & \ctrl{1} & \\
            & & & & \ctrl{1} & \targ{} & \\
            & & & & \ctrl{1} & & \\
            & & & \ctrl{1} & \targ{} & & \\
            & & & \ctrl{1} & & & \\
            & & \ctrl{1} & \targ{} & & & \\
            & & \ctrl{1} & & & & \\
            & \ctrl{1} & \targ{} & & & & \\
            & \ctrl{1} & & & & & \\
            & \targ{} & & & & &
        \end{quantikz}
        \caption{Circuit implementing $\mathcal{L}_2^{(5)}$.}\label{fig:ladder}
    \end{subfigure}
    \begin{subfigure}[b]{0.37\textwidth}
        \centering
        \begin{quantikz}[row sep={.5cm,between origins}, column sep=0.2cm]
            & \ctrl{1} & & & & & & & & \ctrl{1} & \\
            & \ctrl{1} & & & & & & & & \ctrl{1} & \\
            & \targ{} & \ctrl{1} & & & & & & \ctrl{1} & \targ{} & \\
            & & \ctrl{1} & & & & & & \ctrl{1} & & \\
            & & \targ{} & \ctrl{1} & & & & \ctrl{1} & \targ{} & & \\
            & & & \ctrl{1} & & & & \ctrl{1} & & & \\
            & & & \targ{} & \ctrl{1} & & \ctrl{1} & \targ{} & & & \\
            & & & & \ctrl{1} & & \ctrl{1} & & & & \\
            & & & & \targ{} & \ctrl{1} & \targ{} & & & & \\
            & & & & & \ctrl{1} & & & & & \\
            & & & & & \targ{} & & & & &
        \end{quantikz}
        \caption{Circuit implementing $\mathcal{V}_2^{(5)}$.}\label{fig:v_ladder}
    \end{subfigure}
    \caption{Naive implementations of the $\mathcal{F}_1^{(5)}$, $\mathcal{L}_2^{(5)}$, and $\mathcal{V}_2^{(5)}$ operators.}\label{fig:l_v_examples}
\end{figure}

Many quantum arithmetic circuits rely on a few structural primitives that account for most of the gate count and depth complexity of the circuit.
Here, we give the definitions and costs of the primitives used throughout the paper.
Each of these operators admits a straightforward implementation with optimized gate count but suboptimal circuit depth.
The lemmas below show that both gate count and circuit depth can often be simultaneously optimized, in some cases at the expense of ancilla qubits.

The first primitive is the $k$-th order fan-out operator: a collection of parallel $C^kX$ gates all controlled by an additional single qubit.
\begin{definition} The $k$-th order fan-out operator $\mathcal{F}_k^{(n)}$ acts on $n(k+1)+1$ qubits as
    \begin{equation}
        \mathcal{F}_k^{(n)} \left(\ket{c} \otimes \bigotimes_{i=1}^{n} \ket{\bs{x}_i, t_i} \right) = \ket{c} \otimes \bigotimes_{i=1}^{n} \lvert \bs{x}_i, t_i \oplus c \bigwedge_{j=0}^{k-1} x_{i,j} \rangle,
    \end{equation}
    where $c \in \{0,1\}$, $\bs{x}_i \in \{0,1\}^{k}$, and $t_i \in \{0,1\}$.
\end{definition}
A quantum circuit representing the $\mathcal{F}_1^{(5)}$ operator is shown in Figure~\ref{fig:fanout}.
We mainly use the first-order and second-order fan-out operators, $\mathcal{F}_1^{(n)}$ and $\mathcal{F}_2^{(n)}$.
It is well known that these operators can be efficiently implemented, as established by the following lemma~\cite{Fang_2006,Broadbent_2009, Remaud_2025}.
\begin{lemma}\label{lem:fanout}
    The $\mathcal{F}_1^{(n)}$ and $\mathcal{F}_2^{(n)}$ operators can both be implemented with $\mathcal{O}(n)$ $\{CCX, CX\}$ gates and $\mathcal{O}(\log n)$ circuit depth.
\end{lemma}

Another recurring primitive, denoted $\mathcal{L}_k^{(n)}$, is the operator that can be implemented by a ladder of $n$ consecutive $C^kX$ gates.
\begin{definition}
    The $\mathcal{L}_k^{(n)}$ operator acts on $kn+1$ qubits as
    \begin{equation}
        \mathcal{L}_k^{(n)}\left( \bigotimes_{i=0}^{kn} \ket{x_i} \right) = \ket{x_0} \otimes \bigotimes_{i=1}^n \left( \bigotimes_{j=1}^{k-1} \ket{x_{k(i-1)+j}} \otimes \lvert x_{ki} \oplus \prod_{\ell=1}^{k} x_{ki-\ell} \rangle \right),
    \end{equation}
    where $\bs{x} \in \{0, 1\}^{kn+1}$.
\end{definition}
A naive implementation of $\mathcal{L}_2^{(5)}$ is shown in Figure~\ref{fig:ladder}.
The operator $\mathcal{L}_1^{(n)}$ can be implemented similarly but with a ladder of $CX$ gates instead of $CCX$ gates.
It was proved in~\cite{Remaud_2025} that $\mathcal{L}_1^{(n)}$ can be implemented with the same gate count and depth as the $\mathcal{F}_1^{(n)}$ operator, up to a constant factor.
\begin{lemma}\label{lem:ladder_cnot}
    The $\mathcal{L}_1^{(n)}$ operator can be implemented with $\mathcal{O}(n)$ $CX$ gates and $\mathcal{O}(\log n)$ circuit depth.
\end{lemma}

By using a linear number of ancilla qubits, the $\mathcal{L}_2^{(n)}$ operator can be implemented with the same asymptotic gate count and circuit depth complexities as the $\mathcal{L}_1^{(n)}$ operator, as stated by the following lemma.
\begin{lemma}\label{lem:ladder_2_n_ancilla}
    The $\mathcal{L}_2^{(n)}$ operator can be implemented with $\mathcal{O}(n)$ $CCX$ gates and $\mathcal{O}(\log n)$ circuit depth, using $n$ ancilla qubits.
\end{lemma}
The proof of Lemma~\ref{lem:ladder_2_n_ancilla} is provided in Appendix~\ref{app:proof_ladder}.
The implementation of the $\mathcal{L}_2^{(n)}$ operator with $\mathcal{O}(n)$ gates, $\mathcal{O}(\log n)$ depth, and $\mathcal{O}(n)$ ancilla qubits was used in prior work on quantum addition circuits~\cite{Draper_2006, Takahashi_2010}.
The connection between the construction used in these works and the ancilla-free implementation of $\mathcal{L}_2^{(n)}$ was established in~\cite{remaud2025quantumaddersstructurallink}.

The following corollary follows from Lemmas~\ref{lem:ladder_2_n_ancilla} and~\ref{lem:optimal_ckx}:
\begin{corollary}\label{cor:ckx_ladder}
    The $\mathcal{L}_k^{(n)}$ operator can be implemented with $\mathcal{O}(kn)$ $\{CCX$, $CX$, $X\}$ gates and $\mathcal{O}(\log(kn))$ circuit depth, using $n$ ancilla qubits.
\end{corollary}
The proof of Corollary~\ref{cor:ckx_ladder} is provided in Appendix~\ref{app:proof_ckx_ladder}.

Finally, the last structural primitive used in this work, denoted $\mathcal{V}_k^{(n)}$, is defined in terms of $\mathcal{L}_k^{(n)}$ and its adjoint.
\begin{definition}
    The $\mathcal{V}_k^{(n)}$ operator, acting on $kn+1$ qubits, is defined by
    \begin{equation}
        \mathcal{V}_k^{(n)} = \mathcal{L}_k^{(n)} \left( {\mathcal{L}_k^{(n-1)}}^\dagger \otimes I \right).
    \end{equation}
\end{definition}
A naive construction of $\mathcal{V}_k^{(n)}$ is obtained by concatenating a ladder of $C^kX$ gates implementing ${\mathcal{L}_k^{(n-1)}}^\dagger$ with a ladder of $C^kX$ gates implementing $\mathcal{L}_k^{(n)}$, yielding a V-shaped circuit of $C^kX$ gates.
An example is shown in Figure~\ref{fig:v_ladder} for $\mathcal{V}_2^{(5)}$.
While the $\mathcal{V}_k^{(n)}$ operator and its implementation cost have not been explicitly studied in prior work, we show that it plays a central role in the implementation of the comparison operator.

\section{Promise gates}\label{sec:promise_gates}
We introduce the notion of a \emph{promise gate}, a term inspired by the concept of a \emph{promise problem} in computational complexity theory.\footnote{In computational complexity theory, a \emph{promise problem} is a generalization of a decision problem in which the input is promised to belong to a subset of all possible inputs. Any algorithm solving the promise problem is required to answer correctly only on promised inputs; the algorithm's behavior on inputs that violate the promise is left unspecified~\cite{Goldreich_2006}.}
We first define promise gates in Subsection~\ref{sub:promise_gate_def}.
Then, in Subsection~\ref{sub:add_controls_save_ancillae}, we present a theorem demonstrating how promise gates can be used to efficiently trade clean ancilla qubits for control qubits when implementing a unitary gate.
We present a simple and concrete application of this theorem to controlled addition in Subsection~\ref{sub:controlled_addition}.

\subsection{Definition}\label{sub:promise_gate_def}
A promise gate is a unitary acting on two designated registers: a \emph{promise register} and a \emph{target register}.
It is only required to implement a target unitary $U$ on the target register when the promise register satisfies a specified condition; its behavior when the condition is violated is left unspecified.
We distinguish two variants: a \emph{weak} promise gate may act arbitrarily on the promise register when the promise is violated, whereas a \emph{strong} promise gate must preserve it for all inputs.
In this work, we focus on the case in which the promise condition is that the promise register is in the $\ket{0}$ state, as formalized in the following definitions.
\begin{definition}\label{def:weak_promise_gate}
    A \emph{weak promise gate} with target unitary $U$ is a unitary $\tilde{U}$ acting on a $k$-qubit promise register and an $n$-qubit target register such that
    $$
        \tilde{U} (\ket{0}^{\otimes k} \otimes \ket{\phi}) = \ket{0}^{\otimes k} \otimes U\ket{\phi}
    $$
    for every $n$-qubit state $\ket{\phi}$.
    The action of $\tilde{U}$ on inputs whose promise register is not in the state $\ket{0}^{\otimes k}$ is unconstrained.
\end{definition}

\begin{definition}\label{def:strong_promise_gate}
    A \emph{strong promise gate} with target unitary $U$ is a unitary $\hat{U}$ acting on a $k$-qubit promise register and an $n$-qubit target register of the form\footnote{A unitary of this form is known as a \emph{quantum multiplexer} (QMUX)~\cite{Shende_2006}. A strong promise gate is thus a QMUX in which only $U_{\bs{0}}$ is specified.}
    $$
        \hat{U} = \sum_{\bs{j} \in \{0,1\}^k} \ket{\bs{j}}\bra{\bs{j}} \otimes U_{\bs{j}},
    $$
    where each $U_{\bs{j}}$ is a unitary acting on the target register with $U_{\bs{0}} = U$.
    The unitaries $U_{\bs{j}}$ for $\bs{j} \neq \bs{0}$ are unconstrained.
\end{definition}

Any strong promise gate with target unitary $U$ is also a weak promise gate with target unitary $U$: it satisfies the same guarantee on the promised subspace, but additionally preserves the promise register for all inputs.
In contrast, a weak promise gate may act nontrivially on the promise register when the promise is violated.

We denote weak and strong promise gates by the following circuit notations, respectively:
\begin{align}
    \label{eq:weak_promise_gate}
    \begin{quantikz}[row sep={\the\circuitrowsep,between origins}, column sep=\the\circuitcolsep, align equals at=1.5]
        & \qwbundle{k} & \push{\weakdiamondsuit}\wire[d][1]{a} & \\
        & \qwbundle{n} & \gate{U} &
    \end{quantikz}
    \;&:=\;
    \begin{quantikz}[row sep={\the\circuitrowsep,between origins}, column sep=\the\circuitcolsep, align equals at=1.5]
        & \qwbundle{k} & \gate[2]{\tilde{U}} & \\
        & \qwbundle{n} & &
    \end{quantikz}
    \qquad\qquad
    \text{(weak promise gate)}
    \\
    \label{eq:strong_promise_gate}
    \begin{quantikz}[row sep={\the\circuitrowsep,between origins}, column sep=\the\circuitcolsep, align equals at=1.5]
        & \qwbundle{k} & \push{\diamondsuit}\wire[d][1]{a} & \\
        & \qwbundle{n} & \gate{U} &
    \end{quantikz}
    \;&:=\;
    \begin{quantikz}[row sep={\the\circuitrowsep,between origins}, column sep=\the\circuitcolsep, align equals at=1.5]
        & \qwbundle{k} & \gate[2]{\hat{U}} & \\
        & \qwbundle{n} & &
    \end{quantikz}
    \qquad\qquad
    \text{(strong promise gate)}
\end{align}
where $\weakdiamondsuit$ and $\diamondsuit$ mark the promise register.
When the promise register is in the $\ket{0}^{\otimes k}$ state, the promise gate acts as $U$ on the target register:
\begin{equation}\label{eq:promise_gate_action}
    \begin{quantikz}[row sep={\the\circuitrowsep,between origins}, column sep=\the\circuitcolsep, align equals at=1.5]
        \lstick{$\ket{0}$} & \qwbundle{k} & \push{\weakdiamondsuit}\wire[d][1]{a} & \rstick{$\ket{0}$} \\
        & \qwbundle{n} & \gate{U} &
    \end{quantikz}
    \;=\;
    \begin{quantikz}[row sep={\the\circuitrowsep,between origins}, column sep=\the\circuitcolsep, align equals at=1.5]
        \lstick{$\ket{0}$} & \qwbundle{k} & \push{\diamondsuit}\wire[d][1]{a} & \rstick{$\ket{0}$} \\
        & \qwbundle{n} & \gate{U} &
    \end{quantikz}
    \;=\;
    \begin{quantikz}[row sep={\the\circuitrowsep,between origins}, column sep=\the\circuitcolsep, align equals at=1.5]
        \lstick{$\ket{0}$} & \qwbundle{k} & \qw & \rstick{$\ket{0}$}\\
        & \qwbundle{n} & \gate{U} &
    \end{quantikz}
\end{equation}

The left-hand side of~\eqref{eq:weak_promise_gate} and~\eqref{eq:strong_promise_gate} uniquely identifies the promise and target registers together with the target unitary $U$, but does not uniquely determine the promise gate.
Indeed, distinct unitaries may satisfy Definition~\ref{def:weak_promise_gate} or~\ref{def:strong_promise_gate} while differing in their action outside the promised subspace.
To ensure that circuit identities remain well defined, we only use equalities whose validity is independent of how the promise gate acts outside the promised subspace.
For instance:
\begin{equation}
    \begin{quantikz}[row sep={\the\circuitrowsep,between origins}, column sep=\the\circuitcolsep, align equals at=1.5]
        \lstick{$\ket{0}$} & \qwbundle{k} & \push{\diamondsuit}\wire[d][1]{a} & & \rstick{$\ket{0}$} \\
        & \qwbundle{n} & \gate{U} & \gate{U} & 
    \end{quantikz}
    \;=\;
    \begin{quantikz}[row sep={\the\circuitrowsep,between origins}, column sep=\the\circuitcolsep, align equals at=1.5]
        \lstick{$\ket{0}$} & \qwbundle{k} & \gate[2]{\hat{U}} & & \rstick{$\ket{0}$} \\
        & \qwbundle{n} & & \gate{U} & 
    \end{quantikz}
    \;=\;
    \begin{quantikz}[row sep={\the\circuitrowsep,between origins}, column sep=\the\circuitcolsep, align equals at=1.5]
        \lstick{$\ket{0}$} & \qwbundle{k} & & \rstick{$\ket{0}$}\\
        & \qwbundle{n} & &
    \end{quantikz}
\end{equation}
whenever $U^2 = I$, regardless of which promise gate $\hat{U}$ with target unitary $U$ is chosen.
Moreover, we adopt the convention that, within any given circuit, all promise gates sharing the same target unitary $U$ and the same register sizes refer to the same underlying unitary $\tilde{U}$ or $\hat{U}$, and all promise gates sharing the same target unitary $U^\dag$ and the same register sizes refer to the adjoint $\tilde{U}^\dag$ or $\hat{U}^\dag$, so that $\tilde{U}\tilde{U}^\dag = I$ or $\hat{U}\hat{U}^\dag = I$.
This convention enables circuit simplifications such as:
\begin{equation}\label{eq:promise_gate_inverse}
    \begin{quantikz}[row sep={\the\circuitrowsep,between origins}, column sep=\the\circuitcolsep, align equals at=1.5]
        & \qwbundle{k} & \push{\diamondsuit}\wire[d][1]{a} & \push{\diamondsuit}\wire[d][1]{a} & \\
        & \qwbundle{n} & \gate{U} & \gate{U^\dag} &
    \end{quantikz}
    \;=\;
    \begin{quantikz}[row sep={\the\circuitrowsep,between origins}, column sep=\the\circuitcolsep, align equals at=1.5]
        & \qwbundle{k} & \gate[2]{\hat{U}} & \gate[2]{\hat{U}^\dag} & \\
        & \qwbundle{n} & & &
    \end{quantikz}
    \;=\;
    \begin{quantikz}[row sep={\the\circuitrowsep,between origins}, column sep=\the\circuitcolsep, align equals at=1.5]
        & \qwbundle{k} & \qw & \\
        & \qwbundle{n} & \qw &
    \end{quantikz}
\end{equation}
This convention is particularly natural for the $\{CCX, CX, X\}$ gate set used in this paper, since each generator is self-adjoint.
Consequently, constructing a promise gate $\tilde{U}$ or $\hat{U}$ for a target unitary $U$ over this gate set automatically yields a construction for $\tilde{U}^\dag$ or $\hat{U}^\dag$, with the same gate count and circuit depth, simply by reversing the circuit.

Note that a controlled promise gate is not the same as a promise gate whose target unitary is a controlled gate.
A \emph{controlled promise gate} ($C\tilde{U}$ or $C\hat{U}$) is obtained by adding a control qubit to a promise gate.
When the control qubit is in the state $\ket{0}$, the gate acts as the identity on all registers; when it is in the state $\ket{1}$, the promise gate is applied.
For a strong promise gate, this preserves the promise register regardless of the state of the control qubit; for a weak promise gate, the promise register may be modified when the control qubit is in the $\ket{1}$ state and the promise is violated.
On the other hand, a promise gate ($\widetilde{CU}$ or $\widehat{CU}$) whose target unitary is the controlled-$U$ gate $CU$ applies a controlled-$U$ operation to the target register when the promise is satisfied, but its action on the target register when the promise is violated is left unspecified, even when the control qubit is in the state $\ket{0}$.
In this work, we use the following quantum circuit notation to represent controlled promise gates:
\begin{equation}\label{eq:controlled_promise_gate}
    \begin{quantikz}[row sep={\the\circuitrowsep,between origins}, column sep=\the\circuitcolsep, align equals at=2]
        & \qwbundle{c} & \ctrl{0}\wire[d][1]{a} & \\
        & \qwbundle{k} & \push{\diamondsuit}\wire[d][1]{a} & \\
        & \qwbundle{n} & \gate{U} &
    \end{quantikz}
    \;:=\;
    \begin{quantikz}[row sep={\the\circuitrowsep,between origins}, column sep=\the\circuitcolsep, align equals at=2, execute at end picture={
            \draw [line width=0.8pt] (topcontrol.south) -- (ccxbox.north -| topcontrol.south);
        }]
        & \qwbundle{c} & & |[alias=topcontrol, fill, circle, minimum size=4pt, inner sep=0pt]|{} & & \\
        & \qwbundle{k} & & \push{\diamondsuit}\wire[d][1]{a}\gategroup[2, steps=1, style={name=ccxbox, inner sep=1pt}]{} & & \\
        & \qwbundle{n} & & \gate{U} & &
    \end{quantikz}
    \;=\;
    \begin{quantikz}[row sep={\the\circuitrowsep,between origins}, column sep=\the\circuitcolsep, align equals at=2]
        & \qwbundle{c} & \ctrl{1} & \\
        & \qwbundle{k} & \gate[2]{\hat{U}} & \\
        & \qwbundle{n} & &
    \end{quantikz}
\end{equation}
and similarly with $\weakdiamondsuit$ and $\tilde{U}$ for the weak variant.

The notion of a promise gate is related to the concept of conditionally clean ancilla qubits, introduced in~\cite{Nie_2024} and used in several breakthrough results~\cite{Nie_2024, Claudon_2024, Khattar_2025, Remaud_2025}.
A conditionally clean ancilla is a qubit that is guaranteed to be in a known classical state whenever a certain predicate on other qubits holds.
Conditionally clean ancilla qubits characterize the state of workspace qubits: they assert that a qubit is clean under a given condition.
A promise gate captures a complementary behavioral abstraction: it prescribes what a gate is required to do on a subspace, leaving its action on the orthogonal complement unspecified (for weak promise gates) or constrained to preserve the promise register (for strong promise gates).
Conditionally clean ancilla qubits can be viewed as the promise register of a promise gate: the predicate guaranteeing their cleanliness is precisely the condition under which the promise is satisfied.

This shift in perspective offers a key advantage: for weak promise gates, designing a quantum circuit $\mathcal{C}_U$ that implements a unitary $U$ using $m$ clean ancilla qubits is exactly the same problem as designing a quantum circuit $\mathcal{C}_{\tilde{U}}$ that implements a weak promise gate $\tilde{U}$ with target unitary $U$ and a promise register of size~$m$.
To see this, note that if $\mathcal{C}_{\tilde{U}}$ implements $U$ whenever the promise register is in the state $\ket{0}$, then it also implements $U$ when the promise register is replaced by clean ancilla qubits.
Conversely, if $\mathcal{C}_U$ implements $U$ with clean ancilla qubits, then $\mathcal{C}_U$ implements a weak promise gate $\tilde{U}$ when the clean ancilla register is reinterpreted as a promise register.
This equivalence can simplify the design of some quantum circuits that exploit conditionally clean ancilla qubits, as constructing a circuit for a weak promise gate $\tilde{U}$ reduces to the familiar task of constructing a circuit for $U$ with some clean ancilla qubits, a setting for which a large body of techniques already exists.
This is precisely what makes promise gates powerful in practice: having access to clean ancilla qubits is generally very beneficial for reducing circuit depth or gate count, so the additional flexibility afforded by the promise register can yield implementations of $\tilde{U}$ that are significantly cheaper than an implementation of $U$ without ancilla qubits.

For strong promise gates, the equivalence also holds, with the additional condition that the circuit $\mathcal{C}_U$ preserves the ancilla register for all inputs.
This additional condition is naturally satisfied by many practical constructions.
In particular, circuits over $\{CCX, CX, X\}$ that use clean ancilla qubits typically employ them in a compute-uncompute pattern.
In such circuits, the restoration of the ancilla register does not depend on the ancilla qubits being initialized to $\ket{0}$: the uncomputation reverses the computation for every computational basis state of the ancilla register, so the register is preserved regardless of its initial state.
For example, this is the case for the circuit construction for the $\mathcal{L}_2^{(n)}$ operator in Equation~\eqref{eq:l2_first_iteration}.

Moreover, promise gates compose naturally: if two promise gates $\tilde{U}$ and $\tilde{V}$ share the same promise register, their sequential composition $\tilde{V}\tilde{U}$ is itself a promise gate.
More generally, the formalism makes it straightforward to reason modularly about large circuits that use promise gates, since each gate's correctness obligation is confined to the promised subspace.

Promise gates are one of the main ingredients for achieving the results presented in this paper.
In the next subsection, we present a theorem that uses promise gates to trade clean ancilla qubits for controls when implementing a unitary.
We then apply this theorem throughout the paper to construct asymptotically optimal quantum circuits for common arithmetic operators.

\subsection{Add controls, save ancillae}\label{sub:add_controls_save_ancillae}
To make effective use of promise gates, one must address the case in which the promise is not satisfied.
Apart from trivial cases, such as when the promise register consists of clean ancilla qubits, we identify two main techniques for harnessing the power of promise gates.
The circuit identities underlying these techniques are presented in Figures~\ref{fig:trading_controls} and~\ref{fig:controlled_conjugation_promise}.

\begin{figure}[t]
    \begin{subfigure}[b]{\textwidth}
        \centering
        \begin{quantikz}[row sep={\the\circuitrowsep,between origins}, column sep=\the\circuitcolsep, align equals at=2]
            & \qwbundle{k} & \ctrl{1} & \\
            & \qwbundle{n} & \gate{U} &
        \end{quantikz}
        =
        \begin{quantikz}[row sep={\the\circuitrowsep,between origins}, column sep=\the\circuitcolsep, align equals at=2]
            & \qwbundle{k} & \ctrl{2} & & \ctrl{2} & \\
            & \qwbundle{n} & & \gate{U} & & \\
            \lstick{\ket{0}} & & \targ{} & \ctrl{-1} & \targ{} & \rstick{\ket{0}}
        \end{quantikz}
        =
        \begin{quantikz}[row sep={\the\circuitrowsep,between origins}, column sep=\the\circuitcolsep, align equals at=2]
            & \qwbundle{k} & \ctrl{2} & & \ctrl{2} & & \\
            & \qwbundle{n} & & \gate{U} & & \gate{U} & \\
            \lstick{\ket{\psi}} & & \targ{} & \ctrl{-1} & \targ{} & \ctrl{-1} & \rstick{\ket{\psi}}
        \end{quantikz}
        \caption{A $C^kU$ gate can be decomposed into two $C^kX$ gates and one controlled-$U$ gate by using a clean ancilla.
        Whenever $U^2 = I$, the clean ancilla can be replaced by a dirty ancilla at the cost of an additional controlled-$U$ gate.}\label{fig:toggle_detection}
    \end{subfigure}

    \vspace{1em}

    \begin{subfigure}[b]{\textwidth}
        \centering
        \begin{quantikz}[row sep={\the\circuitrowsep,between origins}, column sep=.2cm, align equals at=2]
            & \qwbundle{k} & \ctrl{1} & \\
            & \qwbundle{n} & \gate{U} &
        \end{quantikz}
        =\!\!
        \begin{quantikz}[row sep={\the\circuitrowsep,between origins}, column sep=.2cm, align equals at=2]
            & \qwbundle{k} & \ctrl{2} & \gate{X} & \push{\weakdiamondsuit}\wire[d][1]{a} & \gate{X} & \ctrl{2} & \\
            & \qwbundle{n} & & & \gate{U} & & & \\
            \lstick{\ket{0}} & & \targ{} & & \ctrl{-1} & & \targ{} & \rstick{\ket{0}}
        \end{quantikz}
        \!\!=\!\!
        \begin{quantikz}[row sep={\the\circuitrowsep,between origins}, column sep=.2cm, align equals at=2]
            & \qwbundle{k} & \ctrl{2} & \gate{X} & \push{\diamondsuit}\wire[d][1]{a} & \gate{X} & \ctrl{2} & \gate{X} & \push{\diamondsuit}\wire[d][1]{a} & \gate{X} & \\
            & \qwbundle{n} & & & \gate{U} & & & & \gate{U^\dag} & & \\
            \lstick{\ket{\psi}} & & \targ{} & & \ctrl{-1} & & \targ{} & & \ctrl{-1} & & \rstick{\ket{\psi}}
        \end{quantikz}
        \caption{The controlled-$U$ gate of~\subref{fig:toggle_detection} can be replaced by a controlled promise gate with target unitary $U$, where the top $k$-qubit register acts as the promise register.
        As in~\subref{fig:toggle_detection}, replacing the clean ancilla with a dirty ancilla requires $U^2 = I$.
        In this case, strong promise gates must be used instead of weak promise gates to preserve the state of the $k$-qubit register.
        Note that the two strong promise gates associated with $U$ and $U^\dag$ on the right-hand side satisfy $\hat{U}\hat{U}^\dag = I$, following the convention adopted in Section~\ref{sub:promise_gate_def}.
        }\label{fig:conditionally_clean_ancillae}
    \end{subfigure}
    \caption{Decomposing a multi-controlled $U$ gate to take advantage of promise gates.}\label{fig:trading_controls}
\end{figure}
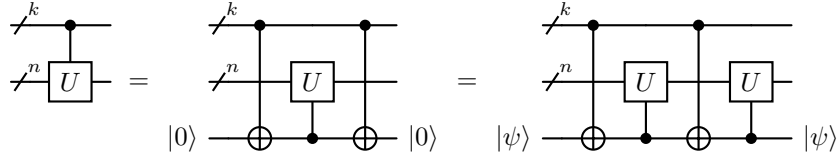
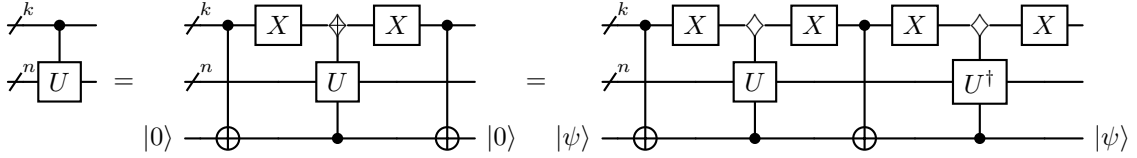

The first technique, well known from prior work on conditionally clean ancilla qubits, arises when the promise is satisfied if and only if the control qubit of a controlled promise gate is in the state $\ket{1}$.
The first equality in Figure~\ref{fig:toggle_detection} shows that a $C^kU$ gate can be implemented by introducing an ancilla qubit that detects whether the $k$-qubit control register is in the all-ones state.
The controlled-$U$ gate then acts nontrivially only when the top $k$-qubit register is in the $\ket{1}^{\otimes k}$ state, which the $X$ gates map to $\ket{0}^{\otimes k}$.
This means that the top register can serve as clean ancilla qubits for implementing $U$ whenever the control qubit is in the $\ket{1}$ state, and can therefore be used as the promise register of a controlled weak promise gate with target unitary $U$, as shown in the first equality of Figure~\ref{fig:conditionally_clean_ancillae}.

Moreover, whenever $U^2=I$, the second equality in Figure~\ref{fig:toggle_detection} allows the clean ancilla to be replaced by a dirty ancilla.
This leads to the second equality in Figure~\ref{fig:conditionally_clean_ancillae}, where strong promise gates are required instead of weak promise gates in order to preserve the state of the top $k$-qubit register.

Another effective use of promise gates arises when a promise gate is canceled by its adjoint whenever the promise is not satisfied.
Consider the well-known circuit identity for controlled conjugations shown in Figure~\ref{fig:controlled_conjugation}: when adding controls to $V^\dag UV$, only $U$ needs to be controlled, not $V$ and $V^\dag$.
Combining this identity with the first equality in Figure~\ref{fig:conditionally_clean_ancillae} yields the circuit identity in Figure~\ref{fig:promise_controlled_conjugation}.
Whenever the promise is not satisfied, the controlled promise gate with target unitary $U$ acts as the identity, and the weak promise gates $\tilde{V}$ and $\tilde{V}^\dag$ associated with $V$ and $V^\dag$ compose to the identity: $\tilde{V}^\dag \tilde{V} = I$.

Figure~\ref{fig:promise_controlled_conjugation} is therefore a valuable generalization of the controlled conjugation circuit identity of Figure~\ref{fig:controlled_conjugation}: not only are the controls on $V$ and $V^\dag$ removed, but $V$ and $V^\dag$ are replaced by weak promise gates $\tilde{V}$ and $\tilde{V}^\dag$ whose promise register can be used as if its qubits were clean ancilla qubits.
This can lead to substantially cheaper implementations of $\tilde{V}$ and $\tilde{V}^\dag$ compared with implementing $V$ and $V^\dag$ without ancilla qubits.
The clean ancilla in Figure~\ref{fig:promise_controlled_conjugation} can also be replaced by a dirty ancilla by using the second equality in Figure~\ref{fig:conditionally_clean_ancillae} instead of the first, yielding the circuit identity in Figure~\ref{fig:promise_controlled_conjugation_dirty}, which requires $U^2=I$.

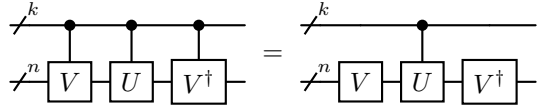
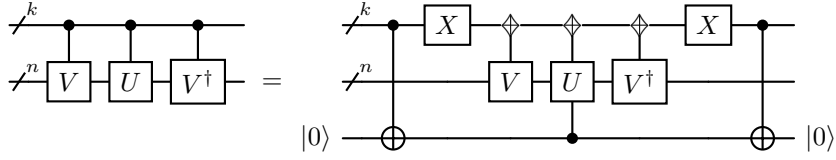
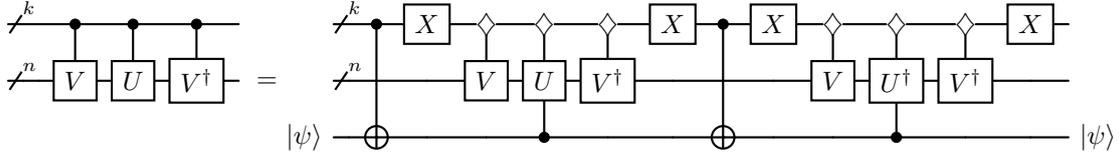
\begin{figure}[t]
    \begin{subfigure}[b]{\textwidth}
        \centering
        \begin{quantikz}[row sep={\the\circuitrowsep,between origins}, column sep=\the\circuitcolsep, align equals at=1.5]
            & \qwbundle{k} & \ctrl{1} & \ctrl{1} & \ctrl{1} & \\
            & \qwbundle{n} & \gate{V} & \gate{U} & \gate{V^\dag} &
        \end{quantikz}
        =
        \begin{quantikz}[row sep={\the\circuitrowsep,between origins}, column sep=\the\circuitcolsep, align equals at=1.5]
            & \qwbundle{k} & & \ctrl{1} & & \\
            & \qwbundle{n} & \gate{V} & \gate{U} & \gate{V^\dag} &
        \end{quantikz}
        \caption{Controlled conjugation: when adding controls to $V^\dag U V$, only the middle gate $U$ needs to be controlled.}\label{fig:controlled_conjugation}
    \end{subfigure}

    \vspace{1em}

    \begin{subfigure}[b]{\textwidth}
        \centering
        \begin{quantikz}[row sep={\the\circuitrowsep,between origins}, column sep=\the\circuitcolsep, align equals at=2]
            & \qwbundle{k} & \ctrl{1} & \ctrl{1} & \ctrl{1} & \\
            & \qwbundle{n} & \gate{V} & \gate{U} & \gate{V^\dag} &
        \end{quantikz}
        =
        \begin{quantikz}[row sep={\the\circuitrowsep,between origins}, column sep=\the\circuitcolsep, align equals at=2]
            & \qwbundle{k} & \ctrl{2} & \gate{X} & \push{\weakdiamondsuit}\wire[d][1]{a} & \push{\weakdiamondsuit}\wire[d][1]{a} & \push{\weakdiamondsuit}\wire[d][1]{a} & \gate{X} & \ctrl{2} & \\
            & \qwbundle{n} & & & \gate{V} & \gate{U} & \gate{V^\dag} & & & \\
            \lstick{\ket{0}} & & \targ{} & & & \ctrl{-1} & & & \targ{} & \rstick{\ket{0}}
        \end{quantikz}
        \caption{Combining~\subref{fig:controlled_conjugation} with Figure~\ref{fig:conditionally_clean_ancillae}: the unitaries $V$ and $V^\dag$ need not be controlled and are replaced by weak promise gates with target unitaries $V$ and $V^\dag$, respectively, using the top $k$-qubit register as the promise register.}\label{fig:promise_controlled_conjugation}
    \end{subfigure}

    \vspace{1em}

    \begin{subfigure}[b]{\textwidth}
        \centering
        \begin{quantikz}[row sep={\the\circuitrowsep,between origins}, column sep=.2cm, align equals at=2]
            & \qwbundle{k} & & \ctrl{1} & \ctrl{1} & \ctrl{1} & \\
            & \qwbundle{n} & & \gate{V} & \gate{U} & \gate{V^\dag} &
        \end{quantikz}
        =\!\!
        \begin{quantikz}[row sep={\the\circuitrowsep,between origins}, column sep=.2cm, align equals at=2]
            & \qwbundle{k} & \ctrl{2} & \gate{X} & \push{\diamondsuit}\wire[d][1]{a} & \push{\diamondsuit}\wire[d][1]{a} & \push{\diamondsuit}\wire[d][1]{a} & \gate{X} & \ctrl{2} & \gate{X} & \push{\diamondsuit}\wire[d][1]{a} & \push{\diamondsuit}\wire[d][1]{a} & \push{\diamondsuit}\wire[d][1]{a} & \gate{X} & \\
            & \qwbundle{n} & & & \gate{V} & \gate{U} & \gate{V^\dag} & & & & \gate{V} & \gate{U^\dag} & \gate{V^\dag} & &\\
            \lstick{\ket{\psi}} & & \targ{} & & & \ctrl{-1} & & & \targ{} & & & \ctrl{-1} & & & \rstick{\ket{\psi}}
        \end{quantikz}
        \caption{Dirty-ancilla variant of~\subref{fig:promise_controlled_conjugation}, obtained by combining~\subref{fig:controlled_conjugation} with the right-hand side of Figure~\ref{fig:conditionally_clean_ancillae}. Requires $U^2 = I$.}\label{fig:promise_controlled_conjugation_dirty}
    \end{subfigure}
    \caption{Using controlled conjugation simplification with promise gates.}\label{fig:controlled_conjugation_promise}
\end{figure}

We now state the following theorem, which shows how to trade clean ancilla qubits for control qubits with little overhead in gate count and circuit depth.
\begin{theorem}\label{thm:ancilla_control_trade}
    Let $W = V^{\dagger} U V$ be an $n$-qubit unitary, where $V$ and $U$ are each implementable over a gate set $\mathcal{G}$ whose elements are all involutory (i.e., $G^2 = I$ for every $G \in \mathcal{G}$), using $m$~clean ancilla qubits with gate counts $c_V$, $c_U$ and circuit depths $d_V$, $d_U$, respectively.
    Then the $k$-controlled gate~$C^{k}W$ can be implemented over the gate set
    \[
        \mathcal{G} \;\cup\; C(\mathcal{G}) \;\cup\; \{CCX, CX, X\},
    \]
    where $C(\mathcal{G})$ denotes the set of singly controlled versions of gates in~$\mathcal{G}$, with gate count~$\mathcal{O}(c_V + c_U + d_U n + k)$, circuit depth~$\mathcal{O}(d_V + d_U \log n + \log k)$, and using $\max(1, m - k + 1)$~clean ancilla qubits.

    Moreover, if $U^2 = I$ and the implementations of $V$ and $U$ preserve the ancilla register for all inputs (so that they yield implementations of strong promise gates), then one clean ancilla can be replaced by one dirty ancilla with the same asymptotic gate count and circuit depth complexities.
\end{theorem}

We rely on the following lemma in the proof of Theorem~\ref{thm:ancilla_control_trade}.
\begin{lemma}\label{lem:parallel_control_gates}
    Let $\mathcal{G} = \{U_1, \dots, U_n\}$ be a set of $n>1$ unitary gates, each satisfying $U_i^2 = I$, and let $U = \bigotimes_{i=1}^{n} U_i$.
    Then the $k$-controlled gate~$C^kU$ can be implemented over the gate set
    \[
        C(\mathcal{G}) \;\cup\; \{CCX, CX, X\},
    \]
    where $C(\mathcal{G})$ denotes the set of singly controlled versions of gates in $\mathcal{G}$, with gate count~$\mathcal{O}(n + k)$ and circuit depth~$\mathcal{O}(\log n + \log k)$, without using any ancilla qubits.
\end{lemma}
\begin{proof}
    When $k=1$, we repeatedly apply the circuit identity of Figure~\ref{fig:toggle_detection}, each application using one dirty ancilla, to obtain the following circuit identity:
    \begin{equation}\label{eq:fanout_u}
        \begin{quantikz}[row sep={\the\circuitrowsep,between origins}, column sep=\the\circuitcolsep, align equals at=5]
            & & & \ctrl{6} & & \\
            & \qwbundle{} & & \gate{U_1} & & \\
            & & & & & \\
            & \qwbundle{} & & \gate{U_2} & & \\
            & & & & & \\
            & \qwbundle{} & & \gate{U_3} & & \\
            \setwiretype{n} & & & \push{\vdots} & & \\
            & & & \wire[u,shorten >=0.18cm][1]{a} & & \\
            & \qwbundle{} & & \gate{U_n}\wire[u][1]{a} & &
        \end{quantikz}
        \;=\;
        \begin{quantikz}[row sep={\the\circuitrowsep,between origins}, column sep=\the\circuitcolsep, align equals at=5]
            & & & \ctrl{6} & \ctrl{1} & \ctrl{6} & & & \\
            & \qwbundle{} & & & \gate{U_1} & & & & \\
            & & & \targ{} & \ctrl{1} & \targ{} & \ctrl{1} & & \\
            & \qwbundle{} & & & \gate{U_2} & & \gate{U_2} & & \\
            & & & \targ{} & \ctrl{1} & \targ{} & \ctrl{1} & & \\
            & \qwbundle{} & & & \gate{U_3} & & \gate{U_3} & & \\
            \setwiretype{n} & & & \push{\vdots} & & \push{\vdots} & & & \\
            & & & \targ{}\wire[u,shorten >=0.18cm][1]{a} & \ctrl{1} & \targ{}\wire[u,shorten >=0.18cm][1]{a} & \ctrl{1} & & \\
            & \qwbundle{} & & & \gate{U_n} & & \gate{U_n} & &
        \end{quantikz}
    \end{equation}
    which uses $n-1$ dirty ancilla qubits.
    To ensure that enough dirty ancilla qubits are available, we split the set $\mathcal{G}$ into two subsets of size $\lfloor n/2 \rfloor$ and $\lceil n/2 \rceil$, respectively.
    Without loss of generality, we may assume that each unitary $U_i$ acts on at least one qubit, so that the unitaries in each subset act on at least $\lfloor n/2 \rfloor$ qubits.
    Applying Equation~\eqref{eq:fanout_u} to each subset requires at most $\lceil n/2 \rceil - 1 \leq \lfloor n/2 \rfloor$ dirty ancilla qubits, so the qubits acted on by one subset can serve as dirty ancilla qubits for the other.
    The fan-out gate is then applied four times over at most $\lceil n/2 \rceil$ qubits each time, and each controlled $U_i$ gate is applied at most twice.
    As stated in Lemma~\ref{lem:fanout}, the fan-out gate can be implemented with $\mathcal{O}(n)$ $CX$ gates and circuit depth $\mathcal{O}(\log n)$.
    Thus, the total gate count is $\mathcal{O}(n)$ and the total circuit depth is $\mathcal{O}(\log n)$.

    When $k>1$, we use the following identity, also based on the circuit identity of Figure~\ref{fig:toggle_detection}:
    \begin{equation}\label{eq:fanout_u_k_ctrl}
        \begin{quantikz}[row sep={\the\circuitrowsep,between origins}, column sep=\the\circuitcolsep, align equals at=3]
            & \qwbundle{k} & & \ctrl{3} & & \\
            & & & & & \\
            & \qwbundle{} & & \gate{U_1} & & \\
            \setwiretype{n} & & & \push{\vdots} & & \\
            & \qwbundle{} & & \gate{U_n}\wire[u][1]{a} & & 
        \end{quantikz}
        \;=\;
        \begin{quantikz}[row sep={\the\circuitrowsep,between origins}, column sep=\the\circuitcolsep, align equals at=3]
            & \qwbundle{k} & & \ctrl{1} & & \ctrl{1} & & & \\
            & & & \targ{} & \ctrl{2} & \targ{} & \ctrl{2} & & \\
            & \qwbundle{} & & & \gate{U_1} & & \gate{U_1} & & \\
            \setwiretype{n} & & & & \push{\vdots} & & \push{\vdots} & & \\
            & \qwbundle{} & & & \gate{U_n}\wire[u][1]{a} & & \gate{U_n}\wire[u][1]{a} & & 
        \end{quantikz}
    \end{equation}
    This reduces the $k$-controlled case to two $C^kX$ gates and two instances of the singly controlled case.
    We apply this identity twice, once for each subset as described above, using the qubits of the other subset as dirty ancilla qubits.
    The two singly controlled instances in each application are then implemented using Equation~\eqref{eq:fanout_u}, where one dirty ancilla can be taken from the $k$ control qubits to compensate for the additional dirty ancilla required by Equation~\eqref{eq:fanout_u_k_ctrl}.
    The two $C^kX$ gates can be implemented with $\mathcal{O}(k)$ gates and $\mathcal{O}(\log k)$ circuit depth using one dirty ancilla (Lemma~\ref{lem:optimal_ckx}).
    Thus, the total gate count is $\mathcal{O}(n + k)$ and the total circuit depth is $\mathcal{O}(\log n + \log k)$.
\end{proof}

Using Lemma~\ref{lem:parallel_control_gates} together with the circuit identities presented in Figure~\ref{fig:controlled_conjugation_promise}, we can now prove Theorem~\ref{thm:ancilla_control_trade}.
\begin{proof}[Proof of Theorem~\ref{thm:ancilla_control_trade}]
    We begin by applying the circuit identity of Figure~\ref{fig:promise_controlled_conjugation}, which introduces one clean ancilla qubit and two $C^kX$ gates, and replaces the $k$-controlled $V$, $V^\dag$, and $U$ gates with weak promise gates $\tilde{V}$ and $\tilde{V}^\dag$ and a controlled weak promise gate $\tilde{U}$, all sharing a promise register of size $k$.
    Each $C^k X$ gate can be implemented with $\mathcal{O}(k)$ gates and $\mathcal{O}(\log k)$ depth using one dirty ancilla from the $n$-qubit register (Lemma~\ref{lem:optimal_ckx}).
    The weak promise gates $\tilde{V}$ and $\tilde{V}^\dag$ are not controlled, so we can use the implementations of $V$ and $V^\dag$ with $c_V$ gates, circuit depth $d_V$, and $m$ clean ancilla qubits, replacing $k$ of the clean ancilla qubits with qubits from the promise register and thereby requiring only $\max(0, m-k)$ clean ancilla qubits.

    It remains to implement the controlled weak promise gate $\tilde{U}$.
    As for $\tilde{V}$, we can implement the weak promise gate $\tilde{U}$ by using the construction for $U$ with $c_U$ gates, circuit depth $d_U$, and $\max(0, m-k)$ clean ancilla qubits, substituting $k$ clean ancilla qubits with the promise register.
    We then add a control to each gate in this implementation and apply the construction of Lemma~\ref{lem:parallel_control_gates} to each of the $d_U$ layers, for a total cost of $\mathcal{O}(d_U n)$ $C(\mathcal{G}) \cup \{CX\}$ gates and $\mathcal{O}(d_U \log n)$ circuit depth.
    The total gate count is therefore $\mathcal{O}(c_V + c_U + d_U n + k)$, and the total circuit depth is $\mathcal{O}(d_V + d_U \log n + \log k)$, using $\max(1, m - k + 1)$ clean ancilla qubits.

    Whenever $U$ satisfies $U^2=I$ and the implementations of $V$ and $U$ preserve the ancilla register for all inputs, we can use the circuit identity of Figure~\ref{fig:promise_controlled_conjugation_dirty} instead of the one of Figure~\ref{fig:promise_controlled_conjugation}, which replaces the clean ancilla qubit with a dirty ancilla qubit at the cost of at most doubling the gate count and circuit depth.
\end{proof}

Note that when $\mathcal{G} = \{CCX, CX, X\}$, the circuit produced by Theorem~\ref{thm:ancilla_control_trade} is over the $\{CCCX$, $CCX$, $CX$, $X\}$ gate set.
To return to the $\{CCX, CX, X\}$ gate set, we can decompose each $C^3X$ gate into $4$ $CCX$ gates using one dirty ancilla via the following identity:
\begin{equation}\label{eq:cccx_to_ccx}
    \begin{quantikz}[row sep={\the\circuitrowsep,between origins}, column sep=\the\circuitcolsep, align equals at=3]
        & \ctrl{1} & \\
        & \ctrl{1} & \\
        & \ctrl{1} & \\
        & \targ{} & \\
        \lstick{\ket{\psi}} & & \rstick{\ket{\psi}}
    \end{quantikz}
    \;=\;
    \begin{quantikz}[row sep={\the\circuitrowsep,between origins}, column sep=\the\circuitcolsep, align equals at=3]
        & \ctrl{1} & & \ctrl{1} & & \\
        & \ctrl{3} & & \ctrl{3} & & \\
        & & \ctrl{1} & & \ctrl{1} & \\
        & & \targ{} & & \targ{} & \\
        \lstick{\ket{\psi}} & \targ{} & \ctrl{-1} & \targ{} & \ctrl{-1} & \rstick{\ket{\psi}}
    \end{quantikz}
\end{equation}
To ensure that each $C^3X$ decomposition has access to a dirty ancilla, we use the same splitting technique as in the proof of Lemma~\ref{lem:parallel_control_gates}: within each circuit layer, we partition the $C^3X$ gates into two groups and use the qubits acted on by one group as dirty ancilla qubits for the other.
This yields a circuit over the $\{CCX, CX, X\}$ gate set with the same asymptotic gate count and circuit depth as in Theorem~\ref{thm:ancilla_control_trade}.
This decomposition introduces only a constant factor overhead in gate count and circuit depth, yielding the following corollary.

\begin{corollary}\label{cor:ancilla_control_trade}
    Under the hypotheses of Theorem~\ref{thm:ancilla_control_trade} with $\mathcal{G} = \{CCX, CX, X\}$, the $k$-controlled gate $C^kW$ can be implemented over the $\{CCX, CX, X\}$ gate set with gate count $\mathcal{O}(c_V + c_U + d_U n + k)$, circuit depth $\mathcal{O}(d_V + d_U \log n + \log k)$, and using $\max(1, m - k + 1)$ clean ancilla qubits.
    Moreover, if $U^2 = I$ and the implementations of $V$ and $U$ preserve the ancilla register for all inputs, then one clean ancilla can be replaced by one dirty ancilla with the same asymptotic gate count and circuit depth complexities.
\end{corollary}

\subsection{Application to controlled addition}\label{sub:controlled_addition}

\begin{figure}[t]
    \makebox[\textwidth][c]{%
    \begin{quantikz}[row sep={\the\circuitrowsep,between origins}, column sep=\the\circuitcolsep]
        \lstick{\ket{a_0}} & & \slice{1} & & & \slice{2} & \ctrl{1} & & & \slice{3} & & & \slice{4} & & & & \ctrl{1}\slice{5} & \slice{6} & & & \slice{7} & \ctrl{1}\slice{8} & \rstick{\ket{a_0}} \\
        \lstick{\ket{b_0}} & & & & & & \ctrl{1} & & & & & & & & & & \ctrl{1} & & & & & \targ{} & \rstick{\ket{s_0}} \\
        \lstick{\ket{a_1}} & \ctrl{1} & & & & \ctrl{2} & \targ{} & \ctrl{1} & & & & \ctrl{1} & & & & \ctrl{1} & \targ{} & & \ctrl{2} & & & \ctrl{1} & \rstick{\ket{a_1}} \\
        \lstick{\ket{b_1}} & \targ{} & & & & & & \ctrl{1} & & & & \targ{} & \gate{X} & & & \ctrl{1} & & \gate{X} & & & & \targ{} & \rstick{\ket{s_1}} \\
        \lstick{\ket{a_2}} & \ctrl{1} & & & \ctrl{2} & \targ{} & & \targ{} & \ctrl{1} & & & \ctrl{1} & & & \ctrl{1} & \targ{} & & & \targ{} & \ctrl{2} & & \ctrl{1} & \rstick{\ket{a_2}} \\
        \lstick{\ket{b_2}} & \targ{} & & & & & & & \ctrl{1} & & & \targ{} & \gate{X} & & \ctrl{1} & & & \gate{X} & & & & \targ{} & \rstick{\ket{s_2}} \\
        \lstick{\ket{a_3}} & \ctrl{1} & & \ctrl{2} & \targ{} & & & & \targ{} & \ctrl{1} & & \ctrl{1} & & \ctrl{1} & \targ{} & & & & & \targ{} & \ctrl{2} & \ctrl{1} & \rstick{\ket{a_3}} \\
        \lstick{\ket{b_3}} & \targ{} & & & & & & & & \ctrl{1} & & \targ{} & \gate{X} & \ctrl{1} & & & & \gate{X} & & & & \targ{} & \rstick{\ket{s_3}} \\
        \lstick{\ket{a_4}} & \ctrl{1} & \ctrl{2} & \targ{} & & & & & & \targ{} & \ctrl{1} & \ctrl{1} & & \targ{} & & & & & & & \targ{} & \ctrl{1} & \rstick{\ket{a_4}} \\
        \lstick{\ket{b_4}} & \targ{} & & & & & & & & & \ctrl{1} & \targ{} & & & & & & & & & & \targ{} & \rstick{\ket{s_4}} \\
        \lstick{\ket{z}} & & \targ{} & & & & & & & & \targ{} & & & & & & & & & & & & \rstick{\ket{z \oplus s_5}} 
    \end{quantikz}
    }
    \caption{Ancilla-free circuit for quantum addition over $5$-bit registers, where $s = a + b$~\cite{Takahashi_2010}.}\label{fig:ripple_carry_adder}
\end{figure}
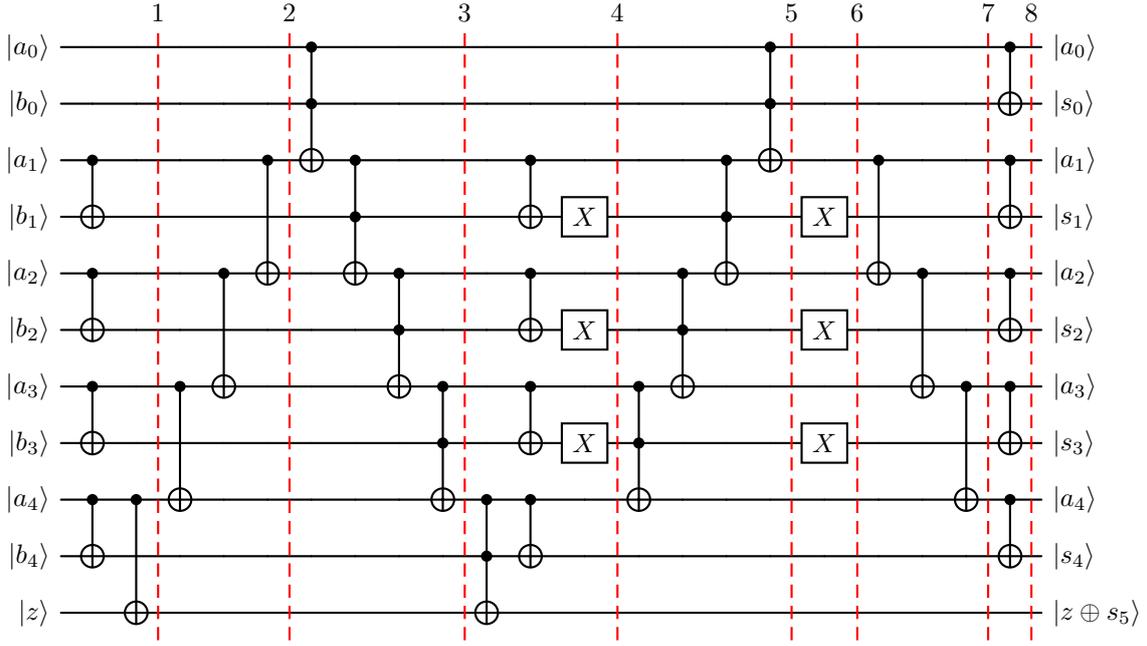

We now demonstrate a straightforward application of Theorem~\ref{thm:ancilla_control_trade}: efficiently implementing a controlled adder by trading clean ancilla qubits for control qubits.

\begin{corollary}\label{cor:controlled_addition}
    The $k$-controlled $n$-bit quantum adder can be implemented over the $\{CCX, CX, X\}$ gate set with $\mathcal{O}(k + n)$ gates, $\mathcal{O}(\log(kn))$ circuit depth, and $\max(1, n-k+1)$ clean ancilla qubits.
\end{corollary}
\begin{proof}
    We start from the ripple-carry adder of~\cite{Takahashi_2010}, an example of which is shown in Figure~\ref{fig:ripple_carry_adder}.
    Let $U_i$ denote the unitary implemented by slice $i$ in Figure~\ref{fig:ripple_carry_adder}.
    The ladders of $CX$ gates in slices 2 and 7, corresponding to $U_2$ and $U_7 = U_2^\dag$, can each be implemented with $\mathcal{O}(n)$ $CX$ gates and circuit depth $\mathcal{O}(\log n)$, as stated in Lemma~\ref{lem:ladder_cnot}.
    The ladders of $CCX$ gates in slices 3 and 5, corresponding to $U_3$ and $U_5 = U_3^\dag$, can each be implemented with $\mathcal{O}(n)$ $CCX$ gates and circuit depth $\mathcal{O}(\log n)$, using $n$ clean ancilla qubits, as stated in Lemma~\ref{lem:ladder_2_n_ancilla}.

    Applying the controlled conjugation identity of Figure~\ref{fig:controlled_conjugation} twice, first to the outer $U_2 / U_2^\dag$ pair and then to the inner $U_3 / U_3^\dag$ pair, the $k$-controlled adder can be written as
    \begin{equation}\label{eq:controlled_adder_decomp}
        C^k A = C^k(U_8) \cdot U_2^\dag \cdot C^k(U_6) \cdot U_3^\dag \cdot C^k(U_4) \cdot U_3 \cdot U_2 \cdot C^k(U_1).
    \end{equation}
    We apply Theorem~\ref{thm:ancilla_control_trade} to implement $U_3^\dag \cdot C^k(U_4) \cdot U_3$ in Equation~\eqref{eq:controlled_adder_decomp} with $\mathcal{O}(k + n)$ gates and $\mathcal{O}(\log(kn))$ circuit depth, using $\max(1, n - k + 1)$ clean ancilla qubits.
    The single $C^3X$ gate we obtain in slice 4 can be decomposed into four $CCX$ gates using one available dirty ancilla, as in Equation~\eqref{eq:cccx_to_ccx}.
    The remaining controlled unitaries $C^k(U_8)$, $C^k(U_6)$, and $C^k(U_1)$ can each be implemented over the $\{CCX, CX, X\}$ gate set with $\mathcal{O}(k + n)$ gates and $\mathcal{O}(\log(kn))$ circuit depth by Lemma~\ref{lem:parallel_control_gates}.
    Thus, the resulting $k$-controlled adder has $\mathcal{O}(k + n)$ $\{CCX, CX, X\}$ gates, $\mathcal{O}(\log(kn))$ circuit depth, and uses $\max(1, n - k + 1)$ clean ancilla qubits.
\end{proof}

Theorem~\ref{thm:ancilla_control_trade} is not limited to simply adding controls to a unitary so that they can be traded for ancilla qubits.
As we will see in the following sections, it is also applicable when an operator decomposes into a multi-controlled unitary and a simpler one, allowing its full benefit to be exploited without adding controls to the entire operator.

\section{Quantum circuits for comparators}\label{sec:comparator}
\subsection{Quantum--quantum comparator} \label{sub:comparator_quantum_quantum}

A quantum--quantum comparator over $n$-bit registers computes the following map:
\begin{equation}
    \ket{\bs a}\ket{\bs b}\ket{z} \mapsto \ket{\bs a}\ket{\bs b}\ket{z \oplus (a < b)},
\end{equation}
where $\bs a, \bs b \in \{0, 1\}^n$ and $z \in \{0, 1\}$.
Note that such a comparator can also be used to compute $a \leq b$, since $a \leq b$ is equivalent to $\neg (b < a)$.
The comparison operator can be implemented by computing the high bit of $a - b$, which is $1$ if and only if $a < b$~\cite{Cuccaro_2004}.
The subtraction $a - b$ can be constructed from an adder by conjugating it with $X$ gates on the register holding $b$.
An example of a comparator obtained from this construction is shown in Figure~\ref{fig:comparator}, using the adder from Figure~\ref{fig:ripple_carry_adder}.
Because this approach relies on a quantum adder, comparators have traditionally inherited the same circuit complexity as addition.
In this section, we show that this need not be the case: comparison can be implemented with lower circuit complexity than the best-known adders.

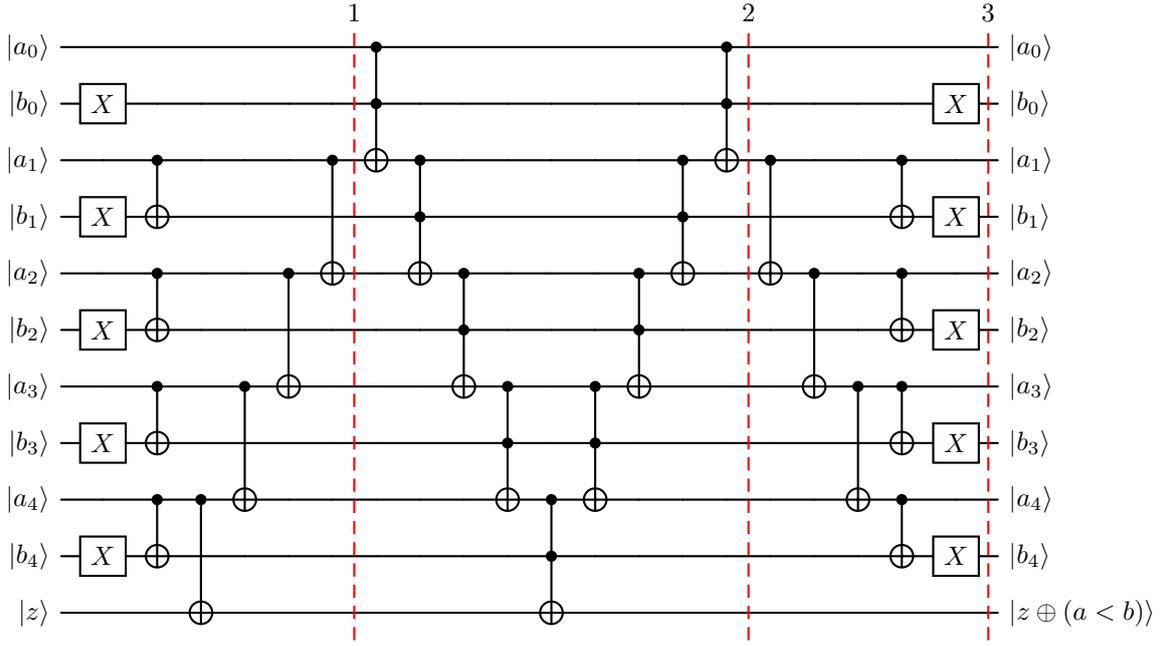
\begin{figure}[t]
    \makebox[\textwidth][c]{%
    \begin{quantikz}[row sep={\the\circuitrowsep,between origins}, column sep=\the\circuitcolsep]
        \lstick{\ket{a_0}} & & & & & & \slice{1} & \ctrl{1} & & & & & & & & \ctrl{1}\slice{2} & & & & & \slice{3} & \rstick{\ket{a_0}} \\
        \lstick{\ket{b_0}} & \gate{X} & & & & & & \ctrl{1} & & & & & & & & \ctrl{1} & & & & & \gate{X} & \rstick{\ket{b_0}} \\
        \lstick{\ket{a_1}} & & \ctrl{1} & & & & \ctrl{2} & \targ{} & \ctrl{1} & & & & & & \ctrl{1} & \targ{} & \ctrl{2} & & & \ctrl{1} & & \rstick{\ket{a_1}} \\
        \lstick{\ket{b_1}} & \gate{X} & \targ{} & & & & & & \ctrl{1} & & & & & & \ctrl{1} & & & & & \targ{} & \gate{X} & \rstick{\ket{b_1}} \\
        \lstick{\ket{a_2}} & & \ctrl{1} & & & \ctrl{2} & \targ{} & & \targ{} & \ctrl{1} & & & & \ctrl{1} & \targ{} & & \targ{} & \ctrl{2} & & \ctrl{1} & & \rstick{\ket{a_2}} \\
        \lstick{\ket{b_2}} & \gate{X} & \targ{} & & & & & & & \ctrl{1} & & & & \ctrl{1} & & & & & & \targ{} & \gate{X} & \rstick{\ket{b_2}} \\
        \lstick{\ket{a_3}} & & \ctrl{1} & & \ctrl{2} & \targ{} & & & & \targ{} & \ctrl{1} & & \ctrl{1} & \targ{} & & & & \targ{} & \ctrl{2} & \ctrl{1} & & \rstick{\ket{a_3}} \\
        \lstick{\ket{b_3}} & \gate{X} & \targ{} & & & & & & & & \ctrl{1} & & \ctrl{1} & & & & & & & \targ{} & \gate{X} & \rstick{\ket{b_3}} \\
        \lstick{\ket{a_4}} & & \ctrl{1} & \ctrl{2} & \targ{} & & & & & & \targ{} & \ctrl{1} & \targ{} & & & & & & \targ{} & \ctrl{1} & & \rstick{\ket{a_4}} \\
        \lstick{\ket{b_4}} & \gate{X} & \targ{} & & & & & & & & & \ctrl{1} & & & & & & & & \targ{} & \gate{X} & \rstick{\ket{b_4}} \\
        \lstick{\ket{z}} & & & \targ{} & & & & & & & & \targ{} & & & & & & & & & & \rstick{\ket{z \oplus (a < b)}} 
    \end{quantikz}
    }
    \caption{Ancilla-free circuit for the quantum--quantum comparator over $5$-bit registers.}\label{fig:comparator}
\end{figure}

A key component of the comparator is the $\mathcal{V}_2^{(n)}$ operator, as can be seen in slice 2 of Figure~\ref{fig:comparator}.
This operator admits the following implementation, using the $\mathcal{L}_2^{(n-1)}$ operator and its adjoint together with a $CCX$ gate in the middle:
\begin{equation}\label{eq:v_to_ladders}
    \begin{quantikz}[row sep={\the\circuitrowsep,between origins}, column sep=\the\circuitcolsep, align equals at=2.5]
        & \qwbundle{} & \gate[3]{\mathcal{V}_2}\vqw{3} & & \\
        & & & & \\
        & & & & \\
        & & \targ{} & & 
    \end{quantikz}
    \;=\;
    \begin{quantikz}[row sep={\the\circuitrowsep,between origins}, column sep=\the\circuitcolsep, align equals at=2.5]
        & \qwbundle{} & \gate[2]{\mathcal{L}^\dag_2} & & \gate[2]{\mathcal{L}_2} & & \\
        & & & \ctrl{1} & & & \\
        & & & \ctrl{1} & & & \\
        & & & \targ{} & & & 
    \end{quantikz}
\end{equation}

This decomposition motivates the following corollary.
\begin{corollary}\label{cor:promise_ladder}
    A strong promise gate with a promise register of size $n$ whose target unitary is the $\mathcal{L}_k^{(n)}$ operator can be implemented with $\mathcal{O}(kn)$ $\{CCX, CX, X\}$ gates and $\mathcal{O}(\log(kn))$ circuit depth.
\end{corollary}
This corollary follows straightforwardly from the construction used to prove Lemma~\ref{lem:ladder_2_n_ancilla}; the proof is provided in Appendix~\ref{app:proof_ladder_promise}.

The following lemma provides a key construction used in the main result of this section.
\begin{lemma}\label{lem:promise_v_sqrt_ancillae}
    A controlled strong promise gate with a promise register of size $2\lceil\sqrt{n}\rceil$ whose target unitary is the $\mathcal{V}_2^{(n)}$ operator can be implemented with $\mathcal{O}(n)$ $\{CCX, CX, X\}$ gates and $\mathcal{O}(\log n)$ circuit depth.
\end{lemma}
\begin{proof}
    The $\mathcal{V}_2^{(n)}$ operator admits the following decomposition:
    \begin{equation}\label{eq:v_ckx_decomposition}
        \begin{quantikz}[row sep={\the\circuitrowsep,between origins}, column sep=\the\circuitcolsep, align equals at=2.5]
            & & \gate[3]{\mathcal{V}_2}\vqw{3} & & \\
            & & & & \\
            & \qwbundle{} & & & \\
            & & \targ{} & & 
        \end{quantikz}
        \;=\;
        \begin{quantikz}[row sep={\the\circuitrowsep,between origins}, column sep=\the\circuitcolsep, align equals at=2.5]
            & & \ctrl{1} & & & \\
            & & \ctrl{1} & & & \\
            & \qwbundle{} & \push{\partctrl}\vqw{1} & \gate{\mathcal{V}_2}\vqw{1} & & \\
            & & \targ{} & \targ{} & & 
        \end{quantikz}
    \end{equation}
    where the $\partctrl$ symbol indicates controls only on alternating qubits:
    \begin{equation}\label{eq:alternating_ctrl}
        \begin{quantikz}[row sep={.5cm,between origins}, column sep=\the\circuitcolsep, align equals at=1.5]
            & \qwbundle{k} & \push{\partctrl}\vqw{1} & & \\
            & & \targ{} & & 
        \end{quantikz}
        \;:=\;
        \begin{quantikz}[row sep={.5cm,between origins}, column sep=\the\circuitcolsep, align equals at=5.5]
            & & & & \\
            & & \ctrl{2} & & \\
            & & & & \\
            & & \ctrl{2} & & \\
            & & & & \\
            & & \ctrl{1} & & \\
            \setwiretype{n} & & \push{\vdots} & & \\
            & & \wire[u,shorten >=0.15cm][1]{a} & & \\
            & & \ctrl{1}\vqw{-1} & & \\
            & & \targ{} & & 
        \end{quantikz}
    \end{equation}
    By Equations~\eqref{eq:v_to_ladders} and~\eqref{eq:v_ckx_decomposition}, we obtain the following implementation of the $\mathcal{V}_2^{(n)}$ operator:
    \begin{equation}
        \begin{quantikz}[row sep={\the\circuitrowsep,between origins}, column sep=\the\circuitcolsep, align equals at=6]
            \setwiretype{n} & & & & & & & & & & & & \push{\vdots} & & \\ 
            & \qwbundle{\!\!\!\!\!\!\!2\lceil\sqrtsign{n}\rceil} & & & & & & & & \gate[2]{\mathcal{L}^\dag_2} & & \gate[2]{\mathcal{L}_2} & \push{\partctrl}\vqw{2}\wire[u,shorten >=0.18cm][1]{a} & \dots & \\
            & & & & & & & & & & \ctrl{1} & & & \dots & \\
            & \qwbundle{\!\!\!\!\!\!\!2\lceil\sqrtsign{n}\rceil} & & & & & & \gate[2]{\mathcal{L}^\dag_2} & & \gate[2]{\mathcal{L}_2} & \push{\partctrl}\vqw{2} & & \push{\partctrl}\vqw{2} & \dots & \\
            & & & & & & & & \ctrl{1} & & & & & \dots & \\
            & \qwbundle{\!\!\!\!\!\!\!2\lceil\sqrtsign{n}\rceil} & & & & \gate[2]{\mathcal{L}^\dag_2} & & \gate[2]{\mathcal{L}_2} & \push{\partctrl}\vqw{2} & & \push{\partctrl}\vqw{2} & & \push{\partctrl}\vqw{2} & \dots & \\
            & & & & & & \ctrl{1} & & & & & & & \dots & \\
            & \qwbundle{\!\!\!\!\!\!\!2\lceil\sqrtsign{n}\rceil} & & \gate[2]{\mathcal{L}^\dag_2} & & \gate[2]{\mathcal{L}_2} & \push{\partctrl}\vqw{2} & & \push{\partctrl}\vqw{2} & & \push{\partctrl}\vqw{2} & & \push{\partctrl}\vqw{2} & \dots & \\
            & & & & \ctrl{1} & & & & & & & & & \dots & \\
            & & & & \ctrl{1} & & \ctrl{1} & & \ctrl{1} & & \ctrl{1} & & \ctrl{1} & \dots & \\
            & & & & \targ{} & & \targ{} & & \targ{} & & \targ{} & & \targ{} & \dots & 
        \end{quantikz}
    \end{equation}
    where the topmost register has size at most $2\lceil\sqrt{n}\rceil$.

    We then introduce $\lceil\sqrt{n}\rceil$ clean ancilla qubits to parallelize the $\mathcal{L}_2$ gates as follows:
    
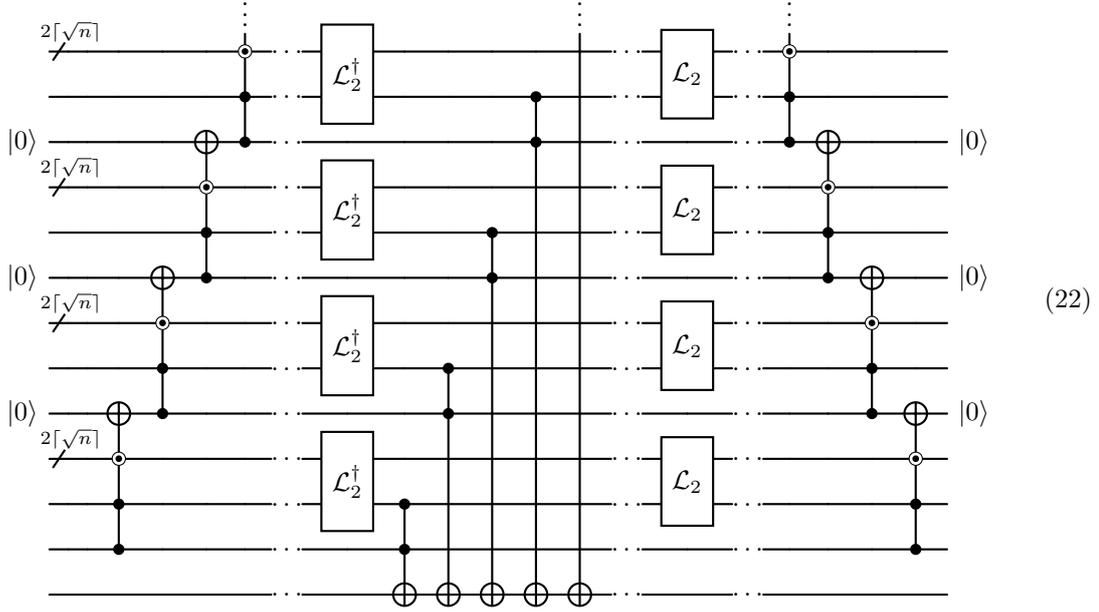
\begin{equation}
        \begin{quantikz}[row sep={.6cm,between origins}, column sep=\the\circuitcolsep, align equals at=7.5]
            \setwiretype{n} & & & & & & \push{\vdots} & & & & & & & \push{\vdots} & & & & \push{\vdots} & & & & \\
            & \qwbundle{\!\!\!\!\!\!\!2\lceil\sqrtsign{n}\rceil} & & & & & \push{\partctrl}\wire[u,shorten >=0.1cm][1]{a} & \dots & \gate[2]{\mathcal{L}^\dag_2} & & & & & \wire[u,shorten >=0.1cm][1]{a} & \dots & \gate[2]{\mathcal{L}_2} & \dots & \push{\partctrl}\wire[u,shorten >=0.1cm][1]{a} & & & & \\
            & & & & & & \ctrl{-1} & \dots & & & & & \ctrl{1} & & \dots & & \dots & \ctrl{-1} & & & & \\
            \lstick{\ket{0}} & & & & & \targ{} & \ctrl{-1} & \dots & & & & & \ctrl{10} & & \dots & & \dots & \ctrl{-1} & \targ{} & & & \rstick{\ket{0}} \\
            & \qwbundle{\!\!\!\!\!\!\!2\lceil\sqrtsign{n}\rceil} & & & & \push{\partctrl}\vqw{-1} & & \dots & \gate[2]{\mathcal{L}^\dag_2} & & & & & & \dots & \gate[2]{\mathcal{L}_2} & \dots & & \push{\partctrl}\vqw{-1} & & & \\
            & & & & & \ctrl{-1} & & \dots & & & & \ctrl{1} & & & \dots & & \dots & & \ctrl{-1} & & & \\
            \lstick{\ket{0}} & & & & \targ{} & \ctrl{-1} & & \dots & & & & \ctrl{7} & & & \dots & & \dots & & \ctrl{-1} & \targ{} & & \rstick{\ket{0}} \\
            & \qwbundle{\!\!\!\!\!\!\!2\lceil\sqrtsign{n}\rceil} & & & \push{\partctrl}\vqw{-1} & & & \dots & \gate[2]{\mathcal{L}^\dag_2} & & & & & & \dots & \gate[2]{\mathcal{L}_2} & \dots & & & \push{\partctrl}\vqw{-1} & & \\
            & & & & \ctrl{-1} & & & \dots & & & \ctrl{1} & & & & \dots & & \dots & & & \ctrl{-1} & & \\
            \lstick{\ket{0}} & & & \targ{} & \ctrl{-1} & & & \dots & & & \ctrl{4} & & & & \dots & & \dots & & & \ctrl{-1} & \targ{} & \rstick{\ket{0}} \\
            & \qwbundle{\!\!\!\!\!\!\!2\lceil\sqrtsign{n}\rceil} & & \push{\partctrl}\vqw{-1} & & & & \dots & \gate[2]{\mathcal{L}^\dag_2} & & & & & & \dots & \gate[2]{\mathcal{L}_2} & \dots & & & & \push{\partctrl}\vqw{-1} & \\
            & & & \ctrl{-1} & & & & \dots & & \ctrl{1} & & & & & \dots & & \dots & & & & \ctrl{-1} & \\
            & & & \ctrl{-1} & & & & \dots & & \ctrl{1} & & & & & \dots & & \dots & & & & \ctrl{-1} & \\
            & & & & & & & \dots & & \targ{} & \targ{} & \targ{} & \targ{} & \targ{}\vqw{-12} & \dots & & \dots & & & & & 
        \end{quantikz}
    \end{equation}
    We can then apply the circuit identity of Figure~\ref{fig:promise_controlled_conjugation}, where $V$ and $V^\dag$ correspond to the $\mathcal{L}_2^\dag$ and $\mathcal{L}_2$ operators and $U$ corresponds to the $CCX$ gate in between.
    Each $CCX$ gate in the middle of the circuit acts nontrivially only when the ancilla qubit on which it is controlled is in the $\ket{1}$ state.
    In turn, this ancilla is in the $\ket{1}$ state only when the register below it is in the $\ket{1}$ state.
    For each $\mathcal{L}_2$ gate, the register on which the associated $CCX$ gate is conditioned can therefore serve as a promise register.
    Since each promise gate uses the register below it as its promise register, we cannot apply all promise gates simultaneously; instead, we proceed in two rounds.
    We obtain the following circuit:
    
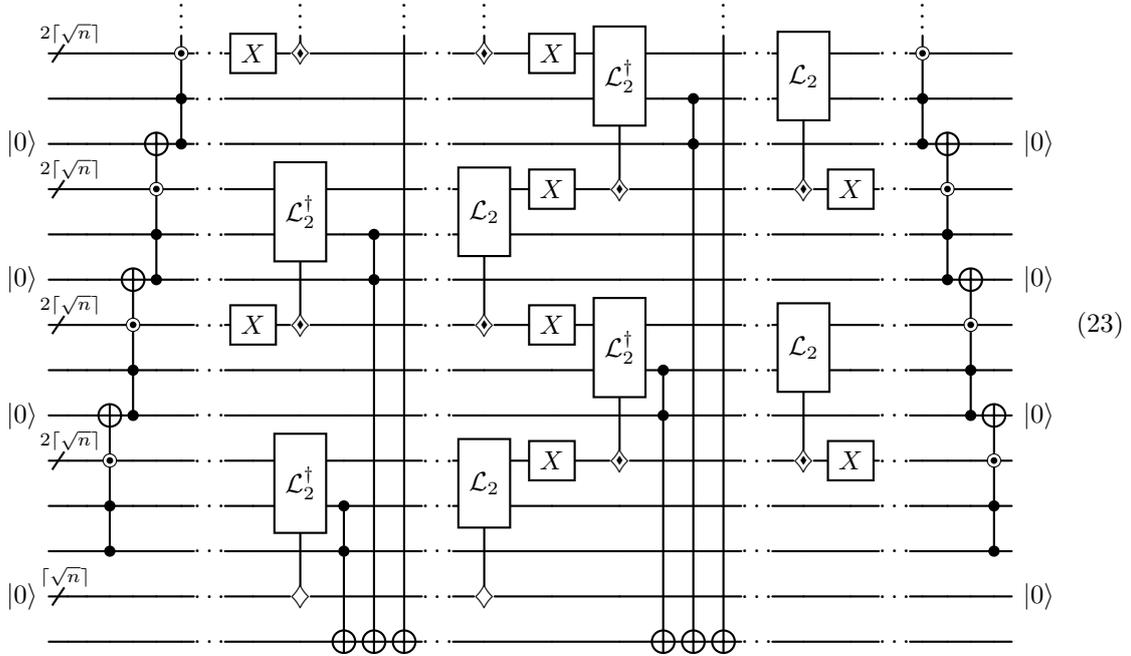
\begin{equation}
        \begin{quantikz}[row sep={.6cm,between origins}, column sep=\the\circuitcolsep, align equals at=8]
            \setwiretype{n} & & &[-0.11cm] &[-0.27cm] &[-0.27cm] &[-0.18cm] \push{\vdots} &[-0.18cm] &[-0.18cm] &[-0.27cm] \push{\vdots} &[-0.18cm] &[-0.18cm] &[-0.18cm] \push{\vdots} &[-0.18cm] &[-0.18cm] \push{\vdots} & & &[-0.18cm] &[-0.18cm] &[-0.18cm] \push{\vdots} &[-0.18cm] &[-0.18cm] &[-0.27cm] &[-0.18cm] &[-0.18cm] \push{\vdots} &[-0.18cm] &[-0.27cm] &[-0.27cm] &[-0.18cm] \\
            & \qwbundle{\!\!\!\!\!\!\!2\lceil\sqrtsign{n}\rceil} & & & & & \push{\partctrl}\wire[u,shorten >=0.1cm][1]{a} & \dots & \gate{X} & \push{\partdiamondsuit}\wire[u,shorten >=0.1cm][1]{a} & & & \wire[u,shorten >=0.1cm][1]{a} & \dots & \push{\partdiamondsuit}\wire[u,shorten >=0.1cm][1]{a} & \gate{X} & \gate[2]{\mathcal{L}^\dag_2} & & & \wire[u,shorten >=0.1cm][1]{a} & \dots & \gate[2]{\mathcal{L}_2} & & \dots & \push{\partctrl}\wire[u,shorten >=0.1cm][1]{a} & & & & \\
            & & & & & & \ctrl{-1} & \dots & & & & & & \dots & & & & & \ctrl{1} & & \dots & & & \dots & \ctrl{-1} & & & & \\
            \lstick{\ket{0}} & & & & & \targ{} & \ctrl{-1} & \dots & & & & & & \dots & & & & & \ctrl{11} & & \dots & & & \dots & \ctrl{-1} & \targ{} & & & \rstick{\ket{0}} \\
            & \qwbundle{\!\!\!\!\!\!\!2\lceil\sqrtsign{n}\rceil} & & & & \push{\partctrl}\vqw{-1} & & \dots & & \gate[2]{\mathcal{L}^\dag_2} & & & & \dots & \gate[2]{\mathcal{L}_2} & \gate{X} & \push{\partdiamondsuit}\vqw{-2} & & & & \dots & \push{\partdiamondsuit}\vqw{-2} & \gate{X} & \dots & & \push{\partctrl}\vqw{-1} & & & \\
            & & & & & \ctrl{-1} & & \dots & & & & \ctrl{1} & & \dots & & & & & & & \dots & & & \dots & & \ctrl{-1} & & & \\
            \lstick{\ket{0}} & & & & \targ{} & \ctrl{-1} & & \dots & & & & \ctrl{8} & & \dots & & & & & & & \dots & & & \dots & & \ctrl{-1} & \targ{} & & \rstick{\ket{0}} \\
            & \qwbundle{\!\!\!\!\!\!\!2\lceil\sqrtsign{n}\rceil} & & & \push{\partctrl}\vqw{-1} & & & \dots & \gate{X} & \push{\partdiamondsuit}\vqw{-2} & & & & \dots & \push{\partdiamondsuit}\vqw{-2} & \gate{X} & \gate[2]{\mathcal{L}^\dag_2} & & & & \dots & \gate[2]{\mathcal{L}_2} & & \dots & & & \push{\partctrl}\vqw{-1} & & \\
            & & & & \ctrl{-1} & & & \dots & & & & & & \dots & & & & \ctrl{1} & & & \dots & & & \dots & & & \ctrl{-1} & & \\
            \lstick{\ket{0}} & & & \targ{} & \ctrl{-1} & & & \dots & & & & & & \dots & & & & \ctrl{5} & & & \dots & & & \dots & & & \ctrl{-1} & \targ{} & \rstick{\ket{0}} \\
            & \qwbundle{\!\!\!\!\!\!\!2\lceil\sqrtsign{n}\rceil} & & \push{\partctrl}\vqw{-1} & & & & \dots & & \gate[2]{\mathcal{L}^\dag_2} & & & & \dots & \gate[2]{\mathcal{L}_2} & \gate{X} & \push{\partdiamondsuit}\vqw{-2} & & & & \dots & \push{\partdiamondsuit}\vqw{-2} & \gate{X} & \dots & & & & \push{\partctrl}\vqw{-1} & \\
            & & & \ctrl{-1} & & & & \dots & & & \ctrl{1} & & & \dots & & & & & & & \dots & & & \dots & & & & \ctrl{-1} & \\
            & & & \ctrl{-1} & & & & \dots & & & \ctrl{2} & & & \dots & & & & & & & \dots & & & \dots & & & & \ctrl{-1} & \\
            \lstick{\ket{0}} & \qwbundle{\!\!\!\!\!\!\!\lceil\sqrtsign{n}\rceil} & & & & & & \dots & & \push{\diamondsuit}\vqw{-2} & & & & \dots & \push{\diamondsuit}\vqw{-2} & & & & & & \dots & & & \dots & & & & & \rstick{\ket{0}} \\
            & & & & & & & \dots & & & \targ{} & \targ{} & \targ{}\vqw{-13} & \dots & & & & \targ{} & \targ{} & \targ{}\vqw{-13} & \dots & & & \dots & & & & & 
        \end{quantikz}
    \end{equation}
    where $\lceil\sqrt{n}\rceil$ clean ancilla qubits were inserted at the bottom to serve as the promise register for the bottommost promise gates, and where $\partdiamondsuit$ denotes $\diamondsuit$ symbols on alternating qubits, analogously to $\partctrl$ in Equation~\eqref{eq:alternating_ctrl}.

    We then replace the clean ancilla qubits with a promise register of size $2\lceil\sqrt{n}\rceil$ and add a control to each gate to obtain a controlled strong promise gate whose target unitary is the $\mathcal{V}_2$ operator.
    By Figure~\ref{fig:controlled_conjugation}, only the $CCX$ gates targeting the bottom qubit need to be controlled.
    We obtain the circuit presented in Figure~\ref{fig:ctrl_promise_v_sqrt}.

    We now analyze the cost of the resulting circuit.
    Using $\lceil\sqrt{n}\rceil$ qubits from the promise register, the first and last ladders of multi-controlled $X$ gates can be implemented with $\mathcal{O}(n)$ gates and $\mathcal{O}(\log n)$ circuit depth, as stated in Corollary~\ref{cor:ckx_ladder}.
    The $CCX$ gates targeting the same qubit and controlled by the bottom qubits can be implemented as follows:
    
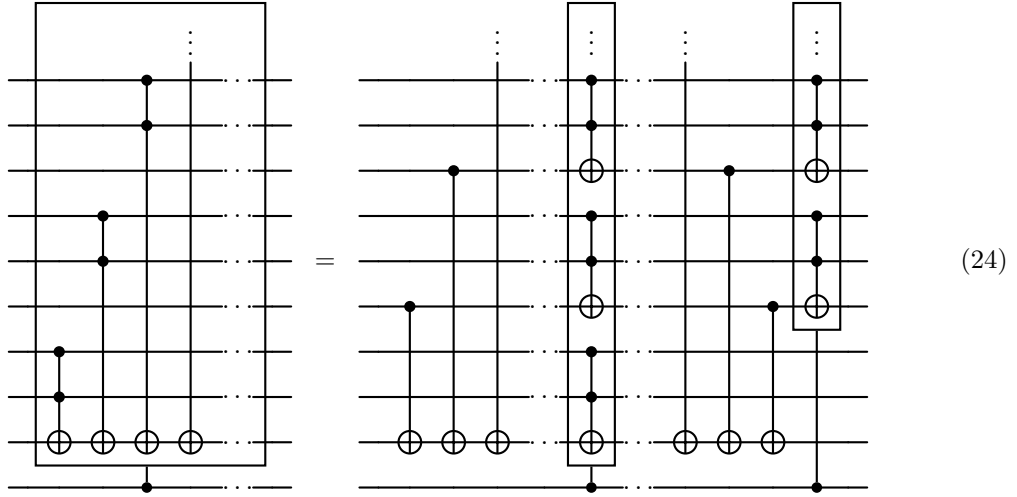
\begin{equation}\label{eq:controlled_fanin}
        \begin{quantikz}[row sep={.6cm,between origins}, column sep=\the\circuitcolsep, align equals at=6, execute at end picture={
            \draw [line width=0.8pt] (ctrl_first.north) -- (ccxbox.south -| ctrl_first.north);
        }]
            \setwiretype{n} & & \gategroup[10, steps=5, style={name=ccxbox, inner sep=1pt}]{} & & & \push{\vdots} & & \\
            & & & & \ctrl{1} & \wire[u,shorten >=0.1cm][1]{a} & \dots & & \\
            & & & & \ctrl{7} & & \dots & & \\
            & & & & & & \dots & & \\
            & & & \ctrl{1} & & & \dots & & \\
            & & & \ctrl{4} & & & \dots & & \\
            & & & & & & \dots & & \\
            & & \ctrl{1} & & & & \dots & & \\
            & & \ctrl{1} & & & & \dots & & \\
            & & \targ{} & \targ{} & \targ{} & \targ{}\vqw{-8} & \dots & & \\
            & & & & |[alias=ctrl_first, fill, circle, minimum size=4pt, inner sep=0pt]|{} & & \dots & & 
        \end{quantikz}
        \;=\;
        \begin{quantikz}[row sep={.6cm,between origins}, column sep=\the\circuitcolsep, align equals at=6, execute at end picture={
            \draw [line width=0.8pt] (ctrl_first.north) -- (ccxbox.south -| ctrl_first.north);
            \draw [line width=0.8pt] (ctrl_second.north) -- (ccxboxdag.south -| ctrl_second.north);
        }]
            \setwiretype{n} & & & & \push{\vdots} & & \push{\vdots}\gategroup[10, steps=1, style={name=ccxbox, inner sep=1pt}]{} & & \push{\vdots} & & & \push{\vdots}\gategroup[7, steps=1, style={name=ccxboxdag, inner sep=1pt}]{} & & \\
            & & & & \wire[u,shorten >=0.1cm][1]{a} & \dots & \ctrl{1} & \dots & \wire[u,shorten >=0.1cm][1]{a} & & & \ctrl{1} & & \\
            & & & & & \dots & \ctrl{1} & \dots & & & & \ctrl{1} & & \\
            & & & \ctrl{6} & & \dots & \targ{} & \dots & & \ctrl{6} & & \targ{} & & \\
            & & & & & \dots & \ctrl{1} & \dots & & & & \ctrl{1} & & \\
            & & & & & \dots & \ctrl{1} & \dots & & & & \ctrl{1} & & \\
            & & \ctrl{3} & & & \dots & \targ{} & \dots & & & \ctrl{3} & \targ{} & & \\
            & & & & & \dots & \ctrl{1} & \dots & & & & & & \\
            & & & & & \dots & \ctrl{1} & \dots & & & & & & \\
            & & \targ{} & \targ{} & \targ{}\vqw{-8} & \dots & \targ{} & \dots & \targ{}\vqw{-8} & \targ{} & \targ{} & & & \\
            & & & & & & |[alias=ctrl_first, fill, circle, minimum size=4pt, inner sep=0pt]|{} & \dots & & & & |[alias=ctrl_second, fill, circle, minimum size=4pt, inner sep=0pt]|{} & & 
        \end{quantikz}
    \end{equation}
    The two sets of $CX$ gates in Equation~\eqref{eq:controlled_fanin} sharing the same target correspond to fan-in gates, and can be implemented with $\mathcal{O}(n)$ gates and $\mathcal{O}(\log n)$ circuit depth (Lemma~\ref{lem:fanout}).
    Similarly, the two sets of $CCX$ gates in Equation~\eqref{eq:controlled_fanin} all controlled by the same control qubit correspond to the second-order fan-out operator and can also be implemented with $\mathcal{O}(n)$ gates and $\mathcal{O}(\log n)$ circuit depth (Lemma~\ref{lem:fanout}).
    Finally, by Corollary~\ref{cor:promise_ladder}, each promise gate with target unitary $\mathcal{L}_2$ or $\mathcal{L}_2^\dag$ can be implemented with $\mathcal{O}(\sqrt{n})$ gates and $\mathcal{O}(\log n)$ circuit depth.
    Thus, the total gate count and circuit depth of the circuit of Figure~\ref{fig:ctrl_promise_v_sqrt} are $\mathcal{O}(n)$ and $\mathcal{O}(\log n)$, respectively.
\end{proof}

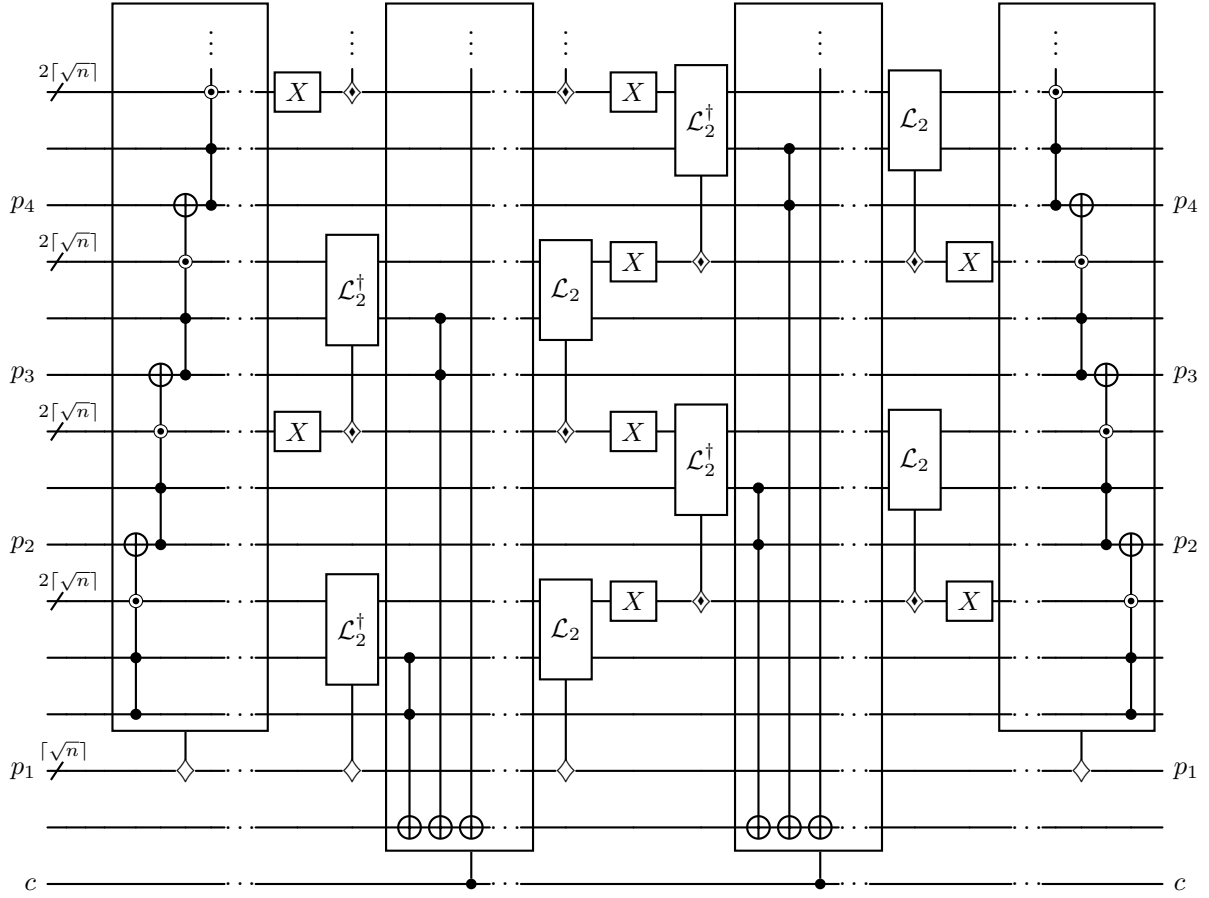
\begin{figure}[t]
    \makebox[\textwidth][c]{%
    \begin{quantikz}[row sep={\the\circuitrowsep,between origins}, column sep=\the\circuitcolsep, align equals at=8.5, execute at end picture={
        \draw [line width=0.8pt] (diamondccx.north) -- (ccxbox.south -| diamondccx.north);
        \draw [line width=0.8pt] (diamondccxdag.north) -- (ccxboxdag.south -| diamondccxdag.north);
        \draw [line width=0.8pt] (ctrl_first.north) -- (fanout_first.south -| ctrl_first.north);
        \draw [line width=0.8pt] (ctrl_second.north) -- (fanout_second.south -| ctrl_second.north);
    }]
        \setwiretype{n} & & & & \gategroup[13, steps=5, style={name=ccxbox, inner sep=1pt}]{} &[-0.25cm] &[-0.25cm] &[-0.17cm] \push{\vdots} &[-0.17cm] & &[-0.17cm] \push{\vdots} & \gategroup[15, steps=4, style={name=fanout_first, inner sep=1pt}]{} &[-0.17cm] &[-0.17cm] \push{\vdots} &[-0.17cm] & \push{\vdots} & & & \gategroup[15, steps=4, style={name=fanout_second, inner sep=1pt}]{} &[-0.17cm] &[-0.17cm] \push{\vdots} &[-0.17cm] & &[-0.17cm] & \gategroup[13, steps=5, style={name=ccxboxdag, inner sep=1pt}]{} &[-0.17cm] \push{\vdots} &[-0.17cm] &[-0.25cm] &[-0.25cm] & \\
        & \qwbundle{\!\!\!\!\!\!\!2\lceil\sqrtsign{n}\rceil} & & & & & & \push{\partctrl}\wire[u,shorten >=0.18cm][1]{a} & \dots & \gate{X} & \push{\partdiamondsuit}\wire[u,shorten >=0.18cm][1]{a} & & & \wire[u,shorten >=0.18cm][1]{a} & \dots & \push{\partdiamondsuit}\wire[u,shorten >=0.18cm][1]{a} & \gate{X} & \gate[2]{\mathcal{L}^\dag_2} & & & \wire[u,shorten >=0.18cm][1]{a} & \dots & \gate[2]{\mathcal{L}_2} & & \dots & \push{\partctrl}\wire[u,shorten >=0.18cm][1]{a} & & & & \\
        & & & & & & & \ctrl{-1} & \dots & & & & & & \dots & & & & & \ctrl{1} & & \dots & & & \dots & \ctrl{-1} & & & & \\
        \lstick{$p_4$} & & & & & & \targ{} & \ctrl{-1} & \dots & & & & & & \dots & & & & & \ctrl{11} & & \dots & & & \dots & \ctrl{-1} & \targ{} & & & \rstick{$p_4$} \\
        & \qwbundle{\!\!\!\!\!\!\!2\lceil\sqrtsign{n}\rceil} & & & & & \push{\partctrl}\vqw{-1} & & \dots & & \gate[2]{\mathcal{L}^\dag_2} & & & & \dots & \gate[2]{\mathcal{L}_2} & \gate{X} & \push{\partdiamondsuit}\vqw{-2} & & & & \dots & \push{\partdiamondsuit}\vqw{-2} & \gate{X} & \dots & & \push{\partctrl}\vqw{-1} & & & \\
        & & & & & & \ctrl{-1} & & \dots & & & & \ctrl{1} & & \dots & & & & & & & \dots & & & \dots & & \ctrl{-1} & & & \\
        \lstick{$p_3$} & & & & & \targ{} & \ctrl{-1} & & \dots & & & & \ctrl{8} & & \dots & & & & & & & \dots & & & \dots & & \ctrl{-1} & \targ{} & & \rstick{$p_3$} \\
        & \qwbundle{\!\!\!\!\!\!\!2\lceil\sqrtsign{n}\rceil} & & & & \push{\partctrl}\vqw{-1} & & & \dots & \gate{X} & \push{\partdiamondsuit}\vqw{-2} & & & & \dots & \push{\partdiamondsuit}\vqw{-2} & \gate{X} & \gate[2]{\mathcal{L}^\dag_2} & & & & \dots & \gate[2]{\mathcal{L}_2} & & \dots & & & \push{\partctrl}\vqw{-1} & & \\
        & & & & & \ctrl{-1} & & & \dots & & & & & & \dots & & & & \ctrl{1} & & & \dots & & & \dots & & & \ctrl{-1} & & \\
        \lstick{$p_2$} & & & & \targ{} & \ctrl{-1} & & & \dots & & & & & & \dots & & & & \ctrl{5} & & & \dots & & & \dots & & & \ctrl{-1} & \targ{} & \rstick{$p_2$} \\
        & \qwbundle{\!\!\!\!\!\!\!2\lceil\sqrtsign{n}\rceil} & & & \push{\partctrl}\vqw{-1} & & & & \dots & & \gate[2]{\mathcal{L}^\dag_2} & & & & \dots & \gate[2]{\mathcal{L}_2} & \gate{X} & \push{\partdiamondsuit}\vqw{-2} & & & & \dots & \push{\partdiamondsuit}\vqw{-2} & \gate{X} & \dots & & & & \push{\partctrl}\vqw{-1} & \\
        & & & & \ctrl{-1} & & & & \dots & & & \ctrl{1} & & & \dots & & & & & & & \dots & & & \dots & & & & \ctrl{-1} & \\
        & & & & \ctrl{-1} & & & & \dots & & & \ctrl{2} & & & \dots & & & & & & & \dots & & & \dots & & & & \ctrl{-1} & \\
        \lstick{$p_1$} & \qwbundle{\!\!\!\!\!\!\!\lceil\sqrtsign{n}\rceil} & & & & & |[alias=diamondccx]| \push{\diamondsuit} & & \dots & & \push{\diamondsuit}\vqw{-2} & & & & \dots & \push{\diamondsuit}\vqw{-2} & & & & & & \dots & & & \dots & & |[alias=diamondccxdag]| \push{\diamondsuit} & & & \rstick{$p_1$} \\
        & & & & & & & & \dots & & & \targ{} & \targ{} & \targ{}\vqw{-13} & \dots & & & & \targ{} & \targ{} & \targ{}\vqw{-13} & \dots & & & \dots & & & & & \\
        \lstick{$c$} & & & & & & & & \dots & & & & & |[alias=ctrl_first, fill, circle, minimum size=4pt, inner sep=0pt]|{} & \dots & & & & & & |[alias=ctrl_second, fill, circle, minimum size=4pt, inner sep=0pt]|{} & \dots & & & \dots & & & & & \rstick{$c$}
    \end{quantikz}}
    \caption{Quantum circuit for a controlled strong promise gate whose target unitary is the $\mathcal{V}_2^{(n)}$ operator. $c$ denotes the control qubit and $\bs{p}$ denotes the promise register of size $2\lceil\sqrt{n}\rceil$.
    The $\partctrl$ and $\partdiamondsuit$ symbols represent $\bullet$ and $\diamondsuit$ on alternating qubits, as in Equation~\eqref{eq:alternating_ctrl}.}\label{fig:ctrl_promise_v_sqrt}
\end{figure}

\begin{theorem}\label{thm:comparator}
    The $n$-bit quantum--quantum comparator can be implemented over the $\{CCX$, $CX$, $X\}$ gate set with $\Theta(n)$ gates and $\Theta(\log n)$ circuit depth, without any ancilla qubits.
\end{theorem}
\begin{proof}
    The two layers of parallel $X$ and $CX$ gates in slices 1 and 3 of Figure~\ref{fig:comparator} contribute $\mathcal{O}(n)$ gates with constant depth.
    The two ladders of $CX$ gates in these slices can each be implemented with $\mathcal{O}(n)$ $CX$ gates and $\mathcal{O}(\log n)$ circuit depth (Lemma~\ref{lem:ladder_cnot}).
    It remains to evaluate the cost of implementing the $\mathcal{V}_2^{(n)}$ operator in slice 2 of Figure~\ref{fig:comparator}.

    Based on the circuit identity of Figure~\ref{fig:promise_controlled_conjugation_dirty} and Equation~\eqref{eq:v_ckx_decomposition}, we have:
    
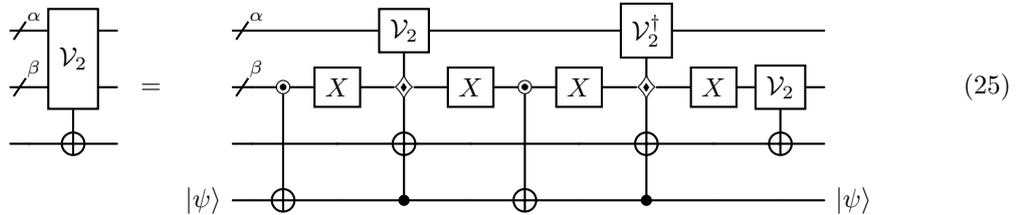
\begin{equation}
        \begin{quantikz}[row sep={\the\circuitrowsep,between origins}, column sep=\the\circuitcolsep, align equals at=2]
            & \qwbundle{\alpha} & \gate[2]{\mathcal{V}_2} & \\
            & \qwbundle{\beta} & & \\ 
            & & \targ{}\vqw{-1} &
        \end{quantikz}
        \;=\;
        \begin{quantikz}[row sep={\the\circuitrowsep,between origins}, column sep=\the\circuitcolsep, align equals at=2]
            & \qwbundle{\alpha} & & & \gate{\mathcal{V}_2} & & & & \gate{\mathcal{V}^\dag_2} & & & \\
            & \qwbundle{\beta} & \push{\partctrl}\vqw{2} & \gate{X} & \push{\partdiamondsuit}\vqw{-1} &  \gate{X} & \push{\partctrl}\vqw{2} & \gate{X} & \push{\partdiamondsuit}\vqw{-1} &  \gate{X} & \gate{\mathcal{V}_2}\vqw{1} & \\
            & & & & \targ{}\vqw{-1} & & & & \targ{}\vqw{-1} & & \targ{} & \\
            \lstick{\ket{\psi}} & & \targ{} & & \ctrl{-1} & & \targ{} & & \ctrl{-1} & & & \rstick{\ket{\psi}} 
        \end{quantikz}
    \end{equation}
    where the dirty ancilla can be taken from the register of size $\beta$, and where the $\partctrl$ and $\partdiamondsuit$ symbols represent $\bullet$ and $\diamondsuit$ on alternating qubits, as in Equation~\eqref{eq:alternating_ctrl}.

    Let $\beta = 2\lceil\sqrt{n}\rceil$ and $\alpha = 2n - \beta$.
    The two $C^{\lceil\sqrt{n}\rceil} X$ gates can be implemented with $\mathcal{O}(\sqrt{n})$ gates and $\mathcal{O}(\log n)$ circuit depth (Lemma~\ref{lem:optimal_ckx}).
    By Lemma~\ref{lem:promise_v_sqrt_ancillae}, the two controlled strong promise gates whose target unitary is the $\mathcal{V}_2$ operator can each be implemented with $\mathcal{O}(n)$ $\{CCX, CX, X\}$ gates and $\mathcal{O}(\log n)$ circuit depth.
    We then implement the $\mathcal{V}_2$ operator on the bottom $\beta$ qubits recursively.
    The total gate count satisfies
    \begin{equation}
        C(n) = \Theta(n) + C(2\lceil\sqrt{n}\rceil)
    \end{equation}
    and the total circuit depth satisfies
    \begin{equation}
        D(n) = \Theta(\log n) + D(2\lceil\sqrt{n}\rceil)
    \end{equation}
    which yield $C(n) = \Theta(n)$ and $D(n) = \Theta(\log n)$, respectively.
\end{proof}

As a direct consequence of Theorem~\ref{thm:comparator}, we obtain an efficient construction for the $k$-controlled comparator.
\begin{corollary}\label{cor:controlled_comparator}
    The $k$-controlled $n$-bit (where $n>1$) quantum--quantum comparator can be implemented over the $\{CCX$, $CX$, $X\}$ gate set with $\Theta(k + n)$ gates and $\Theta(\log(kn))$ circuit depth, without any ancilla qubits.
\end{corollary}
\begin{proof}
    By the controlled conjugation identity of Figure~\ref{fig:controlled_conjugation}, when adding $k$ control qubits to the comparator circuit of Figure~\ref{fig:comparator}, only a single $CX$ gate in slice 1 and the $\mathcal{V}_2^{(n)}$ operator in slice 2 need to be controlled.

    The $k$-controlled $\mathcal{V}_2$ operator can be implemented as follows:
    
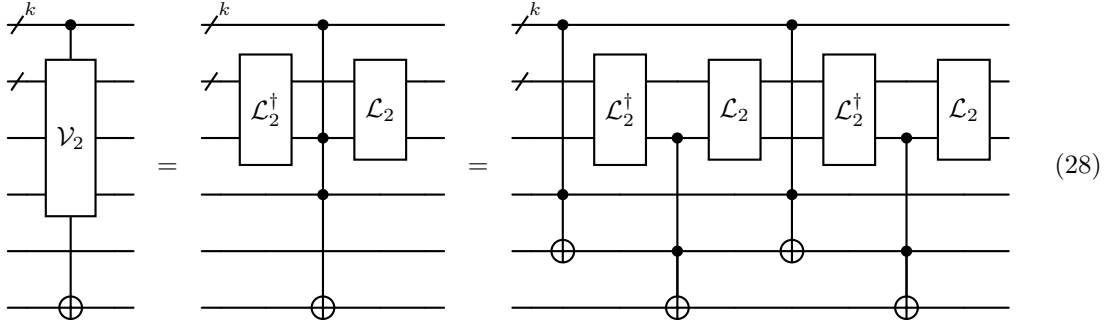
\begin{equation}
        \begin{quantikz}[row sep={\the\circuitrowsep,between origins}, column sep=\the\circuitcolsep, align equals at=3.5]
            & \qwbundle{k} & \ctrl{1} & & \\
            & \qwbundle{} & \gate[3]{\mathcal{V}_2}\vqw{4} & & \\
            & & & & \\
            & & & & \\
            & & & & \\
            & & \targ{} & & 
        \end{quantikz}
        \;=\;
        \begin{quantikz}[row sep={\the\circuitrowsep,between origins}, column sep=\the\circuitcolsep, align equals at=3.5]
            & \qwbundle{k} & & \ctrl{2} & & & \\
            & \qwbundle{} & \gate[2]{\mathcal{L}^\dag_2} & & \gate[2]{\mathcal{L}_2} & & \\
            & & & \ctrl{1} & & & \\
            & & & \ctrl{2} & & & \\
            & & & & & & \\
            & & & \targ{} & & & 
        \end{quantikz}
        \;=\;
        \begin{quantikz}[row sep={\the\circuitrowsep,between origins}, column sep=\the\circuitcolsep, align equals at=3.5]
            & \qwbundle{k} & \ctrl{3} & & & & \ctrl{3} & & & & \\
            & \qwbundle{} & & \gate[2]{\mathcal{L}^\dag_2} & & \gate[2]{\mathcal{L}_2} & & \gate[2]{\mathcal{L}^\dag_2} & & \gate[2]{\mathcal{L}_2} & \\
            & & & & \ctrl{3} & & & & \ctrl{3} & & \\
            & & \ctrl{1} & & & & \ctrl{1} & & & & \\
            & & \targ{} & & \ctrl{1} & & \targ{} & & \ctrl{1} & & \\
            & & & & \targ{} & & & & \targ{} & & 
        \end{quantikz}
    \end{equation}
    using one dirty ancilla qubit, which is available since $n>1$.
    The $C^{k+1}X$ gates can be implemented with $\Theta(k)$ gates and $\Theta(\log k)$ circuit depth (Lemma~\ref{lem:optimal_ckx}).
    Each pair of $\mathcal{L}_2$ and $\mathcal{L}^\dag_2$ operators with a $CCX$ gate in the middle corresponds to a $\mathcal{V}_2$ operator, which can be implemented with $\Theta(n)$ gates and $\Theta(\log n)$ circuit depth by Theorem~\ref{thm:comparator}.
\end{proof}

\subsection{Classical--quantum comparator} \label{sub:comparator_classical_quantum}
We now turn to the classical--quantum comparator and prove the analogous optimal result.
An $n$-bit classical--quantum comparator, parameterized by a classical constant $c \in \{0,1\}^n$, computes the following map:
\begin{equation}
    \ket{\bs a}\ket{z} \mapsto \ket{\bs a}\ket{z \oplus (c < a)},
\end{equation}
where $\bs a \in \{0, 1\}^n$ and $z \in \{0, 1\}$.
An example of a classical--quantum comparator is shown in Figure~\ref{fig:classical_quantum_comparator}.
We show that this comparator can be implemented with the same asymptotically optimal gate count and circuit depth as the quantum--quantum comparator, using one dirty ancilla.

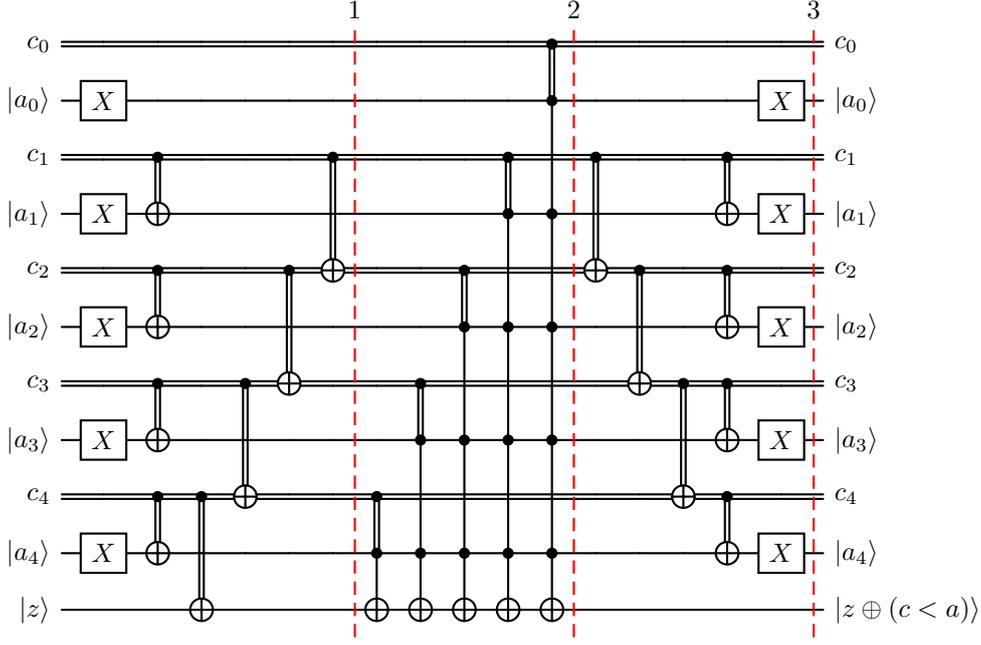
\begin{figure}[t]
    \makebox[\textwidth][c]{%
    \begin{quantikz}[row sep={\the\circuitrowsep,between origins}, column sep=\the\circuitcolsep]
        \setwiretype{c} \lstick{$c_0$} & & & & & & \slice{1} & & & & & \cctrl{1}\slice{2} & & & & & \slice{3} & \rstick{$c_0$} \\
        \lstick{\ket{a_0}} & \gate{X} & & & & & & & & & & \ctrl{2} & & & & & \gate{X} & \rstick{\ket{a_0}} \\
        \setwiretype{c} \lstick{$c_1$} & & \cctrl{1} & & & & \cctrl{2} & & & & \cctrl{1} & & \cctrl{2} & & & \cctrl{1} & & \rstick{$c_1$} \\
        \lstick{\ket{a_1}} & \gate{X} & \targ{} & & & & & & & & \ctrl{2} & \ctrl{2} & & & & \targ{} & \gate{X} & \rstick{\ket{a_1}} \\
        \setwiretype{c} \lstick{$c_2$} & & \cctrl{1} & & & \cctrl{2} & \targ{} & & & \cctrl{1} & & & \targ{} & \cctrl{2} & & \cctrl{1} & & \rstick{$c_2$} \\
        \lstick{\ket{a_2}} & \gate{X} & \targ{} & & & & & & & \ctrl{2} & \ctrl{2} & \ctrl{2} & & & & \targ{} & \gate{X} & \rstick{\ket{a_2}} \\
        \setwiretype{c} \lstick{$c_3$} & & \cctrl{1} & & \cctrl{2} & \targ{} & & & \cctrl{1} & & & & & \targ{} & \cctrl{2} & \cctrl{1} & & \rstick{$c_3$} \\
        \lstick{\ket{a_3}} & \gate{X} & \targ{} & & & & & & \ctrl{2} & \ctrl{2} & \ctrl{2} & \ctrl{2} & & & & \targ{} & \gate{X} & \rstick{\ket{a_3}} \\
        \setwiretype{c} \lstick{$c_4$} & & \cctrl{1} & \cctrl{2} & \targ{} & & & \cctrl{1} & & & & & & & \targ{} & \cctrl{1} & & \rstick{$c_4$} \\
        \lstick{\ket{a_4}} & \gate{X} & \targ{} & & & & & \ctrl{1} & \ctrl{1} & \ctrl{1} & \ctrl{1} & \ctrl{1} & & & & \targ{} & \gate{X} & \rstick{\ket{a_4}} \\
        \lstick{\ket{z}} & & & \targ{} & & & & \targ{} & \targ{} & \targ{} & \targ{} & \targ{} & & & & & & \rstick{\ket{z \oplus (c < a)}} 
    \end{quantikz}
    }
    \caption{Circuit for the classical--quantum comparator over $5$-bit registers.}\label{fig:classical_quantum_comparator}
\end{figure}

\begin{theorem}\label{thm:classical_quantum_comparator}
    The $n$-bit classical--quantum comparator can be implemented over the $\{CCX, CX, X\}$ gate set with $\Theta(n)$ gates and $\Theta(\log n)$ circuit depth, using one dirty ancilla qubit.
\end{theorem}
\begin{proof}
    Consider the classical--quantum comparator circuit given in Figure~\ref{fig:classical_quantum_comparator}.
    The subcircuit in slice 2 can be implemented as follows, using $\mathcal{V}_2^{(n)}$ operators and $n$ dirty ancilla qubits:
    
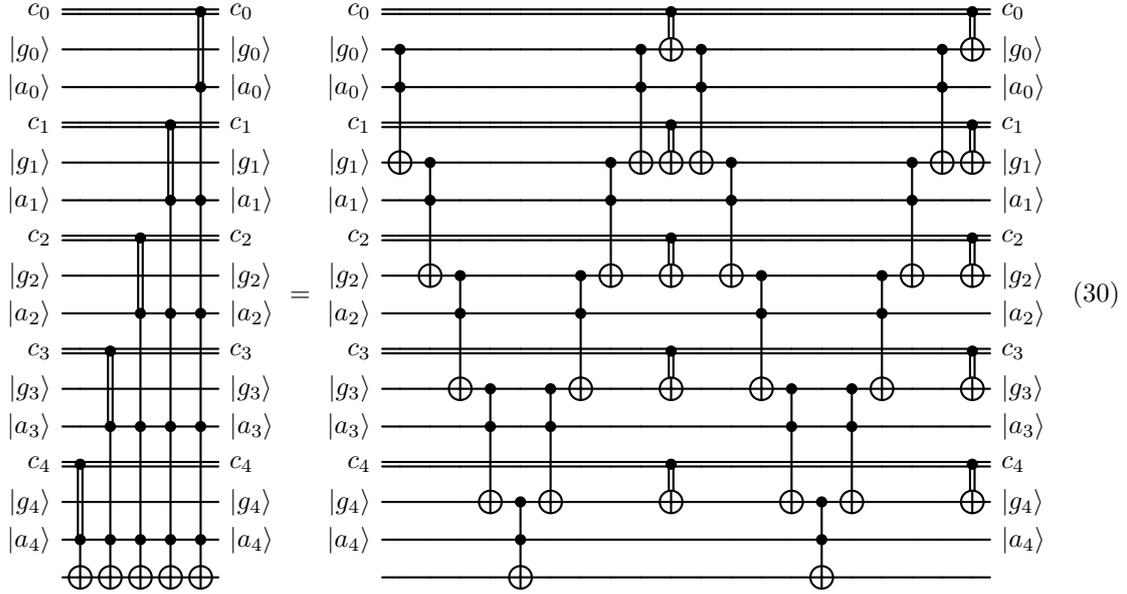
\begin{equation}
    \begin{quantikz}[row sep={.5cm,between origins}, column sep=0.07cm, align equals at=8.5]
        \setwiretype{c} \lstick{$c_0$} & & & & & \cctrl{2} & \rstick{$c_0$} \\
        \lstick{$\ket{g_0}$} & & & & & & \rstick{$\ket{g_0}$} \\
        \lstick{$\ket{a_0}$} & & & & & \ctrl{3} & \rstick{$\ket{a_0}$} \\
        \setwiretype{c} \lstick{$c_1$} & & & & \cctrl{2} & & \rstick{$c_1$} \\
        \lstick{$\ket{g_1}$} & & & & & & \rstick{$\ket{g_1}$} \\
        \lstick{$\ket{a_1}$} & & & & \ctrl{3} & \ctrl{3} & \rstick{$\ket{a_1}$} \\
        \setwiretype{c} \lstick{$c_2$} & & & \cctrl{2} & & & \rstick{$c_2$} \\
        \lstick{$\ket{g_2}$} & & & & & & \rstick{$\ket{g_2}$} \\
        \lstick{$\ket{a_2}$} & & & \ctrl{3} & \ctrl{3} & \ctrl{3} & \rstick{$\ket{a_2}$} \\
        \setwiretype{c} \lstick{$c_3$} & & \cctrl{2} & & & & \rstick{$c_3$} \\
        \lstick{$\ket{g_3}$} & & & & & & \rstick{$\ket{g_3}$} \\
        \lstick{$\ket{a_3}$} & & \ctrl{3} & \ctrl{3} & \ctrl{3} & \ctrl{3} & \rstick{$\ket{a_3}$} \\
        \setwiretype{c} \lstick{$c_4$} & \cctrl{2} & & & & & \rstick{$c_4$} \\
        \lstick{$\ket{g_4}$} & & & & & & \rstick{$\ket{g_4}$} \\
        \lstick{$\ket{a_4}$} & \ctrl{1} & \ctrl{1} & \ctrl{1} & \ctrl{1} & \ctrl{1} & \rstick{$\ket{a_4}$} \\
        & \targ{} & \targ{} & \targ{} & \targ{} & \targ{} & 
    \end{quantikz}
    =
    \begin{quantikz}[row sep={.5cm,between origins}, column sep=0.07cm, align equals at=8.5]
        \setwiretype{c} \lstick{$c_0$} & & & & & & & & & & \cctrl{1} & & & & & & & & & & \cctrl{1} & \rstick{$c_0$} \\
        \lstick{$\ket{g_0}$} & \ctrl{1} & & & & & & & & \ctrl{1} & \targ{} & \ctrl{1} & & & & & & & & \ctrl{1} & \targ{} & \rstick{$\ket{g_0}$} \\
        \lstick{$\ket{a_0}$} & \ctrl{2} & & & & & & & & \ctrl{2} & & \ctrl{2} & & & & & & & & \ctrl{2} & & \rstick{$\ket{a_0}$} \\
        \setwiretype{c} \lstick{$c_1$} & & & & & & & & & & \cctrl{1} & & & & & & & & & & \cctrl{1} & \rstick{$c_1$} \\
        \lstick{$\ket{g_1}$} & \targ{} & \ctrl{1} & & & & & & \ctrl{1} & \targ{} & \targ{} & \targ{} & \ctrl{1} & & & & & & \ctrl{1} & \targ{} & \targ{} & \rstick{$\ket{g_1}$} \\
        \lstick{$\ket{a_1}$} & & \ctrl{2} & & & & & & \ctrl{2} & & & & \ctrl{2} & & & & & & \ctrl{2} & & & \rstick{$\ket{a_1}$} \\
        \setwiretype{c} \lstick{$c_2$} & & & & & & & & & & \cctrl{1} & & & & & & & & & & \cctrl{1} & \rstick{$c_2$} \\
        \lstick{$\ket{g_2}$} & & \targ{} & \ctrl{1} & & & & \ctrl{1} & \targ{} & & \targ{} & & \targ{} & \ctrl{1} & & & & \ctrl{1} & \targ{} & & \targ{} & \rstick{$\ket{g_2}$} \\
        \lstick{$\ket{a_2}$} & & & \ctrl{2} & & & & \ctrl{2} & & & & & & \ctrl{2} & & & & \ctrl{2} & & & & \rstick{$\ket{a_2}$} \\
        \setwiretype{c} \lstick{$c_3$} & & & & & & & & & & \cctrl{1} & & & & & & & & & & \cctrl{1} & \rstick{$c_3$} \\
        \lstick{$\ket{g_3}$} & & & \targ{} & \ctrl{1} & & \ctrl{1} & \targ{} & & & \targ{} & & & \targ{} & \ctrl{1} & & \ctrl{1} & \targ{} & & & \targ{} & \rstick{$\ket{g_3}$} \\
        \lstick{$\ket{a_3}$} & & & & \ctrl{2} & & \ctrl{2} & & & & & & & & \ctrl{2} & & \ctrl{2} & & & & & \rstick{$\ket{a_3}$} \\
        \setwiretype{c} \lstick{$c_4$} & & & & & & & & & & \cctrl{1} & & & & & & & & & & \cctrl{1} & \rstick{$c_4$} \\
        \lstick{$\ket{g_4}$} & & & & \targ{} & \ctrl{1} & \targ{} & & & & \targ{} & & & & \targ{} & \ctrl{1} & \targ{} & & & & \targ{} & \rstick{$\ket{g_4}$} \\
        \lstick{$\ket{a_4}$} & & & & & \ctrl{1} & & & & & & & & & & \ctrl{1} & & & & & & \rstick{$\ket{a_4}$} \\
        & & & & & \targ{} & & & & & & & & & & \targ{} & & & & & & 
    \end{quantikz}
    \end{equation}
    By introducing one clean ancilla qubit, this circuit can be split into two parts, with two $\mathcal{V}^{(\lceil\frac{n}{2}\rceil)}_2$ operators in the first part and two $\mathcal{V}^{(\lfloor\frac{n}{2}\rfloor)}_2$ operators in the second part:
    
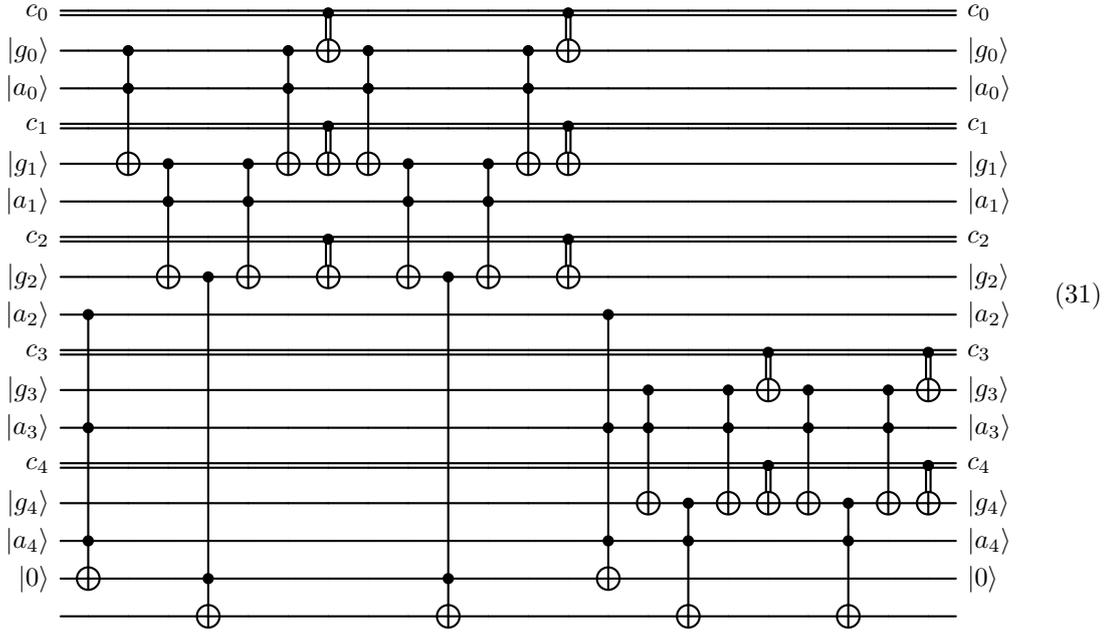
\begin{equation}\label{eq:v_split_example}
    \begin{quantikz}[row sep={.5cm,between origins}, column sep=0.2cm, align equals at=8.5]
        \setwiretype{c} \lstick{$c_0$} & & & & & & & \cctrl{1} & & & & & & \cctrl{1} & & & & & & & & & & \rstick{$c_0$} \\
        \lstick{$\ket{g_0}$} & & \ctrl{1} & & & & \ctrl{1} & \targ{} & \ctrl{1} & & & & \ctrl{1} & \targ{} & & & & & & & & & & \rstick{$\ket{g_0}$} \\
        \lstick{$\ket{a_0}$} & & \ctrl{2} & & & & \ctrl{2} & & \ctrl{2} & & & & \ctrl{2} & & & & & & & & & & & \rstick{$\ket{a_0}$} \\
        \setwiretype{c} \lstick{$c_1$} & & & & & & & \cctrl{1} & & & & & & \cctrl{1} & & & & & & & & & & \rstick{$c_1$} \\
        \lstick{$\ket{g_1}$} & & \targ{} & \ctrl{1} & & \ctrl{1} & \targ{} & \targ{} & \targ{} & \ctrl{1} & & \ctrl{1} & \targ{} & \targ{} & & & & & & & & & & \rstick{$\ket{g_1}$} \\
        \lstick{$\ket{a_1}$} & & & \ctrl{2} & & \ctrl{2} & & & & \ctrl{2} & & \ctrl{2} & & & & & & & & & & & & \rstick{$\ket{a_1}$} \\
        \setwiretype{c} \lstick{$c_2$} & & & & & & & \cctrl{1} & & & & & & \cctrl{1} & & & & & & & & & & \rstick{$c_2$} \\
        \lstick{$\ket{g_2}$} & & & \targ{} & \ctrl{8} & \targ{} & & \targ{} & & \targ{} & \ctrl{8} & \targ{} & & \targ{} & & & & & & & & & & \rstick{$\ket{g_2}$} \\
        \lstick{$\ket{a_2}$} & \ctrl{3} & & & & & & & & & & & & & \ctrl{3} & & & & & & & & & \rstick{$\ket{a_2}$} \\
        \setwiretype{c} \lstick{$c_3$} & & & & & & & & & & & & & & & & & & \cctrl{1} & & & & \cctrl{1} & \rstick{$c_3$} \\
        \lstick{$\ket{g_3}$} & & & & & & & & & & & & & & & \ctrl{1} & & \ctrl{1} & \targ{} & \ctrl{1} & & \ctrl{1} & \targ{} & \rstick{$\ket{g_3}$} \\
        \lstick{$\ket{a_3}$} & \ctrl{3} & & & & & & & & & & & & & \ctrl{3} & \ctrl{2} & & \ctrl{2} & & \ctrl{2} & & \ctrl{2} & & \rstick{$\ket{a_3}$} \\
        \setwiretype{c} \lstick{$c_4$} & & & & & & & & & & & & & & & & & & \cctrl{1} & & & & \cctrl{1} & \rstick{$c_4$} \\
        \lstick{$\ket{g_4}$} & & & & & & & & & & & & & & & \targ{} & \ctrl{1} & \targ{} & \targ{} & \targ{} & \ctrl{1} & \targ{} & \targ{} & \rstick{$\ket{g_4}$} \\
        \lstick{$\ket{a_4}$} & \ctrl{1} & & & & & & & & & & & & & \ctrl{1} & & \ctrl{2} & & & & \ctrl{2} & & & \rstick{$\ket{a_4}$} \\
        \lstick{$\ket{0}$} & \targ{} & & & \ctrl{1} & & & & & & \ctrl{1} & & & & \targ{} & & & & & & & & & \rstick{$\ket{0}$} \\
        & & & & \targ{} & & & & & & \targ{} & & & & & & \targ{} & & & & \targ{} & & & 
    \end{quantikz}
    \end{equation}
    The first two $\mathcal{V}^{(\lceil\frac{n}{2}\rceil)}_2$ operators can be implemented using the bottom $\lceil\frac{n}{2}\rceil$ input qubits as dirty ancilla qubits.
    Similarly, the last two $\mathcal{V}^{(\lfloor\frac{n}{2}\rfloor)}_2$ operators can be implemented using the top $\lceil\frac{n}{2}\rceil$ input qubits as dirty ancilla qubits.
    For instance, for the circuit of Equation~\eqref{eq:v_split_example}, we can use $a_2$, $a_3$, and $a_4$ instead of $g_0$, $g_1$, and $g_2$, and we can use $a_0$ and $a_1$ instead of $g_3$ and $g_4$.
    We obtain the following circuit:
    
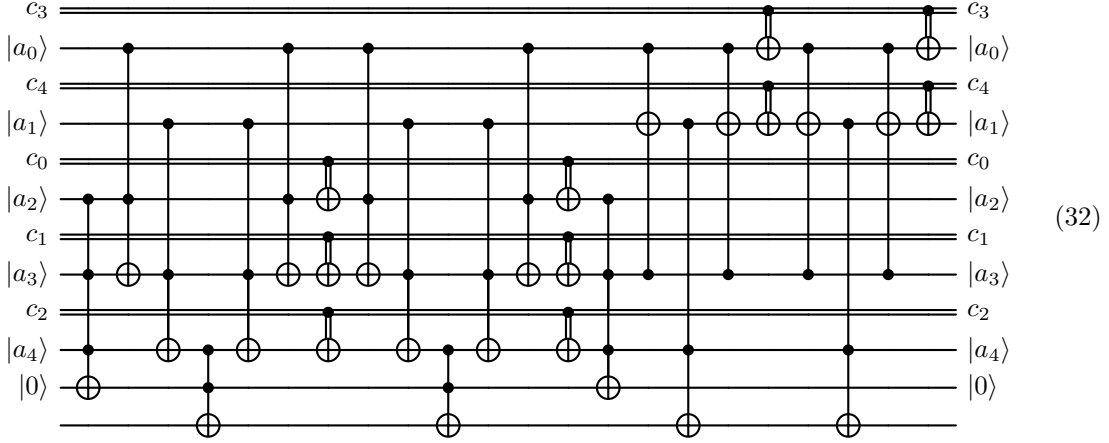
\begin{equation}
    \begin{quantikz}[row sep={.5cm,between origins}, column sep=0.2cm, align equals at=6.5]
        \setwiretype{c} \lstick{$c_3$} & & & & & & & & & & & & & & & & & & \cctrl{1} & & & & \cctrl{1} & \rstick{$c_3$} \\
        \lstick{$\ket{a_0}$} & & \ctrl{4} & & & & \ctrl{4} & & \ctrl{4} & & & & \ctrl{4} & & & \ctrl{2} & & \ctrl{2} & \targ{} & \ctrl{2} & & \ctrl{2} & \targ{}& \rstick{$\ket{a_0}$} \\
        \setwiretype{c} \lstick{$c_4$} & & & & & & & & & & & & & & & & & & \cctrl{1} & & & & \cctrl{1} & \rstick{$c_4$} \\
        \lstick{$\ket{a_1}$} & & & \ctrl{6} & & \ctrl{6} & & & & \ctrl{6} & & \ctrl{6} & & & & \targ{} & \ctrl{6} & \targ{} & \targ{} & \targ{} & \ctrl{6} & \targ{} & \targ{} & \rstick{$\ket{a_1}$} \\
        \setwiretype{c} \lstick{$c_0$} & & & & & & & \cctrl{1} & & & & & & \cctrl{1} & & & & & & & & & & \rstick{$c_0$} \\
        \lstick{$\ket{a_2}$} & \ctrl{2} & \ctrl{2} & & & & \ctrl{2} & \targ{} & \ctrl{2} & & & & \ctrl{2} & \targ{} & \ctrl{3} & & & & & & & & & \rstick{$\ket{a_2}$} \\
        \setwiretype{c} \lstick{$c_1$} & & & & & & & \cctrl{1} & & & & & & \cctrl{1} & & & & & & & & & & \rstick{$c_1$} \\
        \lstick{$\ket{a_3}$} & \ctrl{2} & \targ{} & \ctrl{2} & & \ctrl{2} & \targ{} & \targ{} & \targ{} & \ctrl{2} & & \ctrl{2} & \targ{} & \targ{} & \ctrl{3} & \ctrl{-4} & & \ctrl{-4} & & \ctrl{-4} & & \ctrl{-4} & & \rstick{$\ket{a_3}$} \\
        \setwiretype{c} \lstick{$c_2$} & & & & & & & \cctrl{1} & & & & & & \cctrl{1} & & & & & & & & & & \rstick{$c_2$} \\
        \lstick{$\ket{a_4}$} & \ctrl{1} & & \targ{} & \ctrl{1} & \targ{} & & \targ{} & & \targ{} & \ctrl{1} & \targ{} & & \targ{} & \ctrl{1} & & \ctrl{2} & & & & \ctrl{2} & & & \rstick{$\ket{a_4}$} \\
        \lstick{$\ket{0}$} & \targ{} & & & \ctrl{1} & & & & & & \ctrl{1} & & & & \targ{} & & & & & & & & & \rstick{$\ket{0}$} \\
        & & & & \targ{} & & & & & & \targ{} & & & & & & \targ{} & & & & \targ{} & & & 
    \end{quantikz}
    \end{equation}
    All $\mathcal{V}_2$ operators can then be implemented over the $\{CCX, CX, X\}$ gate set with $\Theta(n)$ gates and $\Theta(\log n)$ circuit depth, as shown in Theorem~\ref{thm:comparator}.
    Finally, the clean ancilla can be replaced by a dirty ancilla using the circuit identity of Figure~\ref{fig:toggle_detection}.
\end{proof}

As a direct consequence of Theorem~\ref{thm:classical_quantum_comparator}, we obtain an efficient construction for the $k$-controlled classical--quantum comparator.
\begin{corollary}\label{cor:controlled_classical_quantum_comparator}
    The $k$-controlled $n$-bit classical--quantum comparator can be implemented over the $\{CCX$, $CX$, $X\}$ gate set with $\Theta(k + n)$ gates and $\Theta(\log(kn))$ circuit depth, using one dirty ancilla qubit.
\end{corollary}
\begin{proof}
    By the controlled conjugation identity of Figure~\ref{fig:controlled_conjugation}, when adding $k$ control qubits to the comparator circuit of Figure~\ref{fig:classical_quantum_comparator}, only a single classically controlled $X$ gate in slice 1 and the subcircuit in slice 2 need to be controlled.
    The classically controlled $X$ gate becomes a classically controlled $C^kX$ gate, implementable with $\Theta(k)$ gates and $\Theta(\log k)$ circuit depth (Lemma~\ref{lem:optimal_ckx}).
    Each $k$-controlled $\mathcal{V}_2$ operator arising in slice 2 can be implemented using the same decomposition as in the proof of Corollary~\ref{cor:controlled_comparator}, using $\Theta(k + n)$ gates and $\Theta(\log(kn))$ circuit depth.
\end{proof}

\section{Quantum circuits for incrementers}\label{sec:incrementer}
\begin{figure}[t]
    \begin{subfigure}[b]{0.39\textwidth}
        \centering
        \begin{quantikz}[row sep={\the\circuitrowsep,between origins}, column sep=\the\circuitcolsep]
            & \ctrl{1} & \ctrl{1} & \ctrl{1} & \ctrl{1} & \ctrl{1} & \gate{X} & \\
            & \ctrl{1} & \ctrl{1} & \ctrl{1} & \ctrl{1} & \targ{} & & \\
            & \ctrl{1} & \ctrl{1} & \ctrl{1} & \targ{} & & & \\
            & \ctrl{1} & \ctrl{1} & \targ{} & & & & \\
            & \ctrl{1} & \targ{} & & & & & \\
            & \targ{} & & & & & & 
        \end{quantikz}
    \caption{Naive implementation using $C^kX$ gates.}\label{fig:naive_incrementer}
    \end{subfigure}
    \begin{subfigure}[b]{0.59\textwidth}
        \centering
        \begin{quantikz}[row sep={\the\circuitrowsep,between origins}, column sep=\the\circuitcolsep]
            & \ctrl{1} & & & \slice{1} & \ctrl{1} & \slice{2} & & & & \ctrl{1}\slice{3} & \gate{X}\slice{4} & \\
            & \ctrl{1} & & & & \targ{} & \gate{X} & & & & \ctrl{1} & \gate{X} & \\
            \lstick{\ket{0}} & \targ{} & \ctrl{1} & & & \ctrl{1} & & & & \ctrl{1} & \targ{} & & \rstick{\ket{0}} \\
            & & \ctrl{1} & & & \targ{} & \gate{X} & & & \ctrl{1} & & \gate{X} & \\
            \lstick{\ket{0}} & & \targ{} & \ctrl{1} & & \ctrl{1} & & & \ctrl{1} & \targ{} & & & \rstick{\ket{0}} \\
            & & & \ctrl{1} & & \targ{} & \gate{X} & & \ctrl{1} & & & \gate{X} & \\
            \lstick{\ket{0}} & & & \targ{} & \ctrl{1} & \ctrl{1} & & \ctrl{1} & \targ{} & & & & \rstick{\ket{0}} \\
            & & & & \ctrl{1} & \targ{} & \gate{X} & \ctrl{1} & & & & \gate{X} & \\
            \lstick{\ket{0}} & & & & \targ{} & \ctrl{1} & & \targ{} & & & & & \rstick{\ket{0}} \\
            & & & & & \targ{} & & & & & & & 
        \end{quantikz}
    \caption{Over the $\{CCX, CX, X\}$ gate set using $n-2$ clean ancilla qubits~\cite{Gidney_incrementer}.}\label{fig:gidney_incrementer}
    \end{subfigure}
    \caption{Quantum circuits for the $6$-bit incrementer.}\label{fig:incrementer}
\end{figure}

An $n$-bit quantum incrementer computes the following map:
\begin{equation}
    \ket{\bs x} \mapsto \ket{(\bs x+1) \bmod 2^n},
\end{equation}
where $\bs x \in \{0, 1\}^n$.
An example of a naive implementation using $C^kX$ gates is shown in Figure~\ref{fig:naive_incrementer}.
A useful property of the increment operator is that it admits the following decomposition, which can be applied recursively to obtain the naive implementation of Figure~\ref{fig:naive_incrementer}:
\begin{equation}\label{eq:inc_recursive}
    \begin{quantikz}[row sep={\the\circuitrowsep,between origins}, column sep=\the\circuitcolsep, align equals at=1.5]
        & \qwbundle{k} & \gate[2]{+1} & \\
        & \qwbundle{n} & & & 
    \end{quantikz}
    \;=\;
    \begin{quantikz}[row sep={\the\circuitrowsep,between origins}, column sep=\the\circuitcolsep, align equals at=1.5]
        & \qwbundle{k} & \ctrl{1} & \gate{+1} & \\
        & \qwbundle{n} & \gate{+1} & & 
    \end{quantikz}
\end{equation}

The increment operator $U$ does not satisfy $U^2 = I$, which is problematic because we cannot directly apply the second circuit identity of Figure~\ref{fig:toggle_detection}, a key tool when working with dirty ancilla qubits.
However, the increment operator satisfies the following circuit identity:
\begin{equation}\label{eq:inc_x_dec}
    \begin{quantikz}[row sep={\the\circuitrowsep,between origins}, column sep=\the\circuitcolsep, align equals at=1]
        & \qwbundle{n} & \gate{+1} & \gate{X} & \gate{+1} & \gate{X} & 
    \end{quantikz}
    \;=\;
    \begin{quantikz}[row sep={\the\circuitrowsep,between origins}, column sep=\the\circuitcolsep, align equals at=1]
        & \qwbundle{n} & \gate{+1} & \gate{-1} &
    \end{quantikz}
    \;=\;
    \begin{quantikz}[row sep={\the\circuitrowsep,between origins}, column sep=\the\circuitcolsep, align equals at=1]
        & \qwbundle{n} & & &
    \end{quantikz}
\end{equation}
which means that instead of using Figure~\ref{fig:toggle_detection}, we can rely on the following circuit identity~\cite{Gidney_incrementer}:
\begin{equation}\label{eq:toggle_detection_inc}
    \begin{quantikz}[row sep={\the\circuitrowsep,between origins}, column sep=\the\circuitcolsep, align equals at=2]
        & \qwbundle{k} & \ctrl{2} & & \ctrl{2} & \\
        & \qwbundle{n} & & \gate{+1} & & \\
        \lstick{\ket{0}} & & \targ{} & \ctrl{-1} & \targ{} & \rstick{\ket{0}}
    \end{quantikz}
    \;=\;
    \begin{quantikz}[row sep={\the\circuitrowsep,between origins}, column sep=\the\circuitcolsep, align equals at=2]
        & \qwbundle{k} & & & & \ctrl{2} & & \ctrl{2} & \\
        & \qwbundle{n} & & \gate{+1} & & \targ{} & \gate{-1} & \targ{} & \\
        \lstick{\ket{\psi}} & & & \ctrl{-1} & & \targ{} & \ctrl{-1} & \targ{} & \rstick{\ket{\psi}}
    \end{quantikz}
\end{equation}
The $C^kX$ gates convert the decrement operation into an increment operation whenever the control register is in the $\ket{1}^{\otimes k}$ state.
When the control register is not in the $\ket{1}^{\otimes k}$ state, the $C^kX$ gates act as the identity, and the increment and decrement gates cancel each other if the dirty ancilla is in the $\ket{1}$ state; otherwise, they both act as the identity when the dirty ancilla is in the $\ket{0}$ state.
The $k$-controlled fan-out gates in Equation~\eqref{eq:toggle_detection_inc} can be implemented as follows:
\begin{equation}\label{eq:controlled_fan_out}
    \begin{quantikz}[row sep={\the\circuitrowsep,between origins}, column sep=\the\circuitcolsep, align equals at=3]
        & \qwbundle{k} & \ctrl{3} & \\
        & & \targ{} & \\
        & & \targ{} & \\
        \setwiretype{n} & & \vdots & \\
        & & \targ{}\wire[u,shorten >=0.18cm][1]{a} & 
    \end{quantikz}
    \;=\;
    \begin{quantikz}[row sep={\the\circuitrowsep,between origins}, column sep=\the\circuitcolsep, align equals at=3]
        & \qwbundle{k} & & \ctrl{1} & & \\
        & & \ctrl{2} & \targ{} & \ctrl{2} & \\
        & & \targ{} & & \targ{} & \\
        \setwiretype{n} & & \vdots & & \vdots & \\
        & & \targ{}\wire[u,shorten >=0.18cm][1]{a} & & \targ{}\wire[u,shorten >=0.18cm][1]{a} & 
    \end{quantikz}
\end{equation}
using two fan-out gates, each implementable with $\mathcal{O}(n)$ gates and $\mathcal{O}(\log n)$ circuit depth (Lemma~\ref{lem:fanout}), and one $C^kX$ gate, implementable with $\mathcal{O}(k)$ gates and $\mathcal{O}(\log k)$ circuit depth (Lemma~\ref{lem:optimal_ckx}).

We first state a series of lemmas for constructions that will be used in our main result: a quantum circuit implementing the increment operator with $\Theta(n)$ $\{CCX, CX, X\}$ gates, $\Theta(\log n)$ circuit depth, and one dirty ancilla qubit.

\begin{figure}[t]
    \centering
    \begin{quantikz}[row sep={\the\circuitrowsep,between origins}, column sep=\the\circuitcolsep, execute at end picture={
            \draw [line width=0.8pt] (topcontrol.south) -- (ccxbox.north -| topcontrol.south);
        }]
        \lstick{$c$} & & & & \slice{1} &[0.3cm] |[alias=topcontrol, fill, circle, minimum size=4pt, inner sep=0pt]|{} & \ctrl{8} \slice{2} & & & & \slice{3} & \ctrl{8} \slice{4} & \rstick{$c$} \\
        & \ctrl{1} & & & & \ctrl{1}\gategroup[10, steps=1, style={name=ccxbox, inner sep=1pt}]{} & & & & & \ctrl{1} & \targ{} & \\
        & \ctrl{1} & & & & \targ{} & \targ{} & & & & \ctrl{1} & \targ{} & \\
        \lstick{$p_0$} & \targ{} & \ctrl{1} & & & \ctrl{1} & & & & \ctrl{1} & \targ{} & & \rstick{$p_0$} \\
        & & \ctrl{1} & & & \targ{} & \targ{} & & & \ctrl{1} & & \targ{} & \\
        \lstick{$p_1$} & & \targ{} & \ctrl{1} & & \ctrl{1} & & & \ctrl{1} & \targ{} & & & \rstick{$p_1$} \\
        & & & \ctrl{1} & & \targ{} & \targ{} & & \ctrl{1} & & & \targ{} & \\
        \lstick{$p_2$} & & & \targ{} & \ctrl{1} & \ctrl{1} & & \ctrl{1} & \targ{} & & & & \rstick{$p_2$} \\
        & & & & \ctrl{1} & \targ{} & \targ{} & \ctrl{1} & & & & \targ{} & \\
        \lstick{$p_3$} & & & & \targ{} & \ctrl{1} & & \targ{} & & & & & \rstick{$p_3$} \\
        & & & & & \targ{} & & & & & & & 
    \end{quantikz}
    \caption{Quantum circuit for a controlled strong promise gate whose target unitary is the $6$-bit incrementer.
    Here $c$ denotes the control qubit and $\bs p$ denotes the promise register.
    The circuit is built from Figure~\ref{fig:gidney_incrementer} by replacing the clean ancilla qubits with a promise register $\bs{p}$ and adding a control to the gates in slices 2 and 4.}\label{fig:promise_incrementer}
\end{figure}
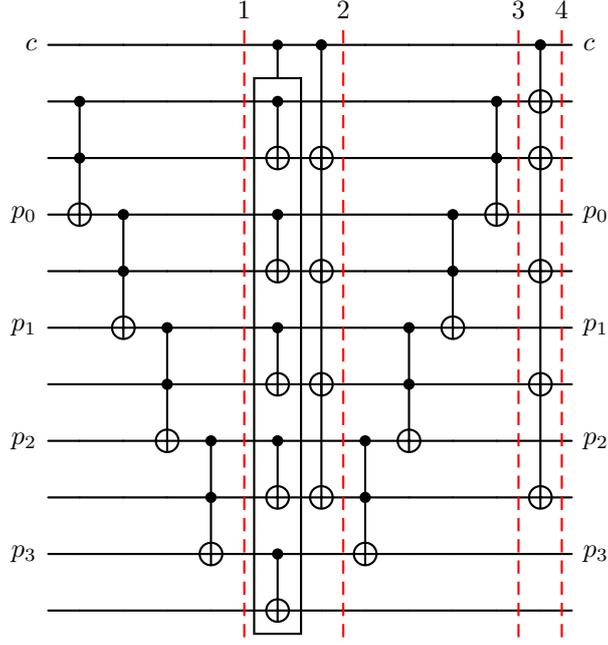

\begin{lemma}\label{lem:promise_inc_linear_ancillae}
    A $k$-controlled strong promise gate with a promise register of size $n-1$ whose target unitary is the $n$-bit incrementer can be implemented with $\mathcal{O}(k + n)$ $\{CCX, CX, X\}$ gates and $\mathcal{O}(\log(kn))$ circuit depth.
\end{lemma}
\begin{proof}
    We start from Gidney's circuit for implementing an $n$-bit incrementer with $\mathcal{O}(n)$ $\{CCX$, $CX$, $X\}$ gates, $\mathcal{O}(n)$ circuit depth, and $n-2$ clean ancilla qubits~\cite{Gidney_incrementer}.
    An example for $n=6$ is given in Figure~\ref{fig:gidney_incrementer}.
    Replacing the $n-2$ clean ancilla qubits with a promise register of the same size yields a strong promise gate whose target unitary is the increment operator.

    We then add a control to each gate to obtain a singly controlled promise gate with the incrementer as target unitary.
    By Figure~\ref{fig:controlled_conjugation}, the two ladders of $CCX$ gates in slices~1 and~3 do not need to be controlled.
    For example, for $n=6$ we obtain the circuit of Figure~\ref{fig:promise_incrementer}.
    All the other layers of gates controlled by the top qubit in slices 2 and 4 can be implemented by applying the construction of Lemma~\ref{lem:parallel_control_gates}, with a total of $\mathcal{O}(n)$ $\{CCX, CX, X\}$ gates and circuit depth $\mathcal{O}(\log n)$.

    We implement the two $CCX$ ladders using the construction of Lemma~\ref{lem:ladder_2_n_ancilla}, with $\mathcal{O}(n)$ $CCX$ gates and circuit depth $\mathcal{O}(\log n)$, using $n - 2 - \mathcal{O}(\log n)$ ancilla qubits from the promise register.
    This would require a promise register of size at least $2n - 4$.
    To use a promise register of smaller size, and to generalize to $k$ control qubits, we use the following circuit identity:
    \begin{equation}\label{eq:promise_inc_decomposition}
        \begin{quantikz}[row sep={\the\circuitrowsep,between origins}, column sep=\the\circuitcolsep, align equals at=3]
            \lstick{\ket{0}} & \qwbundle{n-1} & & \push{\diamondsuit}\wire[d][1]{a} & \rstick{\ket{0}} \\
            & \qwbundle{k} & & \ctrl{1} & \\
            & \qwbundle{\!\!\lceil \frac{n}{2} \rceil} & & \gate[2]{+1} & \\
            & \qwbundle{\!\!\lfloor \frac{n}{2} \rfloor} & & & \\
            \lstick{\ket{0}} & & & & \rstick{\ket{0}}
        \end{quantikz}
        \!\!=\!\!
        \begin{quantikz}[row sep={\the\circuitrowsep,between origins}, column sep=\the\circuitcolsep, align equals at=3]
            \lstick{\ket{0}} & \qwbundle{n-1} & & \push{\diamondsuit}\wire[d][1]{a} & \push{\diamondsuit}\wire[d][1]{a} & \rstick{\ket{0}} \\
            & \qwbundle{k} & & \ctrl{1} & \ctrl{1} & \\
            & \qwbundle{\!\!\lceil \frac{n}{2} \rceil} & & \ctrl{1} & \gate{+1} & \\
            & \qwbundle{\!\!\lfloor \frac{n}{2} \rfloor} & & \gate{+1} & & \\
            \lstick{\ket{0}} & & & & & \rstick{\ket{0}}
        \end{quantikz}
        \!\!=\!\!
        \begin{quantikz}[row sep={\the\circuitrowsep,between origins}, column sep=\the\circuitcolsep, align equals at=3]
            \lstick{\ket{0}} & \qwbundle{n-1} & & & \push{\diamondsuit}\wire[d][3]{a} & & & \push{\diamondsuit}\wire[d][2]{a} & & \rstick{\ket{0}} \\
            & \qwbundle{k} & & \ctrl{1} & & \ctrl{1} & \ctrl{3} & & \ctrl{3} & \\
            & \qwbundle{\!\!\lceil \frac{n}{2} \rceil} & & \ctrl{2} & & \ctrl{2} & & \gate{+1} & & \\
            & \qwbundle{\!\!\lfloor \frac{n}{2} \rfloor}  & & & \gate{+1} & & & & & \\
            \lstick{\ket{0}} & & & \targ{} & \ctrl{-1} & \targ{} & \targ{} & \ctrl{-2} & \targ{} & \rstick{\ket{0}}
        \end{quantikz}
    \end{equation}
    The two singly controlled promise gates whose target unitary is the incrementer can then be implemented as described above with a gate count of $\mathcal{O}(n)$ and circuit depth $\mathcal{O}(\log n)$, using $2\lceil \frac{n}{2} \rceil - 4 < n - 2$ ancilla qubits from the promise register.

    Finally, we exchange the clean ancilla qubit from Equation~\eqref{eq:promise_inc_decomposition} for a dirty ancilla (which can be taken from the promise register) using Equation~\eqref{eq:toggle_detection_inc}.
    Thus, the total gate count is $\mathcal{O}(k+n)$ and the circuit depth is $\mathcal{O}(\log(kn))$, using a promise register of size $n-1$.
\end{proof}

The following lemma relies on Lemma~\ref{lem:promise_inc_linear_ancillae} to further reduce the size of the promise register.
\begin{lemma}\label{lem:promise_inc_sqrt_ancillae}
    A controlled strong promise gate with a promise register of size $2\lceil\sqrt{n}\rceil$ whose target unitary is the $n$-bit increment operator can be implemented with $\mathcal{O}(n)$ $\{CCX, CX, X\}$ gates and $\mathcal{O}(\log n)$ circuit depth.
\end{lemma}
\begin{proof}
    Using Equation~\eqref{eq:inc_recursive}, an incrementer can be implemented as follows:
    
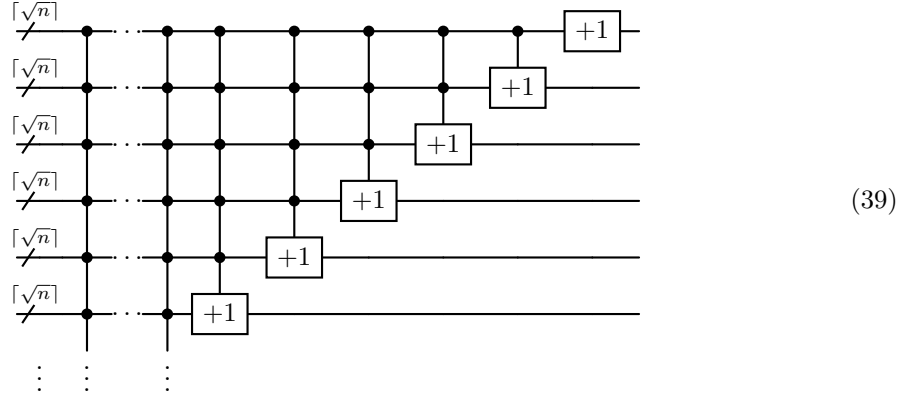
\begin{equation}\label{eq:inc_sqrt_decomposition}
        \begin{quantikz}[row sep={\the\circuitrowsep,between origins}, column sep=\the\circuitcolsep]
            & \qwbundle{\!\!\!\!\!\!\!\lceil\sqrtsign{n}\rceil} & & \ctrl{1} & \dots & \ctrl{1} & \ctrl{1} & \ctrl{1} & \ctrl{1} & \ctrl{1} & \ctrl{1} & \gate{+1} & \\
            & \qwbundle{\!\!\!\!\!\!\!\lceil\sqrtsign{n}\rceil} & & \ctrl{1} & \dots & \ctrl{1} & \ctrl{1} & \ctrl{1} & \ctrl{1} & \ctrl{1} & \gate{+1} & & \\
            & \qwbundle{\!\!\!\!\!\!\!\lceil\sqrtsign{n}\rceil} & & \ctrl{1} & \dots & \ctrl{1} & \ctrl{1} & \ctrl{1} & \ctrl{1} & \gate{+1} & & & \\
            & \qwbundle{\!\!\!\!\!\!\!\lceil\sqrtsign{n}\rceil} & & \ctrl{1} & \dots & \ctrl{1} & \ctrl{1} & \ctrl{1} & \gate{+1} & & & & \\
            & \qwbundle{\!\!\!\!\!\!\!\lceil\sqrtsign{n}\rceil} & & \ctrl{1} & \dots & \ctrl{1} & \ctrl{1} & \gate{+1} & & & & & \\
            & \qwbundle{\!\!\!\!\!\!\!\lceil\sqrtsign{n}\rceil} & & \ctrl{1} & \dots & \ctrl{1} & \gate{+1} & & & & & & \\
            \setwiretype{n} & \push{\vdots} & & \push{\vdots} & & \push{\vdots} & & & & & & & 
        \end{quantikz}
    \end{equation}
    where the last register has size at most $\lceil\sqrt{n}\rceil$.
    We then use the following circuit identity:
    
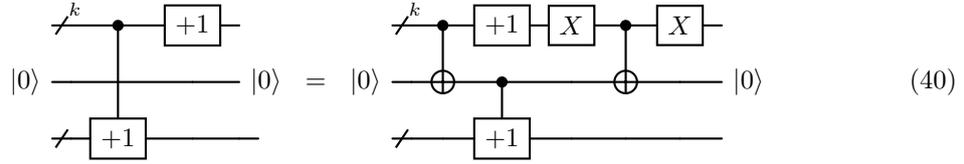
\begin{equation}\label{eq:parallel_inc}
        \begin{quantikz}[row sep={\the\circuitrowsep,between origins}, column sep=\the\circuitcolsep]
            & \qwbundle{k} & \ctrl{2} & \gate{+1} & \\
            \lstick{\ket{0}} & & & & \rstick{\ket{0}} \\
            & \qwbundle{} & \gate{+1} & & & 
        \end{quantikz}
        \;=\;
        \begin{quantikz}[row sep={\the\circuitrowsep,between origins}, column sep=\the\circuitcolsep]
            & \qwbundle{k} & \ctrl{1} & \gate{+1} & \gate{X} & \ctrl{1} & \gate{X} & \\
            \lstick{\ket{0}} & & \targ{} & \ctrl{1} & & \targ{} & & \rstick{\ket{0}} \\
            & \qwbundle{} & & \gate{+1} & & & & 
        \end{quantikz}
    \end{equation}
    which holds because if the top register is in the $\ket{1}^{\otimes k}$ state, the increment operation maps it to the $\ket{0}^{\otimes k}$ state.
    The $X$ gates then switch the register back to the $\ket{1}^{\otimes k}$ state so that the two $C^kX$ gates have the same effect on the ancilla qubit.

    By Equations~\eqref{eq:inc_sqrt_decomposition} and~\eqref{eq:parallel_inc}, we can then obtain the following implementation of an $n$-bit incrementer:
    
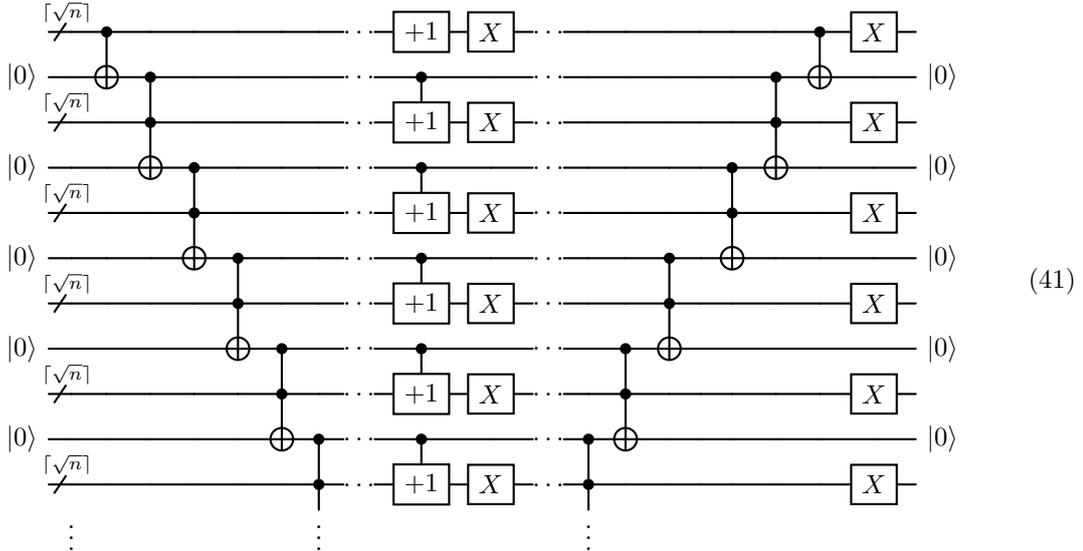
\begin{equation}\label{eq:parallel_inc_sqrt_ancillae}
        \begin{quantikz}[row sep={.6cm,between origins}, column sep=\the\circuitcolsep]
            & \qwbundle{\!\!\!\!\!\!\!\lceil\sqrtsign{n}\rceil} & \ctrl{1} & & & & & & \dots & \gate{+1} & \gate{X} & \dots & & & & & & & \ctrl{1} & \gate{X} & \\
            \lstick{\ket{0}} & & \targ{} & \ctrl{1} & & & & & \dots & \ctrl{1} & & \dots & & & & & & \ctrl{1} & \targ{} & & \rstick{\ket{0}} \\
            & \qwbundle{\!\!\!\!\!\!\!\lceil\sqrtsign{n}\rceil} & & \ctrl{1} & & & & & \dots & \gate{+1} & \gate{X} & \dots & & & & & & \ctrl{1} & & \gate{X} & \\
            \lstick{\ket{0}} & & & \targ{} & \ctrl{1} & & & & \dots & \ctrl{1} & & \dots & & & & & \ctrl{1} & \targ{} & & & \rstick{\ket{0}} \\
            & \qwbundle{\!\!\!\!\!\!\!\lceil\sqrtsign{n}\rceil} & & & \ctrl{1} & & & & \dots & \gate{+1} & \gate{X} & \dots & & & & & \ctrl{1} & & & \gate{X} & \\
            \lstick{\ket{0}} & & & & \targ{} & \ctrl{1} & & & \dots & \ctrl{1} & & \dots & & & \ctrl{1} & & \targ{} & & & & \rstick{\ket{0}} \\
            & \qwbundle{\!\!\!\!\!\!\!\lceil\sqrtsign{n}\rceil} & & & & \ctrl{1} & & & \dots & \gate{+1} & \gate{X} & \dots & & & \ctrl{1} & & & & & \gate{X} & \\
            \lstick{\ket{0}} & & & & & \targ{} & \ctrl{1} & & \dots & \ctrl{1} & & \dots & & \ctrl{1} & \targ{} & & & & & & \rstick{\ket{0}} \\
            & \qwbundle{\!\!\!\!\!\!\!\lceil\sqrtsign{n}\rceil} & & & & & \ctrl{1} & & \dots & \gate{+1} & \gate{X} & \dots & & \ctrl{1} & & & & & & \gate{X} & \\
            \lstick{\ket{0}} & & & & & & \targ{} & \ctrl{1} & \dots & \ctrl{1} & & \dots & \ctrl{1} & \targ{} & & & & & & & \rstick{\ket{0}} \\
            & \qwbundle{\!\!\!\!\!\!\!\lceil\sqrtsign{n}\rceil} & & & & & & \ctrl{1} & \dots & \gate{+1} & \gate{X} & \dots & \ctrl{1} & & & & & & & \gate{X} & \\
            \setwiretype{n} & \push{\vdots} & & & & & & \push{\vdots} & & & & & \push{\vdots} & & & & & & & &
        \end{quantikz}
    \end{equation}
    using $\lceil \sqrt{n} \rceil$ clean ancilla qubits.
    Each controlled increment on a register of size $\lceil\sqrt{n}\rceil$ is activated only when all qubits of the $\lceil\sqrt{n}\rceil$-qubit register above it are in the $\ket{1}$ state.
    That register can therefore serve as a promise register for the controlled increment, as in Figure~\ref{fig:conditionally_clean_ancillae}.
    Since each promise gate uses the register above it as its promise register, we cannot apply all promise gates simultaneously; instead, we proceed in two rounds.
    Moreover, if the increment on the register above has already been performed, there is no need to apply $X$ gates to the promise register, as the increment operator maps the $\ket{1}^{\otimes k}$ state to the $\ket{0}^{\otimes k}$ state.
    Applying this to the circuit of Equation~\eqref{eq:parallel_inc_sqrt_ancillae}, we obtain:
    
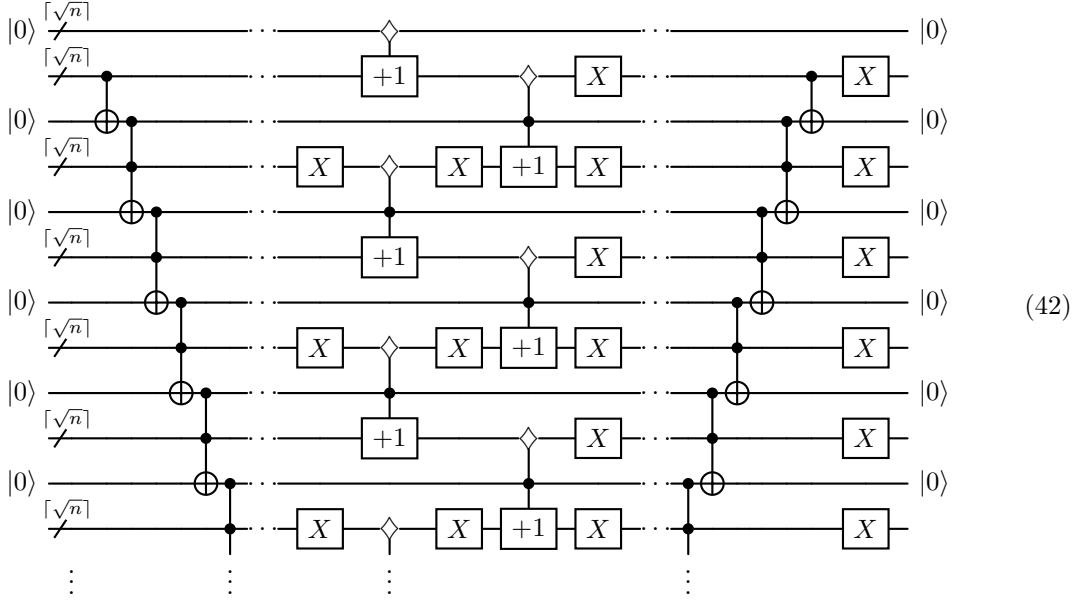
\begin{equation}
        \begin{quantikz}[row sep={.6cm,between origins}, column sep=\the\circuitcolsep]
            \lstick{\ket{0}} & \qwbundle{\!\!\!\!\!\!\!\lceil\sqrtsign{n}\rceil} & &[-0.25cm] &[-0.25cm] &[-0.25cm] &[-0.25cm] &[-0.17cm] &[-0.1cm] \dots & & \push{\diamondsuit}\wire[d][1]{a} & & & & \dots &[-0.1cm] &[-0.17cm] &[-0.25cm] &[-0.25cm] &[-0.25cm] &[-0.25cm] &[-0.25cm] & & \rstick{\ket{0}} \\
            & \qwbundle{\!\!\!\!\!\!\!\lceil\sqrtsign{n}\rceil} & \ctrl{1} & & & & & & \dots & & \gate{+1} & & \push{\diamondsuit}\wire[d][1]{a} & \gate{X} & \dots & & & & & & & \ctrl{1} & \gate{X} & \\
            \lstick{\ket{0}} & & \targ{} & \ctrl{1} & & & & & \dots & & & & \ctrl{1} & & \dots & & & & & & \ctrl{1} & \targ{} & & \rstick{\ket{0}} \\
            & \qwbundle{\!\!\!\!\!\!\!\lceil\sqrtsign{n}\rceil} & & \ctrl{1} & & & & & \dots & \gate{X} & \push{\diamondsuit}\wire[d][1]{a} & \gate{X} & \gate{+1} & \gate{X} & \dots & & & & & & \ctrl{1} & & \gate{X} & \\
            \lstick{\ket{0}} & & & \targ{} & \ctrl{1} & & & & \dots & & \ctrl{1} & & & & \dots & & & & & \ctrl{1} & \targ{} & & & \rstick{\ket{0}} \\
            & \qwbundle{\!\!\!\!\!\!\!\lceil\sqrtsign{n}\rceil} & & & \ctrl{1} & & & & \dots & & \gate{+1} & & \push{\diamondsuit}\wire[d][1]{a} & \gate{X} & \dots & & & & & \ctrl{1} & & & \gate{X} & \\
            \lstick{\ket{0}} & & & & \targ{} & \ctrl{1} & & & \dots & & & & \ctrl{1} & & \dots & & & \ctrl{1} & & \targ{} & & & & \rstick{\ket{0}} \\
            & \qwbundle{\!\!\!\!\!\!\!\lceil\sqrtsign{n}\rceil} & & & & \ctrl{1} & & & \dots & \gate{X} & \push{\diamondsuit}\wire[d][1]{a} & \gate{X} & \gate{+1} & \gate{X} & \dots & & & \ctrl{1} & & & & & \gate{X} & \\
            \lstick{\ket{0}} & & & & & \targ{} & \ctrl{1} & & \dots & & \ctrl{1} & & & & \dots & & \ctrl{1} & \targ{} & & & & & & \rstick{\ket{0}} \\
            & \qwbundle{\!\!\!\!\!\!\!\lceil\sqrtsign{n}\rceil} & & & & & \ctrl{1} & & \dots & & \gate{+1} & & \push{\diamondsuit}\wire[d][1]{a} & \gate{X} & \dots & & \ctrl{1} & & & & & & \gate{X} & \\
            \lstick{\ket{0}} & & & & & & \targ{} & \ctrl{1} & \dots & & & & \ctrl{1} & & \dots & \ctrl{1} & \targ{} & & & & & & & \rstick{\ket{0}} \\
            & \qwbundle{\!\!\!\!\!\!\!\lceil\sqrtsign{n}\rceil} & & & & & & \ctrl{1} & \dots & \gate{X} & \push{\diamondsuit}\wire[d][1]{a} & \gate{X} & \gate{+1} & \gate{X} & \dots & \ctrl{1} & & & & & & & \gate{X} & \\
            \setwiretype{n} & \push{\vdots} & & & & & & \push{\vdots} & & & \push{\vdots} & & & & & \push{\vdots} & & & & &
        \end{quantikz}
    \end{equation}
    where $\lceil \sqrt{n} \rceil$ ancilla qubits were inserted at the top to serve as the promise register for the topmost promise gate.

    We then replace the clean ancilla qubits with a promise register of size $2\lceil\sqrt{n}\rceil$ and add a control to each gate to obtain a controlled strong promise gate whose target unitary is the increment operator.
    By Figure~\ref{fig:controlled_conjugation}, the two ladders of $C^kX$ gates do not need to be controlled.
    Using $\lceil \sqrt{n} \rceil$ qubits from the promise register, these two ladders can be implemented with $\mathcal{O}(n)$ gates and $\mathcal{O}(\log n)$ circuit depth, as stated in Corollary~\ref{cor:ckx_ladder}.
    Finally, using available dirty ancilla qubits from the promise register, we apply the following circuit identity based on Equation~\eqref{eq:toggle_detection_inc}:
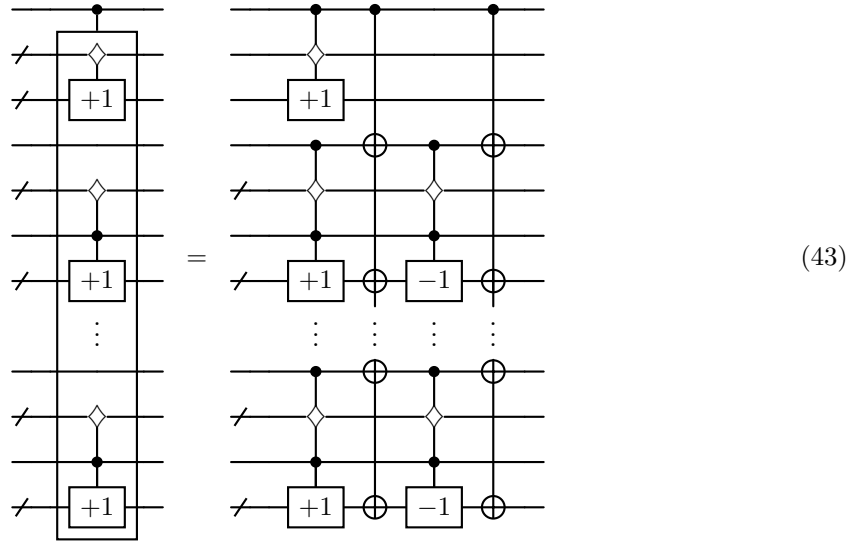
    \begin{equation}
        \begin{quantikz}[row sep={.6cm,between origins}, column sep=\the\circuitcolsep, align equals at=6.5, execute at end picture={\draw [line width=0.8pt] (topcontrol.south) -- (ccxbox.north -| topcontrol.south);}]
            & & & |[alias=topcontrol, fill, circle, minimum size=4pt, inner sep=0pt]|{} & & \\
            & \qwbundle{} & & \push{\diamondsuit}\wire[d][1]{a}\gategroup[11, steps=1, style={name=ccxbox, inner sep=1pt}]{} & & \\
            & \qwbundle{} & & \gate{+1} & & \\
            & & & & & \\
            & \qwbundle{} & & \push{\diamondsuit}\wire[d][1]{a} & & \\
            & & & \ctrl{1} & & \\
            & \qwbundle{} & & \gate{+1} & & \\
            \setwiretype{n} & & & \push{\vdots} & & \\
            & & & & & \\
            & \qwbundle{} & & \push{\diamondsuit}\wire[d][1]{a} & & \\
            & & & \ctrl{1} & & \\
            & \qwbundle{} & & \gate{+1}\wire[u][1]{a} & &
        \end{quantikz}
        \;=\;
        \begin{quantikz}[row sep={.6cm,between origins}, column sep=\the\circuitcolsep, align equals at=6.5]
            & & & \ctrl{1} & \ctrl{7} & & \ctrl{7} & & \\
            & & & \push{\diamondsuit}\wire[d][1]{a} & & & & & \\
            & & & \gate{+1} & & & & & \\
            & & & \ctrl{1} & \targ{} & \ctrl{1} & \targ{} & & \\
            & \qwbundle{} & & \push{\diamondsuit}\wire[d][1]{a} & & \push{\diamondsuit}\wire[d][1]{a} & & & \\
            & & & \ctrl{1} & & \ctrl{1} & & & \\
            & \qwbundle{} & & \gate{+1} & \targ{} & \gate{-1} & \targ{} & & \\
            \setwiretype{n} & & & \push{\vdots} & \push{\vdots} & \push{\vdots} & \push{\vdots} & & \\
            & & & \ctrl{1} & \targ{}\wire[u,shorten >=0.18cm][1]{a} & \ctrl{1} & \targ{}\wire[u,shorten >=0.18cm][1]{a} & & \\
            & \qwbundle{} & & \push{\diamondsuit}\wire[d][1]{a} & & \push{\diamondsuit}\wire[d][1]{a} & & & \\
            & & & \ctrl{1} & & \ctrl{1} & & & \\
            & \qwbundle{} & & \gate{+1}\wire[u][1]{a} & \targ{}\wire[u][3]{a} & \gate{-1}\wire[u][1]{a} & \targ{}\wire[u][3]{a} & & 
        \end{quantikz}
    \end{equation}
    We obtain the circuit presented in Figure~\ref{fig:ctrl_promise_inc_sqrt}.

    We now analyze the cost of the resulting circuit.
    The two ladders of $C^kX$ gates contribute $\mathcal{O}(n)$ gates and $\mathcal{O}(\log n)$ circuit depth, as already discussed.
    The fan-out gates can be implemented with $\mathcal{O}(n)$ $CX$ gates and $\mathcal{O}(\log n)$ circuit depth (Lemma~\ref{lem:fanout}).
    The controlled strong promise gates, whose target unitary is the $\lceil\sqrt{n}\rceil$-bit increment or decrement operator, each require $\mathcal{O}(\sqrt{n})$ gates and $\mathcal{O}(\log n)$ circuit depth (Lemma~\ref{lem:promise_inc_linear_ancillae}).
    Since there are $\mathcal{O}(\sqrt{n})$ such gates arranged in a constant number of layers, their total cost is $\mathcal{O}(n)$ gates and $\mathcal{O}(\log n)$ circuit depth.
    The overall complexity of the circuit presented in Figure~\ref{fig:ctrl_promise_inc_sqrt} is therefore $\mathcal{O}(n)$ $\{CCX, CX, X\}$ gates and $\mathcal{O}(\log n)$ circuit depth.
\end{proof}

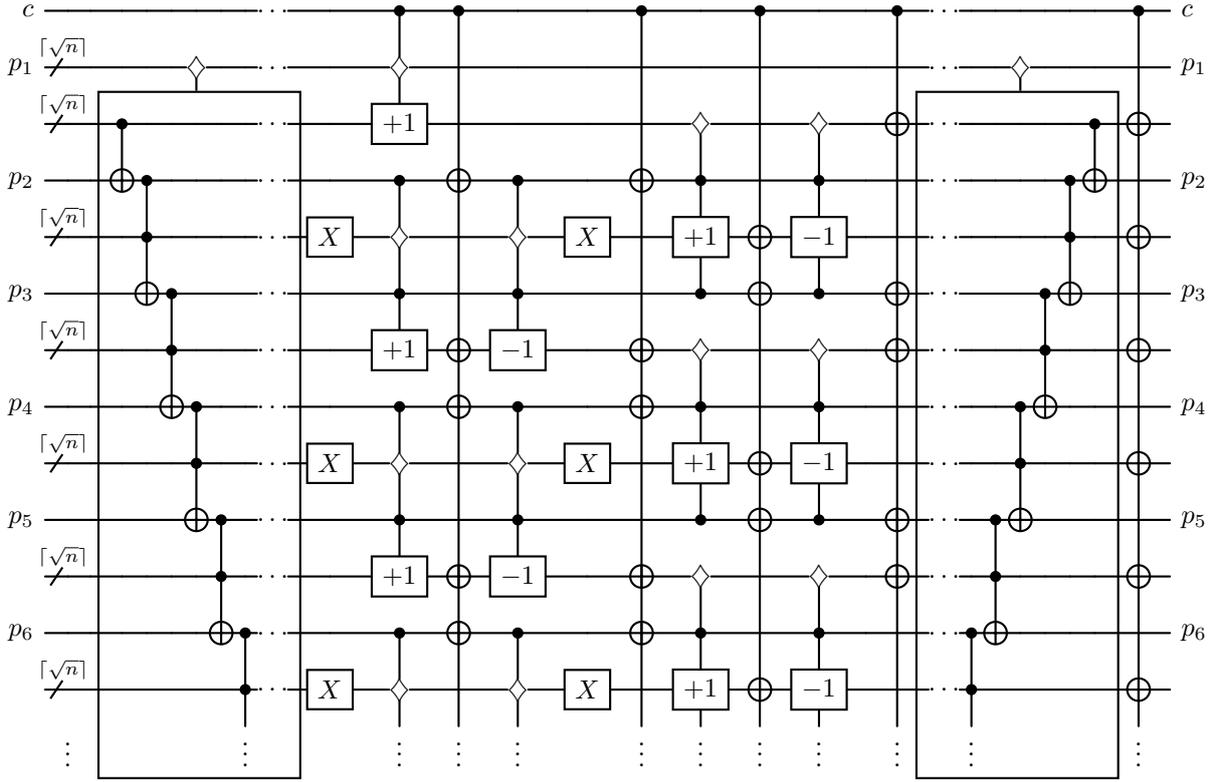
\begin{figure}[t]
    \makebox[\textwidth][c]{%
    \begin{quantikz}[row sep={\the\circuitrowsep,between origins}, column sep=\the\circuitcolsep, execute at end picture={
        \draw [line width=0.8pt] (topdiamond.south) -- (ccxbox.north -| topdiamond.south);
        \draw [line width=0.8pt] (topdiamonddag.south) -- (ccxboxdag.north -| topdiamonddag.south);
    }]
    \lstick{$c$} & & & &[-0.25cm] &[-0.25cm] &[-0.25cm] &[-0.25cm] &[-0.17cm] &[-0.17cm] \dots & & \ctrl{1} & \ctrl{13} & & & \ctrl{13} & & \ctrl{13} & & & \ctrl{13} & \dots &[-0.17cm] &[-0.17cm] &[-0.25cm] &[-0.25cm] &[-0.25cm] &[-0.25cm] & \ctrl{13} & \rstick{$c$} \\
    \lstick{$p_1$} & \qwbundle{\!\!\!\!\!\!\!\lceil\sqrtsign{n}\rceil} & & & & & |[alias=topdiamond]| \push{\diamondsuit} & & & \dots & & \push{\diamondsuit}\wire[d][1]{a} & & & & & & & & & & \dots & & & |[alias=topdiamonddag]| \push{\diamondsuit} & & & & & \rstick{$p_1$} \\
    & \qwbundle{\!\!\!\!\!\!\!\lceil\sqrtsign{n}\rceil} & & \ctrl{1}\gategroup[12, steps=7, style={name=ccxbox, inner sep=1pt}]{} & & & & & & \dots & & \gate{+1} & & & & & \push{\diamondsuit}\wire[d][1]{a} & & \push{\diamondsuit}\wire[d][1]{a} & & \targ{} & \gategroup[12, steps=7, style={name=ccxboxdag, inner sep=1pt}]{}\dots & & & & & & \ctrl{1} & \targ{} & \\
    \lstick{$p_2$} & & & \targ{} & \ctrl{1} & & & & & \dots & & \ctrl{1} & \targ{} & \ctrl{1} & & \targ{} & \ctrl{1} & & \ctrl{1} & & & \dots & & & & & \ctrl{1} & \targ{} & & \rstick{$p_2$} \\
    & \qwbundle{\!\!\!\!\!\!\!\lceil\sqrtsign{n}\rceil} & & & \ctrl{1}  & & & & & \dots & \gate{X} & \push{\diamondsuit}\wire[d][1]{a} & & \push{\diamondsuit}\wire[d][1]{a} & \gate{X} & & \gate{+1} & \targ{} & \gate{-1} & & & \dots & & & & & \ctrl{1} & & \targ{} & \\
    \lstick{$p_3$} & & & & \targ{} & \ctrl{1} & & & & \dots & & \ctrl{1} & & \ctrl{1} & & & \ctrl{-1} & \targ{} & \ctrl{-1} & & \targ{} & \dots & & & & \ctrl{1} & \targ{} & & & \rstick{$p_3$} \\
    & \qwbundle{\!\!\!\!\!\!\!\lceil\sqrtsign{n}\rceil} & & & & \ctrl{1}  & & & & \dots & & \gate{+1} & \targ{} & \gate{-1} & & \targ{} & \push{\diamondsuit}\wire[d][1]{a} & & \push{\diamondsuit}\wire[d][1]{a} & & \targ{} & \dots & & & & \ctrl{1} & & & \targ{} & \\
    \lstick{$p_4$} & & & & & \targ{} & \ctrl{1} & & & \dots & & \ctrl{1} & \targ{} & \ctrl{1} & & \targ{} & \ctrl{1} & & \ctrl{1} & & & \dots & & & \ctrl{1} & \targ{} & & & & \rstick{$p_4$} \\
    & \qwbundle{\!\!\!\!\!\!\!\lceil\sqrtsign{n}\rceil} & & & & & \ctrl{1} & & & \dots & \gate{X} & \push{\diamondsuit}\wire[d][1]{a} & & \push{\diamondsuit}\wire[d][1]{a} & \gate{X} & & \gate{+1} & \targ{} & \gate{-1} & & & \dots & & & \ctrl{1} & & & & \targ{} & \\
    \lstick{$p_5$} & & & & & & \targ{} & \ctrl{1} & & \dots & & \ctrl{1} & & \ctrl{1} & & & \ctrl{-1} & \targ{} & \ctrl{-1} & & \targ{} & \dots & & \ctrl{1} & \targ{} & & & & & \rstick{$p_5$} \\
    & \qwbundle{\!\!\!\!\!\!\!\lceil\sqrtsign{n}\rceil} & & & & & & \ctrl{1} & & \dots & & \gate{+1} & \targ{} & \gate{-1} & & \targ{} & \push{\diamondsuit}\wire[d][1]{a} & & \push{\diamondsuit}\wire[d][1]{a} & & \targ{} & \dots & & \ctrl{1} & & & & & \targ{} & \\
    \lstick{$p_6$} & & & & & & & \targ{} & \ctrl{1} & \dots & & \ctrl{1} & \targ{} & \ctrl{1} & & \targ{} & \ctrl{1} & & \ctrl{1} & & & \dots & \ctrl{1} & \targ{} & & & & & & \rstick{$p_6$} \\
    & \qwbundle{\!\!\!\!\!\!\!\lceil\sqrtsign{n}\rceil} & & & & & & & \ctrl{1} & \dots & \gate{X} & \push{\diamondsuit}\wire[d][1]{a} & & \push{\diamondsuit}\wire[d][1]{a} & \gate{X} & & \gate{+1}\wire[d][1]{a} & \targ{} & \gate{-1}\wire[d][1]{a} & & & \dots & \ctrl{1} & & & & & & \targ{} & \\
    \setwiretype{n} & \push{\vdots} & & & & & & & \push{\vdots} & & & \push{\vdots} & \push{\vdots} & \push{\vdots} & & \push{\vdots} & \push{\vdots} & \push{\vdots} & \push{\vdots} & & \push{\vdots} & & \push{\vdots} & & & & & & \push{\vdots} &
\end{quantikz}
}
\caption{Quantum circuit for a controlled strong promise gate whose target unitary is the increment operator. $c$ denotes the control qubit and $\bs p$ denotes the promise register of size $2\lfloor\sqrt{n}\rfloor$.}\label{fig:ctrl_promise_inc_sqrt}
\end{figure}

We now use Lemma~\ref{lem:promise_inc_sqrt_ancillae} to prove the following theorem.
\begin{theorem}\label{thm:incrementer}
    The $n$-bit increment operator can be implemented over the $\{CCX, CX, X\}$ gate set with $\Theta(n)$ gates and $\Theta(\log n)$ circuit depth, using one dirty ancilla qubit.
\end{theorem}
\begin{proof}
    Based on Equations~\eqref{eq:inc_recursive} and~\eqref{eq:toggle_detection_inc}, and the circuit identity of Figure~\ref{fig:conditionally_clean_ancillae}, we have:
    \begin{equation}
        \begin{quantikz}[row sep={\the\circuitrowsep,between origins}, column sep=\the\circuitcolsep, align equals at=2]
            & \qwbundle{\alpha} & \gate[2]{+1} & \\
            & \qwbundle{\beta} & & \\
            \lstick{\ket{\psi}} & & & \rstick{\ket{\psi}}
        \end{quantikz}
        =
        \begin{quantikz}[row sep={\the\circuitrowsep,between origins}, column sep=\the\circuitcolsep, align equals at=2]
            & \qwbundle{\alpha} & \gate{X} & \push{\diamondsuit}\wire[d][1]{a} & \gate{X} & \ctrl{2} & \gate{X} & \push{\diamondsuit}\wire[d][1]{a} & \gate{X} & \ctrl{2} & \gate{+1} & \\
            & \qwbundle{\beta} & & \gate{+1} & & \targ{} & & \gate{-1} & & \targ{} & & \\
            \lstick{\ket{\psi}} & & & \ctrl{-1} & & \targ{} & & \ctrl{-1} & & \targ{} & & \rstick{\ket{\psi}}
        \end{quantikz}
    \end{equation}
    Let $\alpha = 2\lceil\sqrt{n}\rceil$ and $\beta = n - \alpha$.
    The $\alpha$-controlled fan-out gates can be implemented with a total cost of $\mathcal{O}(n)$ gates and $\mathcal{O}(\log n)$ circuit depth, using Equation~\eqref{eq:controlled_fan_out}.
    By Lemma~\ref{lem:promise_inc_sqrt_ancillae}, the two controlled strong promise gates whose target unitaries are the $\beta$-bit increment and decrement operators can be implemented with $\mathcal{O}(n)$ $\{CCX, CX, X\}$ gates and $\mathcal{O}(\log n)$ circuit depth.
    We then implement the increment operator on the top $\alpha$ qubits recursively.
    The total gate count satisfies
    \begin{equation}
        C(n) = \Theta(n) + C(2\lceil\sqrt{n}\rceil)
    \end{equation}
    and the total circuit depth satisfies
    \begin{equation}
        D(n) = \Theta(\log n) + D(2\lceil\sqrt{n}\rceil)
    \end{equation}
    which yield $C(n) = \Theta(n)$ and $D(n) = \Theta(\log n)$, respectively.
\end{proof}

As a direct consequence of Theorem~\ref{thm:incrementer}, we obtain an efficient construction for the $k$-controlled incrementer.
\begin{corollary}\label{cor:controlled_incrementer}
    The $k$-controlled $n$-bit increment operator can be implemented over the $\{CCX$, $CX$, $X\}$ gate set with $\Theta(k + n)$ gates and $\Theta(\log(kn))$ circuit depth, using one dirty ancilla qubit.
\end{corollary}
\begin{proof}
    The following circuit identity, which is a direct consequence of Equation~\eqref{eq:inc_recursive}, expresses the $k$-controlled $n$-bit increment operator in terms of a $(k+n)$-bit increment and a $k$-bit decrement:
    \begin{equation}
        \begin{quantikz}[row sep={\the\circuitrowsep,between origins}, column sep=\the\circuitcolsep, align equals at=1.5]
            & \qwbundle{k} & \ctrl{1} & \\
            & \qwbundle{n} & \gate{+1} &
        \end{quantikz}
        \;=\;
        \begin{quantikz}[row sep={\the\circuitrowsep,between origins}, column sep=\the\circuitcolsep, align equals at=1.5]
            & \qwbundle{k} & \gate[2]{+1} & \gate{-1} & \\
            & \qwbundle{n} & & &
        \end{quantikz}
    \end{equation}
    The $(k+n)$-bit increment can be implemented with $\Theta(k+n)$ gates and $\Theta(\log(k+n))$ circuit depth using one dirty ancilla qubit (Theorem~\ref{thm:incrementer}).
    The $k$-bit decrement, obtained by conjugating the $k$-bit incrementer with $X$ gates, can be implemented with $\Theta(k)$ gates and $\Theta(\log k)$ circuit depth, using one dirty ancilla qubit taken from the $n$-qubit register.
\end{proof}

\section{Classical--quantum adder and modular multiplication}\label{sec:classical_quantum_adder}
In this section, we show how the results of the previous sections can be combined to construct efficient classical--quantum adders, which in turn can serve as building blocks for modular multiplication circuits used in Shor's factoring algorithm~\cite{Shor_1994}.

\subsection{Classical--quantum adder}
An $n$-bit classical--quantum adder computes the following map:
\begin{equation}
    \ket{\bs x} \mapsto \ket{(\bs x + \bs c) \bmod 2^n},
\end{equation}
where $\bs x \in \{0, 1\}^n$ and $\bs c \in \{0, 1\}^n$ is a classically known constant.

We now show how plugging our optimal constructions for incrementers and comparators into the approach of H\"{a}ner et al.~\cite{Haner_2017} yields improved asymptotic complexities for classical--quantum addition.
\begin{theorem}\label{thm:classical_quantum_adder}
    The $n$-bit classical--quantum adder can be implemented over the $\{CCX, CX, X\}$ gate set with $\Theta(n\log n)$ gates and $\Theta(\log^2 n)$ circuit depth, using one dirty ancilla qubit.
\end{theorem}
\begin{proof}
    The classical--quantum adder can be implemented using the circuit of Figure~\ref{fig:classical_quantum_adder}, following the approach of Häner et al.~\cite{Haner_2017}.
    The carry operator computes the highest bit of the sum $x_L + c_L$.
    This bit is $1$ if and only if $x_L + c_L \geq 2^{\lceil n/2 \rceil}$, so the carry can be computed via the classical--quantum comparison $x_L \geq 2^{\lceil n/2 \rceil} - c_L$.

    By Theorem~\ref{thm:classical_quantum_comparator}, the carry operator can be implemented over the $\{CCX, CX, X\}$ gate set with $\Theta(n)$ gates and $\Theta(\log n)$ circuit depth, using one dirty ancilla qubit.
    The two fan-out gates can each be implemented with $\Theta(n)$ $CX$ gates and $\Theta(\log n)$ circuit depth (Lemma~\ref{lem:fanout}).
    The two controlled incrementers can each be implemented with $\Theta(n)$ gates and $\Theta(\log n)$ circuit depth, using one dirty ancilla qubit taken from the other half of the register (Corollary~\ref{cor:controlled_incrementer}).
    By Theorem~\ref{thm:incrementer}, each increment operator can be implemented with $\Theta(n)$ gates and $\Theta(\log n)$ circuit depth, using one dirty ancilla qubit taken from the other half of the register.

    The total gate count satisfies
    \begin{equation}
        C(n) = \Theta(n) + 2C\left(\frac{n}{2}\right)
    \end{equation}
    and the total circuit depth satisfies
    \begin{equation}
        D(n) = \Theta(\log n) + D\left(\frac{n}{2}\right)
    \end{equation}
    which yield $C(n) = \Theta(n\log n)$ and $D(n) = \Theta(\log^2 n)$, respectively.
\end{proof}

\begin{figure}[t]
    \centering
    \begin{quantikz}[row sep={\the\circuitrowsep,between origins}, column sep=\the\circuitcolsep, align equals at=2]
        \lstick{\ket{x_L}} & \qwbundle{\lceil\frac{n}{2}\rceil} & & & & & \gate[3]{CARRY} & & \gate[3]{CARRY} & & \gate{+c_L} & & \\
        \lstick{\ket{x_H}} & \qwbundle{\lfloor\frac{n}{2}\rfloor} & & & \gate{+1} & \targ{} & & \gate{+1} & & \targ{} & \gate{+c_H} & & \\
        \lstick{\ket{g}} & & & & \ctrl{-1} & \ctrl{-1} & & \ctrl{-1} & & \ctrl{-1} & & & \rstick{\ket{g}}
    \end{quantikz}
    \caption{Circuit for classical--quantum addition~\cite{Haner_2017}, where $\bs c$ is the classical constant added to the quantum register $\bs x$, with $x_H$, $x_L$ and $c_H$, $c_L$ denoting the high- and low-bit halves of $\bs x$ and $\bs c$, respectively.}\label{fig:classical_quantum_adder}
\end{figure}
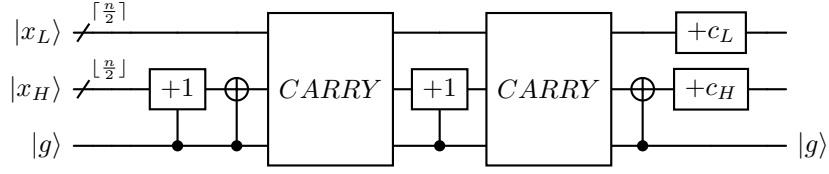

As a direct consequence of Theorem~\ref{thm:classical_quantum_adder}, we obtain an efficient construction for the $k$-controlled classical--quantum adder.
\begin{corollary}\label{cor:controlled_classical_quantum_adder}
    The $k$-controlled $n$-bit classical--quantum adder can be implemented over the $\{CCX$, $CX$, $X\}$ gate set with $\Theta(k + n\log n)$ gates and $\Theta(\log k + \log^2 n)$ circuit depth, using one dirty ancilla qubit.
\end{corollary}
\begin{proof}
    Similary to Equation~\eqref{eq:inc_x_dec}, we have
    \begin{equation}\label{eq:add_x_sub}
        \begin{quantikz}[row sep={\the\circuitrowsep,between origins}, column sep=\the\circuitcolsep, align equals at=1]
            & \qwbundle{n} & \gate{+c} & \gate{X} & \gate{+c} & \gate{X} & 
        \end{quantikz}
        \;=\;
        \begin{quantikz}[row sep={\the\circuitrowsep,between origins}, column sep=\the\circuitcolsep, align equals at=1]
            & \qwbundle{n} & \gate{+c} & \gate{-c} &
        \end{quantikz}
        \;=\;
        \begin{quantikz}[row sep={\the\circuitrowsep,between origins}, column sep=\the\circuitcolsep, align equals at=1]
            & \qwbundle{n} & & &
        \end{quantikz}
    \end{equation}
    This follows from the fact that the one's complement $\overline{a}$ of an two's-complement number $a$ satisfies $\overline{a} = -a - 1 \pmod{2^n}$, so that
    \begin{equation}
        \overline{\overline{a} + c} = -(\overline{a} + c) - 1 = -\overline{a} - c - 1 = a + 1 - c - 1 = a - c \pmod{2^n}.
    \end{equation}

    Let $c = 2c' + r$ where $c' = \lfloor c/2 \rfloor$ and $r = c \bmod 2$.
    Then, based on Equation~\eqref{eq:add_x_sub}, if $r=0$ we have
    \begin{equation}\label{eq:toggle_detection_adder_r_0}
        \begin{quantikz}[row sep={\the\circuitrowsep,between origins}, column sep=\the\circuitcolsep, align equals at=1.5]
            & \qwbundle{k} & \ctrl{1} & \\
            & \qwbundle{n} & \gate{+c} & 
        \end{quantikz}
        \;=\;
        \begin{quantikz}[row sep={\the\circuitrowsep,between origins}, column sep=\the\circuitcolsep, align equals at=1.5]
            & \qwbundle{k} & \ctrl{1} & \\
            & \qwbundle{n} & \gate{+2c'} & 
        \end{quantikz}
        \;=\;
        \begin{quantikz}[row sep={\the\circuitrowsep,between origins}, column sep=\the\circuitcolsep, align equals at=1.5]
            & \qwbundle{k} & & & \ctrl{1} & & \ctrl{1} & \\
            & \qwbundle{n} & & \gate{+c'} & \targ{} & \gate{-c'} & \targ{} & 
        \end{quantikz}
    \end{equation}
    and if $r=1$ we have
    \begin{equation}\label{eq:toggle_detection_adder_r_1}
        \begin{quantikz}[row sep={\the\circuitrowsep,between origins}, column sep=\the\circuitcolsep, align equals at=1.5]
            & \qwbundle{k} & \ctrl{1} & \\
            & \qwbundle{n} & \gate{+c} & 
        \end{quantikz}
        \;=\;
        \begin{quantikz}[row sep={\the\circuitrowsep,between origins}, column sep=\the\circuitcolsep, align equals at=1.5]
            & \qwbundle{k} & \ctrl{1} & \ctrl{1} & \\
            & \qwbundle{n} & \gate{+2c'} & \gate{+1} & 
        \end{quantikz}
        \;=\;
        \begin{quantikz}[row sep={\the\circuitrowsep,between origins}, column sep=\the\circuitcolsep, align equals at=1.5]
            & \qwbundle{k} & & & \ctrl{1} & & \ctrl{1} & \ctrl{1} & \\
            & \qwbundle{n} & & \gate{+c'} & \targ{} & \gate{-c'} & \targ{} & \gate{+1} & 
        \end{quantikz}
    \end{equation}
    The two uncontrolled adders for the constants $+c'$ and $-c'$ can each be implemented with $\Theta(n\log n)$ gates and $\Theta(\log^2 n)$ circuit depth by Theorem~\ref{thm:classical_quantum_adder}, where $-c'$ is obtained by conjugating the $+c'$ circuit with $X^{\otimes n}$ gates at no additional asymptotic cost.
    The two $C^kX$ gates can each be implemented with $\Theta(k)$ gates and $\Theta(\log k)$ circuit depth (Lemma~\ref{lem:optimal_ckx}).
    For $r = 1$, the additional $k$-controlled incrementer can be implemented with $\Theta(k + n)$ gates and $\Theta(\log(kn))$ circuit depth (Corollary~\ref{cor:controlled_incrementer}).
    The total gate count is therefore $\Theta(k + n\log n)$ and the circuit depth is $\Theta(\log k + \log^2 n)$, using one dirty ancilla qubit.
\end{proof}

Historically, and somewhat counterintuitively, classical--quantum addition has been harder to implement efficiently than quantum--quantum addition, due to the lack of available workspace qubits.
Recently, the gap between the two settings has narrowed: Gidney~\cite{Gidney_2025} discovered a classical--quantum adder with a linear number of gates, linear depth, and a constant number of ancilla qubits, matching the costs of the quantum--quantum adder of Cuccaro et al.~\cite{Cuccaro_2004}.
Our results further contribute to closing this gap: the classical--quantum adder presented here achieves the same linearithmic gate count and sublinear circuit depth as Remaud's quantum--quantum adder~\cite{Remaud_2025}, which also uses a minimal number of qubits.

\subsection{Modular multiplication}
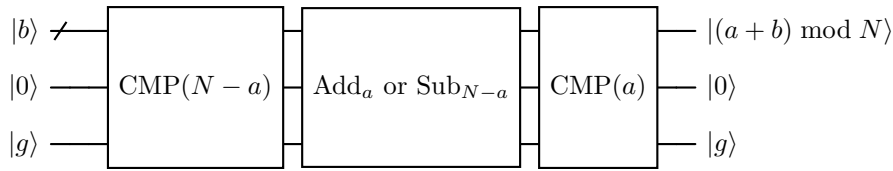
\begin{figure}[t]
    \centering
    \begin{quantikz}[row sep={\the\circuitrowsep,between origins}, column sep=\the\circuitcolsep, align equals at=2]
        \lstick{\ket{b}} & \qwbundle{} & & \gate[3]{\mathrm{CMP}(N-a)} & \gate[3]{\mathrm{Add}_a \text{ or } \mathrm{Sub}_{N-a}} & \gate[3]{\mathrm{CMP}(a)} & & \rstick{\ket{(a+b) \bmod N}} \\
        \lstick{\ket{0}} & & & & & & & \rstick{\ket{0}} \\
        \lstick{\ket{g}} & & & & & & & \rstick{\ket{g}}
    \end{quantikz}
\caption{Construction of a modular adder from a non-modular adder~\cite{Takahashi_2006}.
The first CMP gate computes the comparison $b < N - a$ and stores the result in the clean ancilla qubit: the result is $1$ if and only if adding $a$ to $b$ does not cause an overflow modulo $N$.
The add or subtract gate then adds $a$ to $b$ if the ancilla qubit is in the $\ket{1}$ state (no overflow), or adds $a - N$ if the ancilla qubit is in the $\ket{0}$ state (overflow).
The second CMP gate uncomputes the ancilla qubit by computing $(a + b) \bmod N \geq a$, which equals $1$ if and only if no overflow occurred.}\label{fig:modular_addition}
\end{figure}

\begin{table}[t]
    \centering
    \renewcommand{\arraystretch}{1.2}
    \begin{tabular}{lcccc}
        \toprule
        Reference & Year & Qubits & Depth & Gate Count \\
        \midrule
        Beckman et al.~\cite{Beckman_1996} & 1996 & $5n+1$ & $\mathcal{O}(n^3)$ & $\mathcal{O}(n^3)$ \\
        Vedral et al.~\cite{Vedral_1996} & 1996 & $4n+3$ & $\mathcal{O}(n^3)$ & $\mathcal{O}(n^3)$ \\
        Zalka~\cite{Zalka_1998} & 1998 & $3n+\mathcal{O}(1)$ & $\mathcal{O}(n^3)$ & $\mathcal{O}(n^3)$ \\
        Beauregard~\cite{Beauregard_2003} & 2003 & $2n+3$ & $\mathcal{O}(n^3\log n)$ & $\mathcal{O}(n^3\log^2 n)$ \\
        Takahashi et al.~\cite{Takahashi_2006} & 2006 & $2n+2$ & $\mathcal{O}(n^3\log n)$ & $\mathcal{O}(n^3\log^2 n)$ \\
        Zalka~\cite{Zalka_2006} & 2006 & $1.5n+\mathcal{O}(1)$ & $\mathcal{O}(n^3 \log n)$ & $\mathcal{O}(n^3\log^2 n)$ \\
        Häner et al.~\cite{Haner_2017} & 2016 & $2n+2$ & $\mathcal{O}(n^3)$ & $\mathcal{O}(n^3\log n)$ \\
        Gidney~\cite{Gidney_2018} & 2017 & $2n+1$ & $\mathcal{O}(n^3)$ & $\mathcal{O}(n^3\log n)$ \\
        Gidney et al.~\cite{Gidney_2021} & 2019 & $3n+0.002n\log n$ & $500n^2 + n^2\log n$ & $0.3n^3 + 0.0005n^3 \log n$ \\
        Chevignard et al.~\cite{Chevignard_2025} & 2024 & $n/2 + o(n)$ & $\mathcal{O}(n^3\log n)$ & $\mathcal{O}(n^2\log^4 n)$ \\
        This paper & 2026 & $2n+2$ & $\mathcal{O}(n^2\log^2 n)$ & $\mathcal{O}(n^3\log n)$ \\
        \bottomrule
    \end{tabular}
    \caption{Comparison of space-efficient quantum circuits for Shor's factoring algorithm applied to an $n$-bit RSA integer, in terms of qubit count, circuit depth, and gate count.}\label{tab:shor_circuits}
\end{table}

The classical--quantum adder presented in the previous subsection can be used to implement Shor's factoring algorithm~\cite{Shor_1994}, following the approach of Häner et al.~\cite{Haner_2017}.
Factoring an $n$-bit integer can be done by computing a modular exponentiation, which can be decomposed into $\mathcal{O}(n)$ singly-controlled modular multiplications~\cite{Beauregard_2003}, each controlled by one bit of the exponent register.
Each singly-controlled modular multiplication is implemented via $\mathcal{O}(n)$ singly-controlled modular additions~\cite{Haner_2017}.
In turn, each singly-controlled modular addition is implemented using the circuit of Figure~\ref{fig:modular_addition}, which consists of two singly-controlled comparators and one doubly-controlled adder or subtractor: one control qubit comes from the exponent register, and the other from the output of the comparator, which determines whether to add $a$ or $a - N$.

This yields a total of $\mathcal{O}(n^2)$ singly-controlled comparators and $\mathcal{O}(n^2)$ doubly-controlled adders.
By Corollary~\ref{cor:controlled_classical_quantum_comparator}, a singly-controlled classical--quantum comparator can be implemented with $\mathcal{O}(n)$ gates and $\mathcal{O}(\log n)$ circuit depth, using one dirty ancilla qubit.
By Corollary~\ref{cor:controlled_classical_quantum_adder}, a doubly-controlled classical--quantum adder can be implemented with $\mathcal{O}(n \log n)$ gates and $\mathcal{O}(\log^2 n)$ circuit depth, using one dirty ancilla qubit.
The $\mathcal{O}(n^2)$ comparators each contribute $\mathcal{O}(n)$ gates and $\mathcal{O}(\log n)$ circuit depth, and the $\mathcal{O}(n^2)$ doubly-controlled adders each contribute $\mathcal{O}(n \log n)$ gates and $\mathcal{O}(\log^2 n)$ circuit depth.
The overall implementation therefore has a total gate count of $\mathcal{O}(n^3 \log n)$ and a circuit depth of $\mathcal{O}(n^2 \log^2 n)$.
The total qubit count is $2n + 2$: $n$ qubits for the input register $\ket{x}$, $n$ clean ancilla qubits to compute the modular multiplications, one clean ancilla qubit to store the output of the comparator in Figure~\ref{fig:modular_addition}, and one qubit for the control reused at each step of the semi-classical quantum Fourier transform~\cite{Haner_2017}.

A comparison with other space-efficient implementations of Shor's algorithm is given in Table~\ref{tab:shor_circuits}.

\section{Conclusion}

We presented Clifford+Toffoli quantum circuits for comparators and incrementers that simultaneously match the established lower bounds on gate count and circuit depth while using a provably minimal number of qubits.
Our optimal results are analogous to those of Nie et al.~\cite{Nie_2024}, who showed how to implement the multi-controlled $X$ gate with asymptotically optimal gate count and circuit depth using a minimal number of ancilla qubits.
This raises a natural question: can the family of operators implementable with linear gate count, logarithmic depth, and a minimal number of qubits be further extended?
In particular, this question remains open for quantum addition, for which a gap persists between the known lower bounds and the best-known constructions.

While the proposed circuits are asymptotically optimal in gate count and circuit depth, the constant factors have not been optimized, and numerous avenues remain for improving these constructions, making them more practical, and exploring various resource trade-offs.
These optimizations, as well as detailed resource estimation and cost comparisons with other approaches, are left as future work.

A key tool used in our constructions is the notion of promise gates, which provides a framework for reasoning about conditionally clean ancilla qubits and enabled us to identify new ways of making effective use of them in quantum circuit design.
In particular, we established a general theorem for trading clean ancilla qubits for control qubits with little overhead in gate count and circuit depth.
We expect this technique to find applications beyond the operators studied in this paper.

\section*{Acknowledgements}
This work is dedicated to the memory of Maxime Remaud, with whom I shared many valuable discussions and a deep enthusiasm for this line of research.\\
The author thanks Silas Dilkes and Yuta Kikuchi for reviewing this paper.

\bibliographystyle{quantum}
\bibliography{ref.bib}

\begin{thebibliography}{10}

\bibitem{Douglas_2009}
B.~L. Douglas and J.~B. Wang.
\newblock ``{Efficient quantum circuit implementation of quantum walks}''.
\newblock \href{https://dx.doi.org/10.1103/physreva.79.052335}{Physical Review
  A{\bf 79}}~(2009).

\bibitem{Durr_1999}
Christoph Durr and Peter Hoyer.
\newblock ``{A Quantum Algorithm for Finding the Minimum}''~(1999).
\newblock
  \href{http://arxiv.org/abs/quant-ph/9607014}{arXiv:quant-ph/9607014}.

\bibitem{Hindlycke_2024}
Christoffer Hindlycke and Jan-Åke Larsson.
\newblock ``{Single-qubit rotation algorithm with logarithmic Toffoli count and
  gate depth}''.
\newblock \href{https://dx.doi.org/10.1103/physrevresearch.6.l042027}{Physical
  Review Research{\bf 6}}~(2024).

\bibitem{Shor_1994}
P.W. Shor.
\newblock ``{Algorithms for quantum computation: discrete logarithms and
  factoring}''.
\newblock In Proceedings 35th Annual Symposium on Foundations of Computer
  Science.
\newblock \href{https://dx.doi.org/10.1109/sfcs.1994.365700}{Page 124–134}.
\newblock SFCS-94. IEEE Comput. Soc. Press~(1994).

\bibitem{Haner_2017}
Thomas Haner, Martin Roetteler, and Krysta~M. Svore.
\newblock ``{Factoring using 2n+2 qubits with Toffoli based modular
  multiplication}''.
\newblock \href{https://dx.doi.org/10.26421/qic17.7-8-7}{Quantum Information
  and Computation {\bf 17}, 673–684}~(2017).

\bibitem{Nie_2024}
Junhong Nie, Wei Zi, and Xiaoming Sun.
\newblock ``{Quantum circuit for multi-qubit Toffoli gate with optimal
  resource}''~(2024).
\newblock  \href{http://arxiv.org/abs/2402.05053}{arXiv:2402.05053}.

\bibitem{Claudon_2024}
Baptiste Claudon, Julien Zylberman, César Feniou, Fabrice Debbasch, Alberto
  Peruzzo, and Jean-Philip Piquemal.
\newblock ``{Polylogarithmic-depth controlled-NOT gates without ancilla
  qubits}''.
\newblock \href{https://dx.doi.org/10.1038/s41467-024-50065-x}{Nature
  Communications{\bf 15}}~(2024).

\bibitem{Khattar_2025}
Tanuj Khattar and Craig Gidney.
\newblock ``{Rise of conditionally clean ancillae for efficient quantum circuit
  constructions}''.
\newblock \href{https://dx.doi.org/10.22331/q-2025-05-21-1752}{Quantum {\bf 9},
  1752}~(2025).

\bibitem{Remaud_2025}
Maxime Remaud and Vivien Vandaele.
\newblock ``{Ancilla-Free Quantum Adder with Sublinear Depth}''.
\newblock \href{https://dx.doi.org/10.1007/978-3-031-97063-4_11}{Page
  137–154}.
\newblock Springer Nature Switzerland. ~(2025).

\bibitem{Gidney_incrementer}
Craig Gidney.
\newblock ``{Constructing Large Increment Gates}''.
\newblock
  url:~\url{https://algassert.com/circuits/2015/06/12/Constructing-Large-Increment-Gates.html}.

\bibitem{Cuccaro_2004}
Steven~A. Cuccaro, Thomas~G. Draper, Samuel~A. Kutin, and David~Petrie Moulton.
\newblock ``{A new quantum ripple-carry addition circuit}''~(2004).
\newblock
  \href{http://arxiv.org/abs/quant-ph/0410184}{arXiv:quant-ph/0410184}.

\bibitem{Takahashi_2010}
Yasuhiro Takahashi, Seiichiro Tani, and Noboru Kunihiro.
\newblock ``{Quantum addition circuits and unbounded fan-out}''.
\newblock \href{https://dx.doi.org/10.26421/qic10.9-10-12}{Quantum Information
  and Computation {\bf 10}, 872–890}~(2010).

\bibitem{Gidney_2018}
Craig Gidney.
\newblock ``{Factoring with n+2 clean qubits and n-1 dirty qubits}''~(2018).
\newblock  \href{http://arxiv.org/abs/1706.07884}{arXiv:1706.07884}.

\bibitem{Gidney_2025}
Craig Gidney.
\newblock ``{A Classical-Quantum Adder with Constant Workspace and Linear
  Gates}''~(2025).
\newblock  \href{http://arxiv.org/abs/2507.23079}{arXiv:2507.23079}.

\bibitem{Toffoli_1980}
Tommaso Toffoli.
\newblock ``Reversible computing''.
\newblock \href{https://dx.doi.org/10.1007/3-540-10003-2_104}{Page 632–644}.
\newblock Springer Berlin Heidelberg. ~(1980).

\bibitem{Draper_2000}
Thomas~G. Draper.
\newblock ``{Addition on a Quantum Computer}''~(2000).
\newblock
  \href{http://arxiv.org/abs/quant-ph/0008033}{arXiv:quant-ph/0008033}.

\bibitem{Fang_2006}
M.~Fang, S.~Fenner, F.~Green, S.~Homer, and Y.~Zhang.
\newblock ``{Quantum lower bounds for fanout}''.
\newblock \href{https://dx.doi.org/10.26421/qic6.1-3}{Quantum Information and
  Computation {\bf 6}, 46–57}~(2006).

\bibitem{Shende_2003}
V.V. Shende, A.K. Prasad, I.L. Markov, and J.P. Hayes.
\newblock ``{Synthesis of reversible logic circuits}''.
\newblock \href{https://dx.doi.org/10.1109/tcad.2003.811448}{IEEE Transactions
  on Computer-Aided Design of Integrated Circuits and Systems {\bf 22},
  710–722}~(2003).

\bibitem{Beverland_2020}
Michael Beverland, Earl Campbell, Mark Howard, and Vadym Kliuchnikov.
\newblock ``{Lower bounds on the non-Clifford resources for quantum
  computations}''.
\newblock \href{https://dx.doi.org/10.1088/2058-9565/ab8963}{Quantum Science
  and Technology {\bf 5}, 035009}~(2020).

\bibitem{Gosset_2025}
David Gosset, Robin Kothari, and Chenyi Zhang.
\newblock ``{Multi-qubit Toffoli with exponentially fewer T gates}''~(2025).
\newblock  \href{http://arxiv.org/abs/2510.07223}{arXiv:2510.07223}.

\bibitem{Broadbent_2009}
Anne Broadbent and Elham Kashefi.
\newblock ``{Parallelizing quantum circuits}''.
\newblock \href{https://dx.doi.org/10.1016/j.tcs.2008.12.046}{Theoretical
  Computer Science {\bf 410}, 2489–2510}~(2009).

\bibitem{Draper_2006}
T.G. Draper, S.A. Kutin, E.M. Rains, and K.M. Svore.
\newblock ``{A logarithmic-depth quantum carry-lookahead adder}''.
\newblock \href{https://dx.doi.org/10.26421/qic6.4-5-4}{Quantum Information and
  Computation {\bf 6}, 351–369}~(2006).

\bibitem{remaud2025quantumaddersstructurallink}
Maxime Remaud.
\newblock ``{Quantum adders: on the structural link between the ripple-carry
  and carry-lookahead techniques}''~(2025).
\newblock  \href{http://arxiv.org/abs/2510.00840}{arXiv:2510.00840}.

\bibitem{Goldreich_2006}
Oded Goldreich.
\newblock ``{On Promise Problems: A Survey}''.
\newblock \href{https://dx.doi.org/10.1007/11685654_12}{Page 254–290}.
\newblock Springer Berlin Heidelberg. ~(2006).

\bibitem{Shende_2006}
V.V. Shende, S.S. Bullock, and I.L. Markov.
\newblock ``{Synthesis of quantum-logic circuits}''.
\newblock \href{https://dx.doi.org/10.1109/tcad.2005.855930}{IEEE Transactions
  on Computer-Aided Design of Integrated Circuits and Systems {\bf 25},
  1000–1010}~(2006).

\bibitem{Takahashi_2006}
Y.~Takahashi and N.~Kunihiro.
\newblock ``{A quantum circuit for Shor’s factoring algorithm using 2n+2
  qubits}''.
\newblock \href{https://dx.doi.org/10.26421/qic6.2-4}{Quantum Information and
  Computation {\bf 6}, 184–192}~(2006).

\bibitem{Beckman_1996}
David Beckman, Amalavoyal~N. Chari, Srikrishna Devabhaktuni, and John Preskill.
\newblock ``{Efficient networks for quantum factoring}''.
\newblock \href{https://dx.doi.org/10.1103/physreva.54.1034}{Physical Review A
  {\bf 54}, 1034–1063}~(1996).

\bibitem{Vedral_1996}
Vlatko Vedral, Adriano Barenco, and Artur Ekert.
\newblock ``{Quantum networks for elementary arithmetic operations}''.
\newblock \href{https://dx.doi.org/10.1103/physreva.54.147}{Physical Review A
  {\bf 54}, 147–153}~(1996).

\bibitem{Zalka_1998}
Christof Zalka.
\newblock ``{Fast versions of Shor's quantum factoring algorithm}''~(1998).
\newblock
  \href{http://arxiv.org/abs/quant-ph/9806084}{arXiv:quant-ph/9806084}.

\bibitem{Beauregard_2003}
S.~Beauregard.
\newblock ``{Circuit for Shor’s algorithm using 2n+3 qubits}''.
\newblock \href{https://dx.doi.org/10.26421/qic3.2-8}{Quantum Information and
  Computation {\bf 3}, 175–185}~(2003).

\bibitem{Zalka_2006}
Christof Zalka.
\newblock ``{Shor's algorithm with fewer (pure) qubits}''~(2006).
\newblock
  \href{http://arxiv.org/abs/quant-ph/0601097}{arXiv:quant-ph/0601097}.

\bibitem{Gidney_2021}
Craig Gidney and Martin Ekerå.
\newblock ``{How to factor 2048 bit RSA integers in 8 hours using 20 million
  noisy qubits}''.
\newblock \href{https://dx.doi.org/10.22331/q-2021-04-15-433}{Quantum {\bf 5},
  433}~(2021).

\bibitem{Chevignard_2025}
Clémence Chevignard, Pierre-Alain Fouque, and André Schrottenloher.
\newblock ``{Reducing the Number of Qubits in Quantum Factoring}''.
\newblock \href{https://dx.doi.org/10.1007/978-3-032-01878-6_13}{Page
  384–415}.
\newblock Springer Nature Switzerland. ~(2025).

\end{thebibliography}

\newpage

\appendix
\section{Proofs}
\subsection{Proof of Lemma~\ref{lem:ladder_2_n_ancilla}}\label{app:proof_ladder}
\begin{proof}
    It was shown in~\cite{Remaud_2025} that the $\mathcal{L}_2^{(n-1)}$ operator can be implemented with
    \begin{equation}
        C(n) = 2n - 2 - \left\lfloor \log_2 n \right\rfloor - \left\lfloor \log_2 \frac{2n}{3} \right\rfloor
    \end{equation}
    $C^kX$ gates and a circuit depth of
    \begin{equation}
        D(n) = \left\lfloor \log_2 n \right\rfloor + \left\lfloor \log_2 \frac{2n}{3} \right\rfloor,
    \end{equation}
    without ancilla qubits.
    The recursive algorithm presented in~\cite{Remaud_2025} achieves this by inserting a layer of $C^kX$ gates at the beginning and end of the circuit, then recursing on a smaller $C^kX$ ladder in the middle.
    For example, one iteration of the algorithm for $\mathcal{L}_2^{(8)}$ yields the circuit on the right-hand side of the following identity:
    
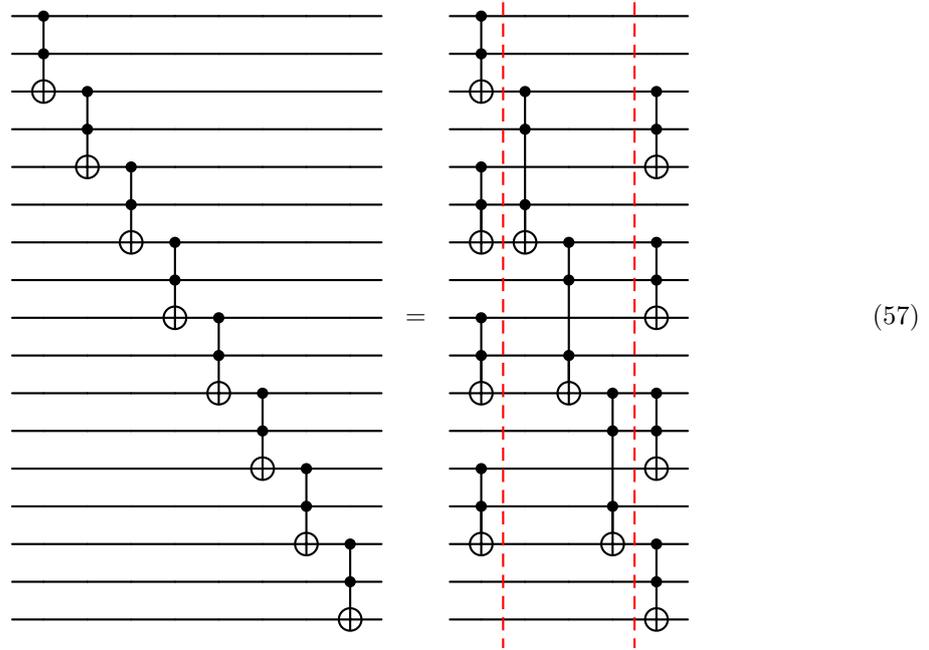
\begin{equation}
        \begin{quantikz}[row sep={.5cm,between origins}, column sep=\the\circuitcolsep, align equals at=9]
            & \ctrl{1} & & & & & & & & \\
            & \ctrl{1} & & & & & & & & \\
            & \targ{} & \ctrl{1} & & & & & & & \\
            & & \ctrl{1} & & & & & & & \\
            & & \targ{} & \ctrl{1} & & & & & & \\
            & & & \ctrl{1} & & & & & & \\
            & & & \targ{} & \ctrl{1} & & & & & \\
            & & & & \ctrl{1} & & & & & \\
            & & & & \targ{} & \ctrl{1} & & & & \\
            & & & & & \ctrl{1} & & & & \\
            & & & & & \targ{} & \ctrl{1} & & & \\
            & & & & & & \ctrl{1} & & & \\
            & & & & & & \targ{} & \ctrl{1} & & \\
            & & & & & & & \ctrl{1} & & \\
            & & & & & & & \targ{} & \ctrl{1} & \\
            & & & & & & & & \ctrl{1} & \\
            & & & & & & & & \targ{} & 
        \end{quantikz}
        \;=\;
        \begin{quantikz}[row sep={.5cm,between origins}, column sep=\the\circuitcolsep, align equals at=9]
             & \ctrl{1} \slice{} & & & \slice{} & & \\
             & \ctrl{1} & & & & & \\
             & \targ{} & \ctrl{1} & & & \ctrl{1} & \\
             & & \ctrl{3} & & & \ctrl{1} & \\
             & \ctrl{2} & & & & \targ{} & \\
             & \ctrl{1} & \ctrl{1} & & & & \\
             & \targ{} & \targ{} & \ctrl{1} & & \ctrl{1} & \\
             & & & \ctrl{3} & & \ctrl{1} & \\
             & \ctrl{2} & & & & \targ{} & \\
             & \ctrl{1} & & \ctrl{1} & & & \\
             & \targ{} & & \targ{} & \ctrl{1} & \ctrl{1} & \\
             & & & & \ctrl{3} & \ctrl{1} & \\
             & \ctrl{2} & & & & \targ{} & \\
             & \ctrl{1} & & & \ctrl{1} & & \\
             & \targ{} & & & \targ{} & \ctrl{1} & \\
             & & & & & \ctrl{1} & \\
             & & & & & \targ{} & 
        \end{quantikz}
    \end{equation}
    Before recursing on the $C^kX$ ladder in the middle slice, we can convert it into a ladder of $CCX$ gates by introducing ancilla qubits, using the following circuit identity:
    
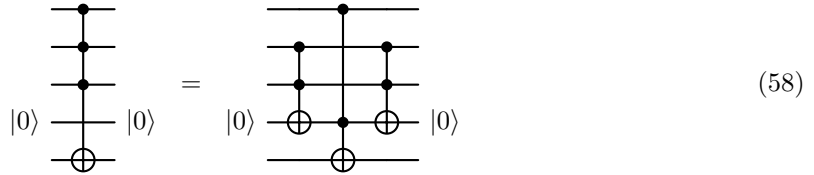
\begin{equation}\label{eq:ckx_to_toffoli_ancilla}
        \begin{quantikz}[row sep={.5cm,between origins}, column sep=\the\circuitcolsep, align equals at=3]
            & \ctrl{1} & \\
            & \ctrl{1} & \\
            & \ctrl{2} & \\
            \lstick{\ket{0}} & & \rstick{\ket{0}} \\
            & \targ{} & 
        \end{quantikz}
        \;=\;
        \begin{quantikz}[row sep={.5cm,between origins}, column sep=\the\circuitcolsep, align equals at=3]
            & & \ctrl{3} & & \\
            & \ctrl{1} & & \ctrl{1} & \\
            & \ctrl{1} & & \ctrl{1} & \\
            \lstick{\ket{0}} & \targ{} & \ctrl{1} & \targ{} & \rstick{\ket{0}} \\
            & & \targ{} & &
        \end{quantikz}
    \end{equation}
    For example, the first iteration of the algorithm for $\mathcal{L}_2^{(8)}$ yields:
    
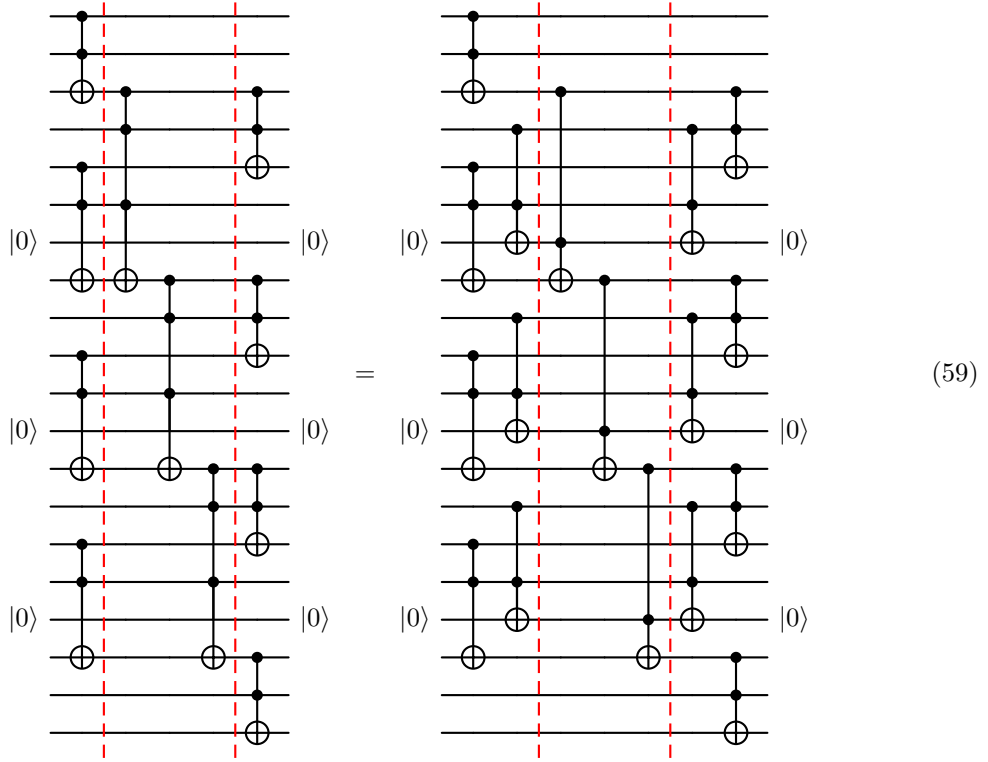
\begin{equation}\label{eq:l2_first_iteration}
        \begin{quantikz}[row sep={.5cm,between origins}, column sep=\the\circuitcolsep, align equals at=10.5]
         & \ctrl{1}\slice{} & & & \slice{} & & \\
         & \ctrl{1} & & & & & \\
         & \targ{} & \ctrl{1} & & & \ctrl{1} & \\
         & & \ctrl{3} & & & \ctrl{1} & \\
         & \ctrl{2} & & & & \targ{} & \\
         & \ctrl{2} & \ctrl{2} & & & & \\
        \lstick{\ket{0}} & & & & & & \rstick{\ket{0}} \\
         & \targ{} & \targ{} & \ctrl{1} & & \ctrl{1} & \\
         & & & \ctrl{3} & & \ctrl{1} & \\
         & \ctrl{2} & & & & \targ{} & \\
         & \ctrl{2} & & \ctrl{2} & & & \\
        \lstick{\ket{0}} & & & & & & \rstick{\ket{0}} \\
         & \targ{} & & \targ{} & \ctrl{1} & \ctrl{1} & \\
         & & & & \ctrl{3} & \ctrl{1} & \\
         & \ctrl{2} & & & & \targ{} & \\
         & \ctrl{2} & & & \ctrl{2} & & \\
        \lstick{\ket{0}} & & & & & & \rstick{\ket{0}} \\
         & \targ{} & & & \targ{} & \ctrl{1} & \\
         & & & & & \ctrl{1} & \\
         & & & & & \targ{} & 
    \end{quantikz}
    \;=\;
    \begin{quantikz}[row sep={.5cm,between origins}, column sep=\the\circuitcolsep, align equals at=10.5]
         & \ctrl{1} & \slice{} & & & \slice{} & & & \\
         & \ctrl{1} & & & & & & & \\
         & \targ{} & & \ctrl{4} & & & & \ctrl{1} & \\
         & & \ctrl{3} & & & & \ctrl{2} & \ctrl{1} & \\
         & \ctrl{1} & & & & & & \targ{} & \\
         & \ctrl{2} & \ctrl{1} & & & & \ctrl{1} & & \\
        \lstick{\ket{0}} & & \targ{} & \ctrl{1} & & & \targ{} & & \rstick{\ket{0}} \\
         & \targ{} & & \targ{} & \ctrl{4} & & & \ctrl{1} & \\
         & & \ctrl{3} & & & & \ctrl{2} & \ctrl{1} & \\
         & \ctrl{1} & & & & & & \targ{} & \\
         & \ctrl{2} & \ctrl{1} & & & & \ctrl{1} & & \\
        \lstick{\ket{0}} & & \targ{} & & \ctrl{1} & & \targ{} & & \rstick{\ket{0}} \\
         & \targ{} & & & \targ{} & \ctrl{4} & & \ctrl{1} & \\
         & & \ctrl{3} & & & & \ctrl{2} & \ctrl{1} & \\
         & \ctrl{1} & & & & & & \targ{} & \\
         & \ctrl{2} & \ctrl{1} & & & & \ctrl{1} & & \\
        \lstick{\ket{0}} & & \targ{} & & & \ctrl{1} & \targ{} & & \rstick{\ket{0}} \\
         & \targ{} & & & & \targ{} & & \ctrl{1} & \\
         & & & & & & & \ctrl{1} & \\
         & & & & & & & \targ{} & 
    \end{quantikz}
    \end{equation}
    The identity in Equation~\eqref{eq:ckx_to_toffoli_ancilla} introduces two additional gates for each gate in the middle $C^kX$ ladder.
    As a result, the total gate count is at most triple that of the circuit constructed by the algorithm of~\cite{Remaud_2025}, minus the $n-1$ gates in the first and last layers.
    This yields a $CCX$ count of at most
    \begin{equation}
        3(C(n) - n + 1) = 3n - 3 - 3\left\lfloor \log_2 n \right\rfloor - 3\left\lfloor \log_2 \frac{2n}{3} \right\rfloor.
    \end{equation}
    At each recursive call, the two slices of gates inserted around the middle ladder have a depth of~$2$ instead of~$1$, which at most doubles the circuit depth:
    \begin{equation}
        2D(n) = 2\left\lfloor \log_2 n \right\rfloor + 2\left\lfloor \log_2 \frac{2n}{3} \right\rfloor.
    \end{equation}
    Finally, at each recursive call, the number of ancilla qubits introduced equals the number of gates in the middle $C^kX$ ladder. Therefore, the total number of ancilla qubits in the final circuit is at most $C(n)$ minus the $n - 1$ gates in the first and last layers:
    \begin{equation}
        C(n) - n + 1 = n - 1 - \left\lfloor \log_2 n \right\rfloor - \left\lfloor \log_2 \frac{2n}{3} \right\rfloor.
    \end{equation}
\end{proof}

\subsection{Proof of Corollary~\ref{cor:ckx_ladder}}\label{app:proof_ckx_ladder}
\begin{proof}
    We apply the same algorithm as the one described in Appendix~\ref{app:proof_ckx_ladder} for the proof of Lemma~\ref{lem:ladder_2_n_ancilla}.
    For example, one iteration of the algorithm for $\mathcal{L}_k^{(8)}$ yields the following circuit:
    
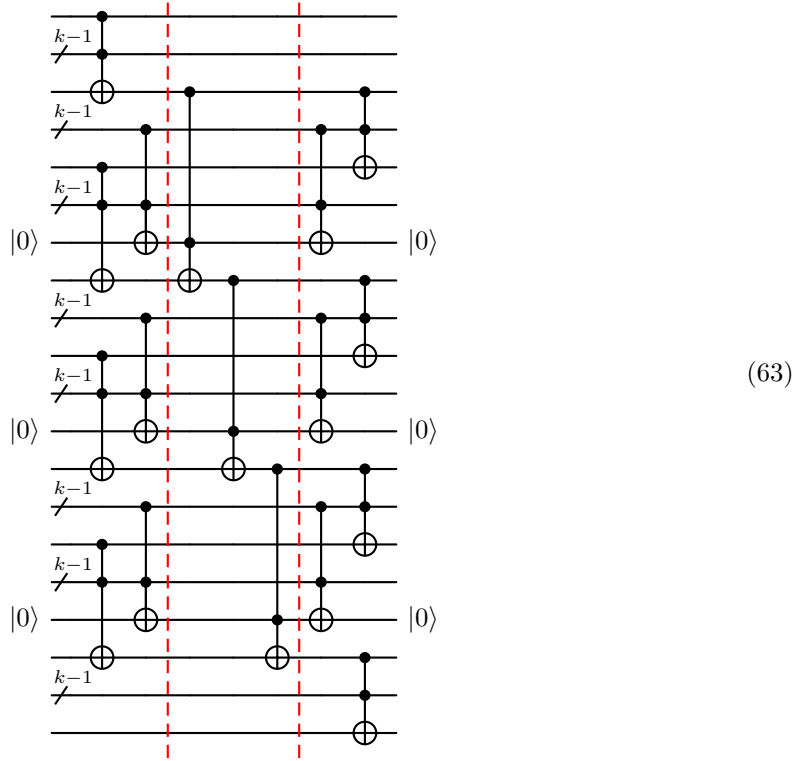
\begin{equation}
        \begin{quantikz}[row sep={.5cm,between origins}, column sep=\the\circuitcolsep, align equals at=10.5]
            & & \ctrl{1} & \slice{} & & & \slice{} & & & \\
            & \qwbundle{\!\!\!\!k-1} & \ctrl{1} & & & & & & & \\
            & & \targ{} & & \ctrl{4} & & & & \ctrl{1} & \\
            & \qwbundle{\!\!\!\!k-1} & & \ctrl{3} & & & & \ctrl{2} & \ctrl{1} & \\
            & & \ctrl{1} & & & & & & \targ{} & \\
            & \qwbundle{\!\!\!\!k-1} & \ctrl{2} & \ctrl{1} & & & & \ctrl{1} & & \\
            \lstick{\ket{0}} & & & \targ{} & \ctrl{1} & & & \targ{} & & \rstick{\ket{0}} \\
            & & \targ{} & & \targ{} & \ctrl{4} & & & \ctrl{1} & \\
            & \qwbundle{\!\!\!\!k-1} & & \ctrl{3} & & & & \ctrl{2} & \ctrl{1} & \\
            & & \ctrl{1} & & & & & & \targ{} & \\
            & \qwbundle{\!\!\!\!k-1} & \ctrl{2} & \ctrl{1} & & & & \ctrl{1} & & \\
            \lstick{\ket{0}} & & & \targ{} & & \ctrl{1} & & \targ{} & & \rstick{\ket{0}} \\
            & & \targ{} & & & \targ{} & \ctrl{4} & & \ctrl{1} & \\
            & \qwbundle{\!\!\!\!k-1} & & \ctrl{3} & & & & \ctrl{2} & \ctrl{1} & \\
            & & \ctrl{1} & & & & & & \targ{} & \\
            & \qwbundle{\!\!\!\!k-1} & \ctrl{2} & \ctrl{1} & & & & \ctrl{1} & & \\
            \lstick{\ket{0}} & & & \targ{} & & & \ctrl{1} & \targ{} & & \rstick{\ket{0}} \\
            & & \targ{} & & & & \targ{} & & \ctrl{1} & \\
            & \qwbundle{\!\!\!\!k-1} & & & & & & & \ctrl{1} & \\
            & & & & & & & & \targ{} & 
    \end{quantikz}
    \end{equation}
    The recursive call on the middle ladder of $CCX$ gates yields a total of $\mathcal{O}(n)$ $CCX$ gates and a circuit depth of $\mathcal{O}(\log n)$, by Lemma~\ref{lem:ladder_2_n_ancilla}.
    The resulting circuit uses the same number of ancilla qubits as in Lemma~\ref{lem:ladder_2_n_ancilla}.
    Finally, all the $C^kX$ gates in the first two and last two layers can be implemented with a total gate count of $\mathcal{O}(kn)$ and a circuit depth of $\mathcal{O}(\log(kn))$, by Lemma~\ref{lem:optimal_ckx}.
\end{proof}

\subsection{Proof of Corollary~\ref{cor:promise_ladder}}\label{app:proof_ladder_promise}
\begin{proof}
    We use the same construction as in the proof of Lemma~\ref{lem:ladder_2_n_ancilla}.
    As shown in that proof, the $\mathcal{L}_2^{(n)}$ operator can be implemented with $\mathcal{O}(n)$ $CCX$ gates and $\mathcal{O}(\log n)$ circuit depth, using $n$ clean ancilla qubits.
    Each ancilla is introduced via the identity in Equation~\eqref{eq:ckx_to_toffoli_ancilla}, in which the ancilla participates in a compute-uncompute pattern: it is computed by a $CCX$ gate, used as a control for another $CCX$ gate, and then uncomputed by a second $CCX$ gate with the same controls.
    This uncomputation does not depend on the ancilla being initialized to $\ket{0}$, so the ancilla register is preserved regardless of its initial state.
    Replacing the clean ancilla qubits with a promise register of size $n$ therefore yields a strong promise gate whose target unitary is the $\mathcal{L}_2^{(n)}$ operator, implemented with $\mathcal{O}(n)$ $\{CCX, CX, X\}$ gates and $\mathcal{O}(\log n)$ circuit depth.

    The generalization to $\mathcal{L}_k^{(n)}$ follows by applying the same argument to the construction in the proof of Corollary~\ref{cor:ckx_ladder}, where the $C^kX$ gates in the outer layers are likewise used in compute-uncompute pairs.
\end{proof}

\end{document}